\documentclass[sigconf]{acmart}
\usepackage{array,color,url}
\usepackage[ruled,linesnumbered]{algorithm2e}
\usepackage{indentfirst}
\usepackage{algpseudocode}
\usepackage{tabu}
\usepackage{bm}
\usepackage{bbm}
\usepackage{float}
\usepackage{epstopdf}
\usepackage{subfigure}
\usepackage{color,xcolor}
\usepackage{comment}
\usepackage{booktabs}
\usepackage{multirow}
\usepackage{url}
\usepackage{verbatimbox}
\usepackage{comment}
\usepackage{diagbox}

\begin{document}
\title{LDP-IDS: Local Differential Privacy for Infinite Data Streams}
\author{Xuebin Ren$^{1}$, Liang Shi$^{1}$, Weiren Yu$^2$, Shusen Yang$^{1}$, Cong Zhao$^3$, Zongben Xu$^{1}$}
\affiliation{ 
	\institution{$^1$Xi'an Jiaotong University, Xi'an, China}
}
\affiliation{ 
	\institution{$^2$University of Warwick, Coventry, United Kingdom}
}
\affiliation{ 
	\institution{$^3$Imperial College London, London, United Kingdom}
}
\email{xuebinren@mail.xjtu.edu.cn, sl1624@stu.xjtu.edu.cn, weiren.yu@warwick.ac.uk} \email{shusenyang@mail.xjtu.edu.cn, c.zhao@imperial.ac.uk, zbxu@mail.xjtu.edu.cn}

\textfloatsep 0.1mm plus 0.1mm \intextsep 0.1mm plus 0.1mm
\begin{abstract}
Streaming data collection is essential to real-time data analytics in various IoTs and mobile device-based systems, which, however, may expose end users' privacy. Local differential privacy (LDP) is a promising solution to privacy-preserving data collection and analysis. However, existing few LDP studies over streams are either applicable to finite streams only or suffering from great utility loss due to simply adopting the budget division method in centralized differential privacy. In this paper, we study this problem by first proposing LDP-IDS, a novel LDP paradigm for infinite streams, and designing baseline approaches under the budget division framework. Particularly, we develop two budget division methods that are adaptive to sparsity changes in streams, with better data utility and communication efficiency. To improve the poor utility in budget division-based LDP, we then propose a population division framework that can not only avoid the high sensitivity of LDP noise to the budget division but also require significantly less communication. Under the population division framework, we also present two data-adaptive methods with theoretical analysis to further improve the estimation accuracy by leveraging the sparsity of data streams. We conduct extensive experiments on synthetic and real-world datasets to evaluate the effectiveness of utility of LDP-IDS. Experimental results demonstrate that, compared to the budget division-based solutions, our population division-based and data-adaptive algorithms for LDP-IDS can significantly reduce the utility loss and communication cost.
\end{abstract}

\maketitle

\section{Introduction}


The proliferation of smart devices and 5G technologies has been greatly boosting data streaming applications, such as event monitoring, log stream analysis, and video querying. These applications often adopt a client/server distributed architecture, where massive users' devices continuously produce data reports and the back-end server conducts real-time analytics over the aggregate data stream. Despite offering valuable information, the continuous collection of streaming data casts severe privacy risks. 
For example, call detail records of mobile phones can be collected for crowd analysis but potentially reveal users' location~\cite{xiao2021longitudinal}. Smart metering data can be mined to improve utility services, which, however, may expose users' daily activities~\cite{hassan2019differential}. 

Differential privacy (DP) has emerged as the de-facto standard for private data analysis with rigorous mathematical proof. DP for data streams has also attracted extensive interests. According to the granularity of privacy protection, these studies can be broadly classified into three categorizes: \textit{event-level}, \textit{user-level} and \textit{$w$-event privacy}. Early researches mainly focus on \textit{event-level privacy} for \textit{finite streams}~\cite{RN35,10.1145/1806689.1806787,Chan2011PrivateAC,bolot2013private,Chen2017PeGaSusDD} and \textit{user-level privacy} for \textit{infinite streams}~\cite{Fan2012RealtimeAM,Fan2013DifferentiallyPM,Fan2014AnAA}. However, the former that hides a single event in streams is insufficient for protecting users' privacy, while the latter that protects a user's occurrence in infinite streams is impractical for most realistic scenarios~\cite{Kellaris2014DifferentiallyPE}.
To break the dilemma, \textit{$w$-event privacy} for infinite streams is proposed~\cite{Kellaris2014DifferentiallyPE}, which aims to guarantee $\epsilon$-DP for any time window consisting of $w$ consecutive time instances (or time interval or timestamps for simplicity). 
Due to meaningful protection and applicability, $w$-event privacy has become the research trend and achieved fruitful results~\cite{Wang2016RescueDPRS,wang2020towards,wang2018privacy}. Nonetheless, these studies are based on \textit{central/centralized differential privacy} (CDP), which relies on a trusted aggregator\footnote{Or server, we use both terms interchangably in the paper.} and are prone to honest-but-curious adversaries. 


Recently, \textit{local differential privacy} (LDP)~\cite{kasiviswanathan2011can, duchi2013local, duchi2014privacy} has demonstrated a great potential in accomplishing analytic tasks without relying on a trusted aggregator. Unlike CDP, LDP has the advantage of guaranteeing massive end users' privacy locally, and thereby has been successfully deployed by many well-known corporations, e.g., Google~\cite{Erlingsson2014RAPPORRA}, Microsoft~\cite{ding2017collecting}, Apple~\cite{AppleDP}, and Uber~\cite{johnson2018towards}. Contemporary studies on LDP mainly focus on static (non-streaming) data analysis, including frequency~\cite{Wang2017LocallyDPFE,Erlingsson2014RAPPORRA,kairouz2016extremal} and mean estimation~\cite{Wang2019CollectingAA,ye2019privkv}. 
For evolving (streaming) data analytics, there are only a few work, including event-level LDP for infinite streams~\cite{joseph2018local,wang2020continuous} and user-level LDP for finite streams~\cite{erlingsson2019amplification,baocgm}.
Under event-level privacy, Joseph \textit{et al.}~\cite{joseph2018local} propose THRESH, which aims at reducing  privacy loss at time slots with no significant population-wide updates. Despite being compatible to infinite streams, event-level LDP cannot protect a user's coarse-grained data.
For user-level LDP, Bao \textit{et al.}~\cite{baocgm} present a correlated Gaussian mechanism CGM via utilizing autocorrelations in streams. However, under the analytic Gaussian mechanism, CGM achieves only approximate LDP (i.e., $(\epsilon,\delta)$-LDP) rather than pure LDP, and limited to finite streams only, meaning that the service has to be restarted periodically for infinite stream scenarios. Table 1 summarizes DP studies on data streams from both aspects of privacy granularity and applicable architecture.
To the best of our knowledge, there is no prior work on $w$-event LDP for infinite streams, which can persistently provide strong and practical protection for indefinitely streaming data collection.

\begin{table}[tb]\caption{Summary of DP research on data streams}
	\begin{tabular}{|c|l|c|c|c|}
		\hline
		\multicolumn{2}{|c|}{}    & \textbf{Event-level}                                                              & \textbf{User-level}                                   & \textbf{$w$-event level}                                                          \\ \hline
		\multicolumn{2}{|c|}{\textbf{CDP}} & \cite{RN35,10.1145/1806689.1806787,Chan2011PrivateAC,bolot2013private,Chen2017PeGaSusDD} & \cite{Fan2012RealtimeAM,Fan2013DifferentiallyPM,Fan2014AnAA} & \cite{Kellaris2014DifferentiallyPE,Wang2016RescueDPRS,wang2020towards,wang2018privacy} \\ \hline
		\multicolumn{2}{|c|}{\textbf{LDP}} & \cite{joseph2018local,wang2020continuous}                                            & \cite{baocgm,erlingsson2019amplification}                    & Our work                                                                        \\ \hline
	\end{tabular}
\end{table}

In this paper, we propose LDP-IDS, a pure $\epsilon$-LDP based paradigm over \textit{infinite streams} under the framework of \textit{$w$-event privacy}. There are three technical challenges for LDP-IDS:


$\bullet$~\textbf{No access to raw data}. In CDP studies~\cite{Kellaris2014DifferentiallyPE}, to reduce overall noise, it is a common technique to mainly update at remarkable timestamps or assign different budget at different timestamps according to the sparsity in raw streams. However, this is difficult in LDP protocols, since the aggregator no longer has access to the raw data streams, which have to be perturbed at the user's end locally.

$\bullet$~\textbf{Utility loss in budge division}. Even if some methods can be formulated to mine the characteristics of raw streams underlying the LDP perturbed streams, the budget division methodology, commonly used in CDP, is not efficient in LDP. This is because the budget division, incurs only quadratic utility degradation in CDP, would incur an approximately exponential utility degradation in LDP~\cite{Wang2017LocallyDPFE,Wang2017LocallyDPHHI}.

$\bullet$~\textbf{High communication cost}. As streaming data generates, massive end users persistently release their perturbed data to the aggregator at each timestamp. That often causes high communication cost for resource-constrained devices. It is desirable to consider the communication efficiency in designing LDP mechanisms. 

To address the above challenges, we propose LDP-IDS, an LDP algorithm for infinite data streams. Specifically, in this paper, we make the following contributions:

\begin{itemize}
	\item We first formulate the problem of infinite streaming data collection with LDP, which aims at realizing statistical analysis over LDP perturbed streams while providing meaningful privacy protection (i.e., $w$-event LDP). 
	\item We construct a unified distortion analysis for streaming data analytics with LDP. Based on the analysis, we present two budge division-based baseline solutions, which dynamic allocates budget according to the non-deterministic sparsity in data streams. Compared to naive methods that evenly divide the privacy budget and enforce the same LDP in each time window of size $w$, the two baselines can effectively improve the utility via leveraging the stream characteristics.
	\item We propose a novel population division-based framework for streaming data collection with LDP, which achieves significantly higher data utility and less communication overheads. By building an analogy between budget division and population division, we design several population division-based solutions with much better utility and communication reduction than the baselines. 
\end{itemize}

We implemented all proposed LDP algorithms for streams and conducted extensive experiment evaluation on both synthetic and real-world datasets. Experimental results show that, compared with the budget division algorithms, population division-based algorithms achieve significant reduction in utility loss and communication overhead, respectively. 
To further demonstrate the real-time event monitoring performance, we evaluated the above-threshold detection performance of our LDP algorithms. Results shows that our methods effectively detect changes in LDP protected streams.

The remainder of this paper is organized as follows. Section~\ref{sec: review} reviews the related work. Section~\ref{sec: preliminaries} introduces the background knowledge. Section~\ref{sec: problem} formulates the problem. Section~\ref{sec: baseline} provides LDP solutions via the budget division methodology. Section~\ref{sec: polo} proposes the novel framework of population division, gives two detailed algorithms with extensive analysis. Section~\ref{sec: experiment} presents an extensive set of experiment results. Finally, Section~\ref{sec: conclusion} concludes the paper.

\section{Related Work}\label{sec: review}

\textbf{Differential Privacy on Streams.}
\textit{Centralized Differential Privacy} (CDP)~\cite{Dwork-405} on streams or temporal data originally focuses on two notions: event-level DP and user-level DP~\cite{10.1145/1806689.1806787}.
The former aims to hide a single event in a stream while the latter tries to hide all  events of a single user. 
For example, Dwork \textit{et al.}~\cite{10.1145/1806689.1806787} initiate the study and propose a binary tree technique based event-level DP algorithm for finite streams. Chan \textit{et al.}~\cite{Chan2011PrivateAC} adopt the same technique to produce partial sum for binary counting in both finite and infinite data streams. Dwork \textit{et al.}~\cite{RN35} further propose a cascade buffers counter with event-level DP to adaptively update the counter according to stream density. Bolot \textit{et al.}~\cite{bolot2013private} introduce the notion of decayed privacy and gradually reduce the privacy cost for past data.  
Then, Chen \textit{et al.}~\cite{Chen2017PeGaSusDD} propose \textsf{PeGaSus} to answer multiple queries while satisifying event-level DP for data streams in a framework of perturb-group-smooth. Nonetheless, event-level DP is usually insufficient for privacy protection while user-level DP can only be achieved on a finite stream. For instance, Fan \textit{et al.}~\cite{Fan2014AnAA} present \textsf{FAST} for realizing user-level DP on finite streams with a framework of sampling-and-filtering.
To address the dilemma, Kellaris \textit{et al.}~\cite{Kellaris2014DifferentiallyPE} propose a notion of $w$-event DP for infinite streams, which ensures $\epsilon$-DP for any time window including $w$ consecutive timestamps. Based on a \textit{sliding window} methodology, they further propose two methods satisfying $w$-event privacy, \textit{Budget Distribution} (\textsf{BD}) and \textit{Budget Absorption} (\textsf{BA}) to effectively allocate privacy budget considering that the statistics on streams may not change significantly in successive timestamps. Moreover, Wang \textit{et al.}~\cite{Wang2016RescueDPRS} propose a multi-dimensional stream release mechanism RescueDP by applying the idea of $w$-event DP to \textsf{FAST} and grouping the dimensions with similar trends. All in all, these CDP solutions cannot be directly applied to LDP settings since the untrusted server can no longer observe the raw data. Besides, LDP noise is generally larger than CDP noise under the same privacy parameter.


\textbf{Local Differential Privacy (LDP).}
In the \textit{local setting} where the server may be untrustworthy, LDP is proposed to perturb data at end users~\cite{kasiviswanathan2011can,duchi2013local,duchi2014privacy}, which has also been extensively studied and applied into many analytic applications~\cite{Wang2017LocallyDPFE, Erlingsson2014RAPPORRA, murakami2019utility,Qin2016HeavyHE,Wang2017LocallyDPHHI,wang2018privtrie, cormode2018marginal,ZhangWLHC18,nguyen2016collecting,ye2019privkv,ren2018textsf,qin2017generating,wang2019answering}.
Nonetheless, most LDP studies concentrate on batch data analysis but seldom consider the stream settings. 
It has been demonstrated in~\cite{tang2017privacy} that, in Apple's LDP implementation, privacy loss accumulated in a short period would be too large to provide meaningful protection. 
To this end, Erlingsson \textit{et al.}~\cite{Erlingsson2014RAPPORRA} introduce a memoiazation mechanism to provide longitudinal (i.e., long-term) LDP guarantee in cases when underlying true value changes in an uncorrelated fashion. Arcolezi~\textit{et al.}~\textit{et al.}~\cite{xiao2021longitudinal} adopt the same memoiazation technique to avoid average attack in logintudinal analysis. Inspired by the binary tree technique in the CDP case~\cite{10.1145/1806689.1806787}, Erlingsson \textit{et al.}~\cite{erlingsson2019amplification} further propose an online protocol that guarantees longitudinal LDP regardless of whether the true value is independent or correlated. However, the construction of binary tree mainly applies to a finite stream, thus limiting its applications to infinite streams.
Joseph \textit{et al.}~\cite{joseph2018local} propose an LDP algorithm \textsf{THRESH} for evolving data, which merely consumes privacy budget at global update timeslots that are selected via users' LDP voting. \textsf{THRESH} relies on the assumption of the number of global updates, and therefore is not applicable to infinite streams either. Besides, Wang \textit{et al.}~\cite{wang2020continuous} extends a threshold-based data release algorithm from CDP to LDP for real-valued streams. Nevertheless, this work focuses on event-level LDP and lacks sufficient protection for infinite streams. Bao \textit{et al.}~\cite{baocgm} propose an ($\epsilon, \delta$) user-level LDP algorithm for finite streaming data collection using the analytic Gaussian mechanism, which focuses on approximate DP and has to renew privacy budget periodically. Wang \textit{et al.}~\cite{wang2020towards} propose a pattern-aware stream data collection mechanism with a metric based $w$-event LDP, which is not comparable to our work. More importantly, all these approaches enforce LDP over streams via a budget division methodology, which causes severe utility loss as reporting with low LDP budget is rather noisy.  

Recently, several LDP studies~\cite{nguyen2016collecting,Wang2017LocallyDPFE,cormode2019answering} have shown that population division is generally better than budget division in LDP, which can be seen as the effect of amplification via subsampling~\cite{balle2018privacy}. Wang \textit{et al.}~\cite{Wang2017LocallyDPFE,Wang2017LocallyDPHHI} point out that one can partition users to answer multiple questions with LDP, which still satisfy the same LDP under the parallel composition but can achieve much higher accuracy than splitting privacy budget and adopting the sequential composition. In particular, they further adopt this idea in the marginal release problem in LDP~\cite{ZhangWLHC18}, where users are divided into different groups to report on different marginals but with the entire privacy budget. The similar user partition idea is actually adopted in \cite{erlingsson2019amplification}, where each user randomly reports with the entire privacy budget on the nodes at a fixed level of the binary tree. Despite these pioneering studies, the idea of population division cannot be directly extended to infinite stream collection and analytics with LDP.


\section{Preliminaries}\label{sec: preliminaries}
In this section, we first present the background about $w$-event privacy and existing methods in centralized setting. Then, we introduce LDP and its building block, frequency oracle.

\subsection{$w$-event Privacy in Centralized Setting}

On data streams, $w$-event privacy can strike a nice balance between event-level privacy for infinite streams and user-level privacy for finite streams. Therefore, we follow such a privacy definition in this paper and present its definition here. 

We first give the notion about \textit{stream prefix} and \textit{neighboring streams}. A stream prefix of an infinite series $S=(D_1,D_2,...)$ at timestamp $t$ is defined as $S_t=(D_1,D_2,...,D_t)$, where $D_i$ is a snapshot of the stream at $i$. Let $w$ be a positive integer, two stream prefixes $S_t, S_t'$ are called \textit{$w$-neighboring}, if for each $S_t[i], S_t'[i]$ such that $i\in [t]$ and $ S_t[i] \neq S_t'[i]$, it holds that $S_t[i], S_t'[i]$ are neighboring; and for each $S_t[i_1],S_t[i_2], S_t'[i_1], S_t'[i_2]$ with $i_1 \le i_2$, $S_t[i_1] \neq S_t'[i_1]$ and $S_t[i_2] \neq S_t'[i_2]$, it holds that $i_2-i_1+1 \le w$.
	
\begin{definition}[$w$-event Privacy~\cite{Kellaris2014DifferentiallyPE}]
	Let $\mathcal{M}$ be a mechanism that takes as input a stream prefix of arbitrary size and $\mathcal{O}$ denotes the set of all possible outputs
	of $\mathcal{M}$. $\mathcal{M}$ satisfies $w$-event $ \epsilon $-DP (or, simply, $w$-event privacy) if for all sets $O \subseteq \mathcal{O}$, all $w$-neighboring stream prefixes $S_t, S_t'$, and all $t$, it holds that 
$\text{Pr}[\mathcal{M}(S_t) \in O] \le e^\epsilon \cdot \text{Pr}[\mathcal{M}(S_t') \in O]$. 
\end{definition}
A mechanism satisfies $w$-event privacy can provide $\epsilon$-DP guarantee in any sliding window of size $w$. Or, for any mechanism with $w$-event privacy, $\epsilon$ can be viewed as the total available privacy budget in any sliding window of size $w$~\cite{Kellaris2014DifferentiallyPE}.

\subsection{Existing Methods with $w$-event CDP}\label{sub: existing methods}
 By properly allocating with different portions of the total budget $\epsilon$, a mechanism composed of a series sub mechanisms over the timestamps can satisfy $w$-event privacy~\cite{Kellaris2014DifferentiallyPE}.
A naive method is to evenly apply $\epsilon/w$-DP histogram release mechanism at every timestamp. Unfortunately, with the increase of $w$, the allocated budget becomes much small, which causes large perturbation noise at each timestamp.
Another simple method is to release an $\epsilon$-DP fresh histogram at one timestamp while other timestamps in a window is directly approximated with this result. However, the fixed sampling strategy cannot accurately follow the update patterns in the dynamic stream, thus leading to large errors.

\textsf{BD} (budget distribution) and \textsf{BA} (budget absorption) are benchmark adaptive methods for infinite stream release with $w$-event CDP\cite{Kellaris2014DifferentiallyPE}. 
Both \textsf{BD} and \textsf{BA} can be summarized into three components: \textit{private dissimilarity calculation}, \textit{private strategy determination}, and \textit{privacy budget allocation}. In \textit{private dissimilarity calculation}, a dissimilarity $dis$ between the current $\mathbf{c}_t$ and the last update $\mathbf{c}_l$ is computed and perturbed with some fixed \textit{dissimilarity budget} $\epsilon_{t,1}$. In \textit{private strategy determination}, some \textit{publication budget} $\epsilon_{t,2}$ is assigned (How to assign is designed in \textit{privacy budget allocation}) for potential publication of noisy statistic, which can derive a potential publication error $err$. Then, $dis$ and $err$ is compared to decide the private strategy for current release. If $err<dis$, publish with perturbation (i.e., $\mathbf{r}_t=\mathbf{c}_t+\langle Lap(1/\epsilon_{t,2}) \rangle^d$); otherwise, approximate by the previous release (i.e., $\mathbf{r}_t=\mathbf{c}_l$). In above process, $\epsilon_{t,1}$ is fixed for each timestamp, but $\epsilon_{t,2}$ is assigned based on different rules in \textsf{BD} and \textsf{BA}. In \textsf{BD}, $\epsilon_{t,2}$ is distributed in an exponentially decaying way to the timestamps where a publication is chosen, and reuses the budget spent in timestamps out of the current sliding window. While in \textsf{BA}, $\epsilon_{t,2}$ is uniformly assigned first and then unused budget is absorbed at timestamps where approximation is chosen.


\subsection{Local Differential Privacy (LDP)}

In the LDP paradigm, $\mathcal{M}$ is a randomized mechanism that takes each distributed user's input $v$ and outputs a perturbed value before sending to the central aggregator, which collects the perturbed data and reconstructs the aggregated statistics. 
	
\begin{definition}[{Local Differential Privacy}]
A mechanism $\mathcal{M}$ satisfies $\epsilon$-local differential privacy (i.e., $\epsilon$-LDP), if and only if, for any input $v$ and $v'$ in domain $Dom(\mathcal{M})$, we have
\[{\forall}O \subseteq \mathcal{O}, \text{Pr}[\mathcal{M}(v) \in O] \le e^\epsilon \cdot \text{Pr}[\mathcal{M}(v') \in O],\]	
\noindent where $\mathcal{O}$ is the set of all possible outputs
of $\mathcal{M}$.
\label{LDP}
\end{definition}


Despite receiving data from individuals, LDP ensures the central aggregator cannot infer the input with high confidence.
As a DP variant, LDP inherits the properties of CDP, including sequential/parallel composition and post-processing theorems~\cite{Dwork2014TheAF}\cite{McSherry2009PrivacyIQ}.

\subsection{Frequency Oracle under LDP}
LDP data analyses are commonly built on some frequency oracle (FO) protocols, which enable frequency estimation of any value $v$ in a domain $\Omega=\{\omega_1, \omega_2,\ldots, \omega_d\}$ of size $d=|\Omega|$ under LDP.  
A common FO protocol is \textit{Generalized Randomized Response (GRR)}.
The idea of GRR method is that with a private data $v \in \Omega$, each user sends the true value to the central aggregator with probability $p$, and randomly sends a value in the candidate set $\Omega\setminus\{v\}$ with probability $1-p$. A GRR-based LDP mechanism $\mathcal{M}$ with the domain $Dom(\mathcal{M})=\Omega$ is defined as follows.
\begin{align}
	{\forall} \overline{v} \in \Omega,  \text{Pr}[\mathcal{M}(v)=\overline{v}]=
	\begin{cases}
		p=\frac{e^{\epsilon}}{e^{\epsilon}+d-1},& \text{if }  \overline{v}=v\\
		q=\frac{1}{e^{\epsilon}+d-1},& \text{if } \overline{v} \neq v
	\end{cases}
\end{align}


GRR-based FO works as follows. The aggregator aims to calculate the \textit{frequency} of each distinct item $v$ or $\omega_k \in \Omega$, denoted as $\mathbf{c}[k]$. It first counts the frequency of $\omega_k$ in perturbed data, which is denoted as $\mathbf{c}'[k]$. Then, assuming $n$ is the number of participant users, the estimated frequency ${\overline{\mathbf{c}}}[k]$ of $\omega_k$ through GRR protocol can be obtained as ${\overline{\mathbf{c}}_\text{GRR}} [k]=(\mathbf{c}'[k]/n-q)/({p-q}).$ We use FO($\overline{D}_t, \epsilon$) to denote this FO process.
It is shown in \cite{Wang2017LocallyDPFE} that this is an unbiased estimation of the true frequency, with the variance
\begin{align}\label{var}
\text{Var}[{\overline{\mathbf{c}}_\text{GRR}} [k];\epsilon,n] = \frac{d-2+e^\epsilon}{n\cdot (e^\epsilon-1)^2}+\frac{f_k\cdot(d-2)}{n \cdot(e^\epsilon-1)}
\end{align}
where $f_v$ is the frequency of $v$ and there is $\sum\nolimits_{k = 1}^d f_k=1$. Considering $f_v$ is often small, the above variance is also simply approximated as $\frac{d-2+e^\epsilon}{n\cdot (e^\epsilon-1)^2}$~\cite{ZhangWLHC18}.
Although there are also other FOs (e.g., OUE)~\cite{ZhangWLHC18}, the estimation variance of these FOs can all be seen as a function of parameter $\epsilon$ and population $n$. 

For simplicity, without specifying the FO used, we use $V(\epsilon,n)$ to represent the estimation variance from $n$ users with budget $\epsilon$.

\section{Problem Definition}\label{sec: problem}

In this subsection, we propose LDP-IDS, a novel LDP paradigm for infinite streams under the framework of $w$-event privacy. 


\begin{figure}[tbp]
	\includegraphics[width=0.45\textwidth]{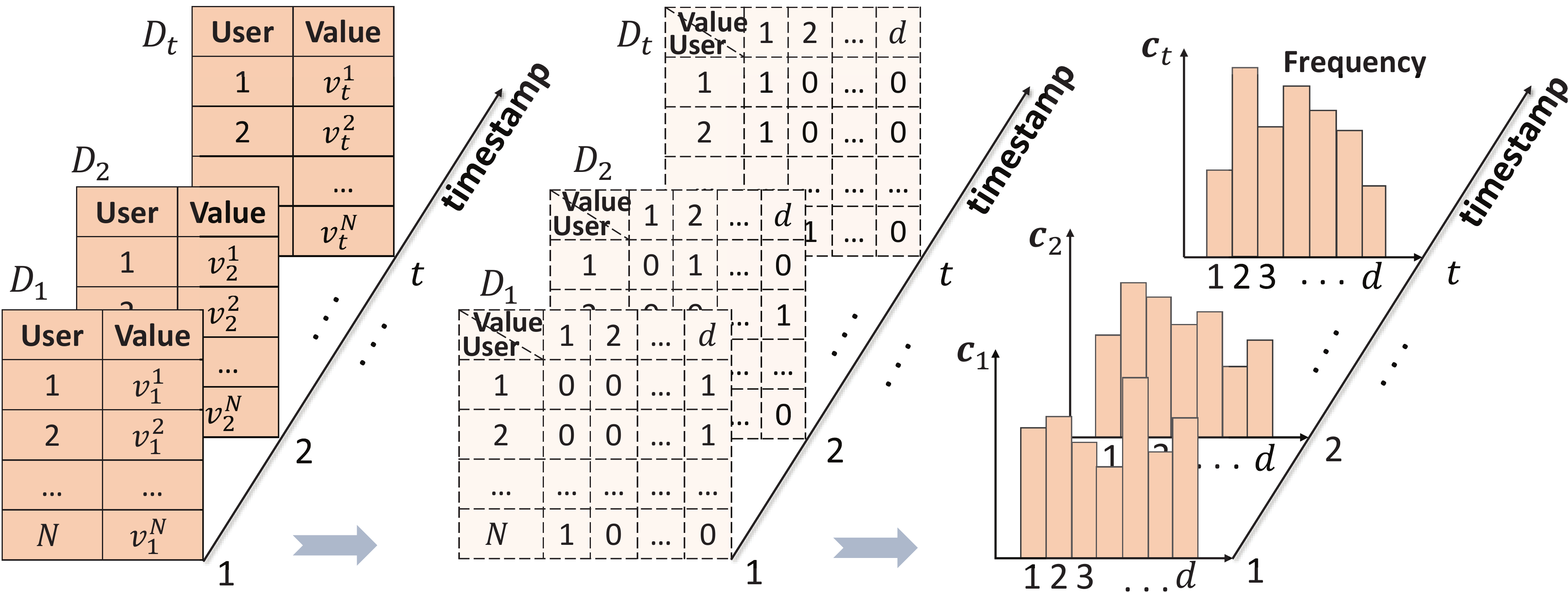}
	\vspace{-5mm}\caption{An illustration of streaming data release}\label{fig: streamcollection}
\end{figure}

We assume that there is a distributed system consisting of a central server and $N$ distributed users $\{1,2,\ldots, N\}$ that continuously report the value of a data item (e.g., user's location in an area, or units of certain measurement) at discrete timestamps. For example, an area is divided into $d$ disjoint regions and the server aims to illustrate the time-evolving population density map from users' location reports. Let $v_t^j$ represent the report of user $j$ at timestamp $t$ and $v_t^j$ come from a domain $\Omega$ with the carnality of $d$. Then, each user has an infinite data stream $V^j=(v_1^j, v_2^j, \ldots)$. 
Meanwhile, at every timestamp $t$, the central server receives all users' reports $\{v_t^1,\ldots,v_t^N\}$, which can be transformed into a binary database $D_t$ with with $d$ columns and $n$ rows. 
As shown in Fig.~\ref{fig: streamcollection}, the server aims to release the statistical histogram $\mathbf{c}_t=\langle \mathbf{c}_t[1],\mathbf{c}_t[2],\ldots,\mathbf{c}_t[d] \rangle$ over all $n$ users' data continuously at timestamp $t$. For example, $\mathbf{c}_t[k]=\frac{1}{n}\sum_j \mathbbm{1}_{\{k|v_t^j=\omega_k\}}(k)$ is the frequency~\footnote{Other aggregate analyses, such as count and mean estimation, can be applicable, as the query type is orthogonal to the streaming data setting.} of every unique value in $\Omega$ and $\mathbbm{1}_{X}(k)$ is an indicator function that equals to $1$ if $k \in X$, and $0$ otherwise. However, direct data collection or release would compromise individual's privacy.  
Specifically, we consider that the central server is not trustworthy and assumed to be \textit{honest-but-curious}. Instead of directly reporting $v_t^i$, each user would choose to send a perturbed value $\overline{v}_t^j$ with LDP or report nothing at each timestamp $t$. 
Therefore, our goal is to design an LDP solution that helps the server to collect data and release an estimated histogram $\mathbf{r}_t=\langle \mathbf{r}_t[1],\mathbf{r}_t[2],\ldots,\mathbf{r}_t[d] \rangle$ at each timestamp $t$ where $\mathbf{r}_t[k]$ denotes the estimated frequency for each value in the domain.


Considering the infinity of streaming data, users also wish to adopt a meaningful privacy paradigm similar to $w$-event privacy in the centralized setting. We naturally extend the definition of $w$-event privacy to the local setting. Before that, we first define the notion of $w$-neighboring in the local setting as follows.

\begin{definition}[$w$-neighboring] Let  $V_t$ and $V'_t$ denote two stream prefixes defined on the same domain $\Omega^t$.
	Let $w$ be a positive integer. $V_t$ and $V'_t$ are $w$-neighboring, if for each $V_t[i_1], V_t[i_2],V'_t[i_1] ,V'_t[i_2]$ with $i_1 \leq i_2$, $V_t[i_1] \neq V'_t[i_1]$ and $V_t[i_2] \neq  V'_t[i_2]$, it holds that $i_2-i_1+1 \leq w$.
\end{definition}

That is to say, if two stream prefixes are $w$-neighboring, then their elements are \textit{the same} while all their \textit{same} elements consist of a window of up to $w$ timestamps. This is slightly different from the definition in the central setting.

\begin{definition}\label{def: w-event LDP}[$w$-event LDP]
	Let $\mathcal{M}$ be a mechanism that takes as input stream prefix $V_t=(v_1,v_2,\ldots,v_t)$ consisting of a single user's arbitrary number of consecutive input value $v_t$. Also let $ \mathcal{O}$ be the set of all possible outputs of $\mathcal{M}$. We say that $\mathcal{M}$ satisfies $w$-event $\epsilon$-LDP (i.e., $w$-event LDP) if for any $w$-neighboring stream prefixes $V_t, V'_t$, and all $t$, it holds that
	\[{\forall} O \subseteq  \mathcal{O} , \text{Pr}[\mathcal{M}(V_t) \in O] \le e^\epsilon \text{Pr}[\mathcal{M}(V'_t) \in O].\]	
\end{definition}
In other words, a $w$-event LDP mechanism will provide each user $\epsilon$-LDP for any sliding window of size $w$.

\section{Budget Division-based Methods}\label{sec: baseline}
In this section, we first present the budget division framework for streaming data collection with LDP. Then, based on this framework, we introduce our LDP methods for the problem defined.

\subsection{Budget Division Framework for LDP}
Inspired by the studies in the centralized setting, the following theorem can be derived for designing LDP mechanisms.

\begin{theorem}\label{thm-w-event LDP}
		Let $\mathcal{M}$ be a mechanism that takes as input stream prefix $V_t$ consisting of a single user's arbitrary number of consecutive input value $v_t$, i.e., $V_t[i]=v_i$, and outputs a transcript $o=(o_1,o_2,...,o_t) \in Range(\mathcal{M})$. Suppose that we can decompose $\mathcal{M}$ into $t$ mechanisms $\mathcal{M}_1,\mathcal{M}_2,...,\mathcal{M}_t,$ such that $\mathcal{M}_i(v_i)=o_i$, each $\mathcal{M}_i$ generates independent randomness and achieves $\epsilon_i$-LDP. Then, $\mathcal{M}$ satisfies $w$-event $\epsilon$-LDP if for any user and any timestamp $i \in [t]$, there is $(\sum_{\tau=i-w+1}^{i} \epsilon_{\tau}) \le \epsilon$.
\end{theorem}

\begin{proof}
See Appendix~\ref{appendice: w-event LDP}.
\end{proof}

This theorem enables a $w$-event LDP mechanism to view $\epsilon$ as the total privacy budget in any sliding window of size $w$, and appropriately allocate portions of it across the timestamps, as shown in Fig.~\ref{fig: BDframework}. According to the above theorem, some straightforward approaches can be summarized to solve the problem defined in Section~\ref{sec: problem}, based on different allocation methods of the LDP budget.

\begin{figure}[tbp]
	\includegraphics[width=0.3\textwidth]{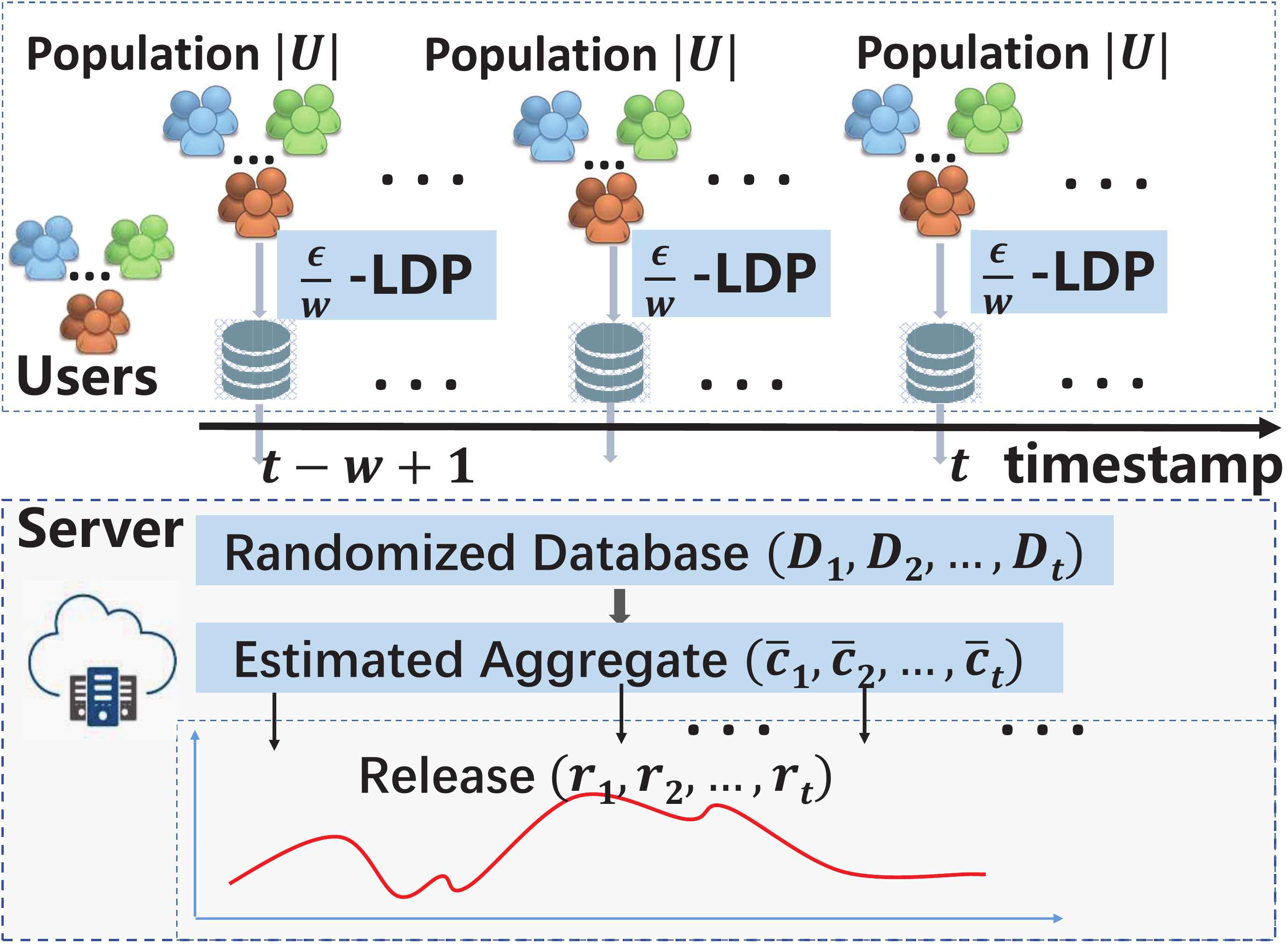}
\vspace{-2mm}	\caption{Illustration of budget division framework}\label{fig: BDframework}
\end{figure}

\subsection{Baseline $w$-event LDP Methods}
\subsubsection{{\textbf{LDP Budget Uniform Method} (\textsf{LBU})}}\label{subsub: LBU}
One straightforward approach is to uniformly assign the LDP budget $\epsilon$ to all $w$ timestamps in the sliding windows. 
At each timestamp, each user reports the perturbed value with an FO using the fixed budget $\epsilon/w$ for satisfying $w$-event LDP.
Recall that $V(\epsilon,n)$ represents the LDP estimation variance from $n$ users with privacy budget $\epsilon$, without specifying the FO.
Since $\mathbf{r}_t$ is an unbiased estimate of $\mathbf{c}_t$, the mean square error (MSE) between the true stream prefix $C_t=(\mathbf{c}_1,\mathbf{c}_2,\ldots,\mathbf{c}_t)$ and the released stream prefix $R_t=(\mathbf{r}_1,\mathbf{r}_2,\ldots,\mathbf{r}_t)$, is indeed the estimation variance of $\mathbf{r}_t$ i.e., $\text{MSE}_\textsf{LBU}=\text{Var}[\mathbf{r}_t;{\epsilon}/{w},N]=V({\epsilon}/{w},N)$.
If $w$ is large, privacy budget allocated at each timestamp is very small, leading to a large noise scale. 

	
\subsubsection{{\textbf{LDP Sampling Method}}(\textsf{LSP})}\label{subsub: sample} 
Each user invests the entire budget $\epsilon$ on a single (sampling) timestamp within the window, while saving budget for the next $w-1$ timestamps via approximation. 
At the last sampling timestamp $l$, the MSE of \textsf{LSP} equals to estimation variance $V(\epsilon,N)$. For non-sampling timestamps, it equals to the sum of the variance of last release at the sampling timestamp (i.e., $\text{Var}[\mathbf{r}_{l}]$~\footnote{For simplicity, we use $\text{Var}[\mathbf{r}_{l}]$ to denote the average variance over $d$ dimensions of vector$\mathbf{r}_{l}$, or $\text{Var}[\mathbf{r}_{l}]=\frac{1}{d}\sum\nolimits_{k=1}^d \text{Var}(\mathbf{r}_{l}[k])$.}), and the variance of true statisics at the current timestamp $t$ from that at the last sampling timestamp $l$ (i.e., $(\mathbf{c}_{t}-\mathbf{c}_{l})^2$, which is data dependent).
Therefore, the MSE of \textsf{LSP} in a window of size $w$ can be calculated as $\text{MSE}_\textsf{LSP}=V(\epsilon,N)+\frac{1}{w}\sum\nolimits_{k=1}^{w-1}(\mathbf{c}_t-\mathbf{c}_{t+k})^2$.
An implicit assumption motives this method is that $\mathbf{c}_{t}$ (or $D_t$) in the stream does not fluctuate too much. Therefore, for streams with few changes, \textsf{LSP} may work better by saving up privacy budgets; otherwise, the estimation error on those skipped timestamps may become excessively large. 

Considering the general non-deterministic sparsity in data streams, both \textsf{LBU} and \textsf{LSP} can not achieve better utility in general cases. 

\subsection{Adaptive Budget Division Methods}\label{sec: EpsilonBDBA}
In this subsection, we propose two adaptive methods by constructing a unified distortion analysis under LDP.

\textsf{BD/BA} in the centralized setting~\cite{Kellaris2014DifferentiallyPE} inspire us that higher utility can be achieved by adaptively allocating privacy budget in data streams. As summarized in Section~\ref{sub: existing methods}, \textsf{BD/BA} compares the dissimilarity $dis$ in aggregate statistics $\mathbf{c}_{t}$ with the potential publication error $err$ at each time to adaptively choose between publication and approximation. 
However, in the local setting, since the central server cannot observe individuals' reports or directly obtain the true $\mathbf{c}_{t}$, the design of such LDP solutions is challenging. In particular, it is infeasible to accomplish the private dissimilarity calculation or data publication by adding noise over the true statistics, but we need to use FO protocols to do so. However, with FO protocols, it remains unclear how to model the \textit{dissimilarity $dis$} and \textit{publication error $err$} under LDP for empirically optimal strategy determination. 

\subsubsection{\textbf{Private dissimilarity estimation}}\label{subsub: pde} To address the above challenges, we first redefine the dissimilarity measure $dis^*$ as the square error between the true statistics $\mathbf{c}_t$ of current timestamp and the previous release $\mathbf{r}_l$, i.e.,
\begin{align}\label{eq: dis definition}
	\setlength{\abovedisplayskip}{3pt}
	\setlength{\belowdisplayskip}{3pt}
    dis^*=\frac{1}{d}\sum_{k=1}^d(\mathbf{c}_t[k]-\mathbf{r}_l[k])^2
\end{align}
Then, in $\mathcal{M}_{t,1}$, we aim to obtain the dissimilarity $dis^*$ privately, i.e., from users' LDP perturbed data using the dissimilarity budget $\epsilon_{t,1}$. 

\begin{theorem}\label{theorem: dissimilarity}
Let $\overline{\mathbf{c}}_{t,1}$ denote the unbiased estimate of $\mathbf{c}_t$ from an $\epsilon$-LDP frequency count over the perturbed data in $\mathcal{M}_{t,1}$.
Then, the following dissimilarity measure
\begin{align}\label{eq: estimation}	
dis=\frac{1}{d} \sum_{k=1}^d  (\overline{\mathbf{c}}_{t,1}[k]-\mathbf{r}_l[k])^2 -\frac{1}{d} \sum_{k=1}^d \text{Var}(\overline{\mathbf{c}}_{t,1}[k]).
\end{align}
is $\epsilon$-LDP and an unbiased estimation of $dis^*$ in Eq. (\ref{eq: dis definition}).
\end{theorem}

\begin{proof}
	See Appendix~\ref{appendice: dis calculation}.
\end{proof}

Therefore, the dissimilarity can be calculated from $\overline{\mathbf{c}}_{t,1}$ while satisfying $\epsilon_{t,1}$-LDP. In the left term $\frac{1}{d} \sum_{k=1}^d (\overline{\mathbf{c}}_{t,1}[k]-\mathbf{r}_l[k])^2$ of Eq~(\ref{eq: estimation}), $\overline{\mathbf{c}}_{t,1}$ is obtained from FO while $\mathbf{r}_l$ is publicly known. The right term $\frac{1}{d} \sum_{k=1}^d \text{Var}(\overline{\mathbf{c}}_{t,1}[k])$ denoted as $V(\epsilon_{t,1},N)$, can be calculated based on the population $N$ and the dissimilarity budget $\epsilon_{t,1}$.

\subsubsection{\textbf{Private strategy determination}}
The key to choose the strategy of approximation or publication is to compare the dissimilarity (i.e., the potential \textit{approximation error}) with the potential \textit{publication error}. Considering that LDP protocols (e.g., GRR) are different from CDP mechanism, the LDP-based publication error should also be re-formulated.

Here, in the LDP setting, considering that $dis$ in Eq (\ref{eq: dis definition}) is an $L_2$ distance measure, we propose to use Mean Square Error (MSE) to measure the publication error, denoted as $err$.
Suppose $\overline{\mathbf{c}}_{t,2}$ as the histogram estimated via FO protocol (e.g., GRR), the estimation error can be measured as
\begin{align}
    err=\frac{1}{d} \sum_{k=1}^d(\overline{\mathbf{c}}_{t,2}[k]-\mathbf{c}_t[k])^2
\end{align}

Since $\overline{\mathbf{c}}_{t,2}$ is an unbiased estimation of $\mathbf{c}_t$, i.e., $\mathbb{E}({\overline{\mathbf{c}}_{t,2}})=\mathbf{c}_t$, 
\begin{align}\label{eq: mse}
		\setlength{\abovedisplayskip}{3pt}
	\setlength{\belowdisplayskip}{3pt}
    err=\frac{1}{d} \sum_{k=1}^d\textup{Var}(\overline{\mathbf{c}}_{t,2}[k])
\end{align}
which denotes as $V(\epsilon_{t,2},N)$ and can be calculated from the user population $N$ and the \textit{publication budget} $\epsilon_{t,2}$.
And, taking GRR as the LDP FO, it can be written as
\begin{align}
	err=\frac{1}{d} \sum_{k=1}^d\textup{Var}(\overline{\mathbf{c}}_{t,2}[k])=\frac{d-2+e^{\epsilon_{t,2}}}{N(e^{\epsilon_{t,2}}-1)^2}+\frac{d-2}{N(e^{\epsilon_{t,2}}-1)}. \nonumber
\end{align}
It is worth noting that, $err$ is independent of $f_v$ in Eq.~(\ref{var}).

Based on above formulations, an empirically optimal strategy at current timestamp $t$ can be determined as follows.
\begin{itemize}
    \item If $dis<err$, the approximation strategy is chosen. For example, the server can directly publish the last released value without consuming the publication budget $\epsilon_{t,2}$.
    \item Otherwise, the perturbation strategy is chosen. Each user reports value via a LDP FO using the publication budget $\epsilon_{t,2}$ to the server, who releases a freshly estimated statistics.
\end{itemize}

\subsubsection{\textbf{Privacy budget allocation}}\label{subsub: BDBA}

From the high level, 
we evenly divide the entire budget in a time window, $\epsilon$, for two components: private dissimilarity estimation and private strategy determination. That is to say, the entire dissimilarity budget and publication budget, in a time window is $\sum_{i=t-w+1}^{t} \epsilon_{i,1}=\sum_{i=t-w+1}^{t} \epsilon_{i,2}=\epsilon/2$. In the private dissimilarity estimation, the dissimilarity budget is divided evenly to each timestamp in the time window, i.e., $\epsilon_{i,1}=2\epsilon/w$. 
However, we aim to invest the publication budget economically to the timestamps, which leads to two different methods, \textbf{LDP budget distribution} (\textsf{LBD}) and \textbf{LDP budget absorption} (\textsf{LBA}).

In \textbf{\textsf{LBD}}, the publication budget is distributed in an exponentially decreasing way to the timestamps where a publication to occur. 
Algorithm~\ref{Alg BD} gives the details. For each time window, the entire budget $\epsilon$ is evenly divided into $\epsilon/2$ as dissimilarity budget and $\epsilon/2$ as publication budget, respectively. 
In sub mechanism $\mathcal{M}_{t,1}$, $\epsilon/2$ dissimilarity budget is uniformly distributed to each timestamp (Line 3). Then, all users apply an FO with the budget to report their data, which can be used to estimate a dissimilarity error $dis$ (Lines 4-6). In sub mechanism $\mathcal{M}_{t,2}$, the remaining publication budget at current time is calculated first (Line 7). Then, half of it is pre-assigned as the potential publication budget, which is used to estimate the potential publication error $err$ (Lines 8-9). By comparing $dis$ and $err$, a strategy will be chosen between publication and approximation. If publication is chosen, the potential publication budget is truly used to get a fresh publication $\overline{\mathbf{c}}_{t,2}$ (Lines 11-13). Otherwise, the last release is published as an approximation. At this time, the potential publication budget is not truly used and thus reset as $0$ (Line 15). 

\begin{algorithm}[htbp]\footnotesize
	\SetAlgoLined
	\caption{LDP Budget Distribution (\textsf{LBD})}\label{Alg BD}
	\KwIn{Total privacy budget $\epsilon$, window size $w$}
	\KwOut{Released statistics $R_t=(\mathbf{r}_1,\mathbf{r}_2,\ldots,\mathbf{r}_t, \ldots)$}
	\BlankLine
	Initialize $\mathbf{r}_0=\langle 0, \ldots, 0 \rangle^d$; \\
	\For{each timestamp $t$}
	{
		\tcp{Sub Mechanism $\mathcal{M}_{t,1}$:}
		Set dissimilarity budget $\epsilon_{t,1}=\epsilon/(2w)$; \\
		$\overline{D}_{t,1}\leftarrow$ All Users report via an FO with privacy budget $\epsilon_{t,1}$;\\
		Estimate $\overline{\mathbf{c}}_{t,1}\leftarrow \text{FO}(\overline{D}_{t,1},\epsilon_{t,1})$ ; \\
		Calculate $dis=\frac{1}{d}\sum\nolimits_{k=1}^d(\overline{\mathbf{c}}_{t,1}[k]-\mathbf{r}_{t-1}[k])^2-\frac{1}{d}\sum\nolimits_{k=1}^d\text{Var}(\mathbf{\overline{c}}_{t,1}[k])$;\\[5pt]
		
		\tcp{Sub Mechanism $\mathcal{M}_{t,2}$:}
		Calculate remaining publication budget $\epsilon_{rm}=\epsilon/2- \sum\nolimits_{i=t-w+1}^{t-1}\epsilon_{i,2}$;  \\
		Set potential publication budget $\epsilon_{t,2}=\epsilon_{rm}/2$; \\
		Calculate potential publication error $err$ by Eq.~(\ref{var});  \\
		\eIf {$dis > err$}
		{\tcp{Publication Strategy}
			$\overline{D}_{t,2}\leftarrow$ All Users report via an FO with budget $\epsilon_{t,2}$;\\
			Estimate $\overline{\mathbf{c}}_{t,2}\leftarrow \text{FO}(\overline{D}_{t,2},\epsilon_{t,2})$; \\
			\Return $\mathbf{r}_t=\overline{\mathbf{c}}_{t,2}$; } 
		{\tcp{Approximation Strategy} \Return $\mathbf{r}_t=\mathbf{r}_{t-1}$; set $\epsilon_{t,2}=0$.}
	}
\end{algorithm}

In \textbf{\textsf{LBA}}, the publication budget is uniformly allocated budget at all timestamps then the unused budget is absorbed at the timestamps where publication is chosen. 
Algorithm~\ref{Alg BA} gives the details. Similarly, in each time window of size $w$, the entire budget $\epsilon$ is evenly divided into $\epsilon/2$ as dissimilarity budget and $\epsilon/2$ as publication budget, respectively. In sub mechanism $\mathcal{M}_{t,1}$, the process is identical to that of \textsf{LBD}. In sub mechanism $\mathcal{M}_{t,2}$, the number of timestamps to be nullified $t_N$ is calculated based on the used publication budget at the publication timestamp $l$, and then skipped with approximation (Lines 4-6). Then, the number of timestamps to be absorbed, and potential publication budget can be calculated to derive the potential publication error $err$ (Lines 8-10). In the next, by comparing $err$ and $dis$, an empirically optimal strategy is chosen between publication and approximation (Lines 11-16). 

\begin{algorithm}[htbp]\footnotesize
	\SetAlgoLined
	\caption{LDP Budget Absorption (\textsf{LBA})}\label{Alg BA}
	\KwIn{Total privacy budget $\epsilon$, window size $w$}
	\KwOut{Released statistics $R_t=(\mathbf{r}_1,\mathbf{r}_2,\ldots,\mathbf{r}_t, \ldots)$}
	\BlankLine
	Initialize $\mathbf{r}_0=\langle 0, \ldots, 0 \rangle^d$, last publication timestamp $l=0$, and $\epsilon_{l,2}=0$; \\
	\For{each timestamp $t$}
	{
		\tcp{Sub Mechanism $\mathcal{M}_{t,1}$:}
		Same as Lines 3-6 in \textbf{Algorithm~\ref{Alg BD}} \\[5pt]
		
		\tcp{Sub Mechanism $\mathcal{M}_{t,2}$:}
		Calculate timestamps to be nullified $t_N= \frac{\epsilon_{l,2}}{ \epsilon/(2w)} -1$;\\
		\eIf{$t-l \leq t_N$}
		{
			\Return $\mathbf{r}_t=\mathbf{r}_{t-1}$;}
		{
			Calculate timestamps can be absorbed $t_A=t-(l+t_N)$;\\
			Set potential publication budget $\epsilon_{t,2}=\epsilon/(2w) \cdot $min$(t_A,w)$;\\
			Calculate potential publication error $err$ by Eq.~(\ref{var}); \\
			\eIf {$dis > err$}
			{\tcp{Perturbation Strategy}
				$\overline{D}_{t,1}\leftarrow$ All Users report via an FO with budget $\epsilon_{t,2}$;\\
				Estimate $\overline{\mathbf{c}}_{t,2}\leftarrow \text{FO}(\overline{D}_{t,2},\epsilon_{t,2})$; \\
				\Return $\mathbf{r}_t=\overline{\mathbf{c}}_{t,2}$, set $l=t$; }
			{\tcp{Approximation Strategy} \Return $\mathbf{r}_t=\mathbf{r}_{t-1}$; set $\epsilon_{t,2}=0$.}
		}
	}
\end{algorithm}



\subsection{Analysis}
\subsubsection{Privacy Analysis}\label{sec: budget-division privacy analysis}

Both \textsf{LBD} and \textsf{LBA} satisfy $w$-event LDP. 
\begin{theorem}\label{theorem: LBD privacy}
	\textsf{LBD} and \textsf{LBA} satisfy $w$-event LDP for each user.
\end{theorem}
\begin{proof}
	See Appendix~\ref{appendice: privacy proof BD/BA}.
\end{proof}
\subsubsection{Utility Analysis}\label{sec: budget-division analysis}
For simplicity, in both \textsf{LBD} and \textsf{LBA}, we assume there are $m<w$ publications occur at the timestamps $p_1,p_2,...,p_m$ in the window of size $w$. Besides, no budget is recycled from past timestamps outside the window, and each publication approximates the same number of skipped/nullified publications.
Similar to the analysis of \textsf{LSP}, at any timestamp $t$, if publication occurs, then the MSE of the release $\mathbf{r}_t$ is $\text{MSE}_\text{pub}=\text{Var}[\mathbf{r}_t]$; if approximation is chosen, its MSE equals to the sum of the variance of last release at timestamp $l$ (i.e., $\text{Var}[\mathbf{r}_{l}]$), and the variance of the true statisics at the current timestamp $t$ from that at timestamp $l$ (i.e., $(\mathbf{c}_{t}-\mathbf{c}_{l})^2$), i.e., $\text{MSE}_\text{apr}=\text{Var}[\mathbf{r}_{l}]+(\mathbf{c}_{t}-\mathbf{c}_{l})^2$. Then we express MSE in a whole time window as follows
\begin{align}\label{eq: BD/BA-MSE}
\text{MSE}_{LBD/LBA} =&\dfrac{1}{w} \left[ \dfrac{w}{m}\sum\limits_{i = 1}^m\text{Var}[\mathbf{r}_{p_i}]
+ \sum_{i=1}^{m} \sum_{t=p_i}^{p_{i+1}-1} (\mathbf{c}_{t}-\mathbf{c}_{p_i})^2 \right]
\end{align}
where the second term in the bracket solely dependes on the underlying data and shows the data-dependent characteristics of \textsf{LBD} and \textsf{LBA}. In the following, we analyze the left term in the bracket.

In \textsf{LBD}, since the budget is distributed to the $m$ publications in an exponentially decreasing way, the budget sequence of $\epsilon_{t,2}$ is then $\epsilon/4, \epsilon/8, \ldots, \epsilon/2^{m+1}$.  There is
\begin{align}
\sum\limits_{i = 1}^m\text{Var}_\textsf{LBD}[\mathbf{r}_{p_i}]=\sum\limits_{i = 1}^m V(\epsilon/2^{i+1}),N)<m \cdot V(\epsilon/2^{m+1},N)
\end{align}
where $V(\epsilon,n)$ denotes the estimation variance of an FO from $n$ users' LDP data using budget $\epsilon$.
As we can see, with the increase of $m$, the error of \textsf{LBD} would increase dramatically.  

In \textsf{LBA}, due to $m$ publications in the assumption, there are $w-m$ approximations. Since each publication approximates the same number of skipped/nullified publications, there are $\frac{w-m}{2\cdot m}$ skipped ( whose budgets are absorbed) and $\frac{w-m}{2\cdot m}$ nullified publications in average. Then, each publication receives those skipped budget $\frac{(\frac{w-m}{2\cdot m}+1)\cdot \epsilon}{2\cdot w}=\frac{(w+m)\cdot \epsilon}{4\cdot w \cdot m}$ and incurs MSE of $V(\frac{w+m}{4\cdot w \cdot m}\cdot \epsilon,N)$.
\begin{align}\label{eq: BAerror}
\sum\limits_{i = 1}^m\text{Var}_\textsf{LBA}[\mathbf{r}_{p_i}]=m\cdot V(\frac{w+m}{4\cdot w \cdot m}\cdot \epsilon,N)
\end{align}
Compared to \textsf{LBD}, \textsf{LBA}'s error increases with $m$ more mildly.

%

\subsubsection{Communication Analysis}\label{subsub: communication budget}
Given an FO, the LDP perturbed data in different methods has the same packet size in each communication. So, the communication cost can be simply measured by the average communication times of each user per timestamp, or \textit{communication frequency per user} (CFPU).  

In both \textsf{LBD/LBA}, each user have to communication twice with the server at publication timestamps and only once at approximation timestamps. Therefore, when there are $m$ publications in a window of size $w$, the average CFPU is $(2m+w-m)/{w}=1+m/w$.

\section{Population Division-based Methods}\label{sec: polo}
In this section, we first present a novel idea of population division for LDP over data streams. Then, we propose to transform the above budget division-based methods to the population division-based methods for better utility and communication efficiency.

\begin{figure}[tbp]
	\includegraphics[width=0.3\textwidth]{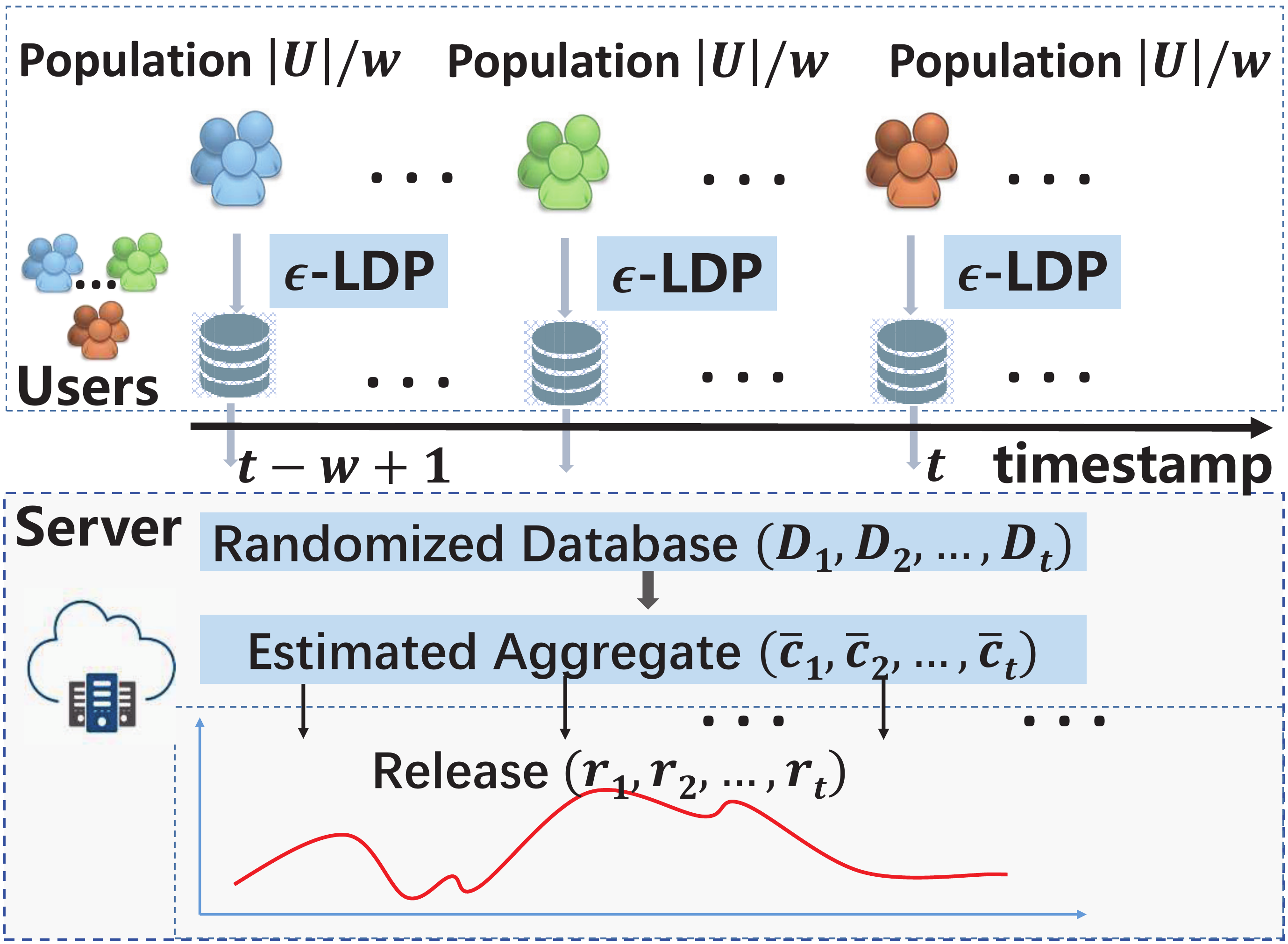}
\vspace{-2mm}	\caption{Illustration of population division framework}\label{fig: PDframework}
\end{figure}

\subsection{Basic idea} \label{subsub: LPU}
The budget division framework provides a feasible solution to infinite streaming data collection with LDP.
However, we notice that, data utility in LDP scenarios is much more sensitive to privacy budget than that in CDP. Recall that $V(\epsilon,n)$ denotes the FO estimation variance from $n$ users with privacy budget $\epsilon$. According to Eq~(\ref{var}), with fixed $n$, $V(\epsilon,n)$ is $O((e^\epsilon-1)^{-1})$ in terms of average budget $\epsilon$. It increases sharply as the budget assigned to each timestamp becomes small. 
Recently, several previous studies \cite{Wang2017LocallyDPFE,Wang2017LocallyDPHHI} indicate that it can achieve much smaller overall error by partitioning users into groups and using the entire privacy budget in each group. With the fixed $\epsilon$, $V(\epsilon,n)$ is $O(n^{-1})$ in terms of average user population $n$, which increases much mildly as $n$ becomes small. 
Therefore, we adopt this idea in streaming data collection with LDP. Intuitively, a baseline methods can be derived to achieve $w$-event LDP.

{{\textbf{LDP Population Uniform Method} (\textsf{LPU})}}. At the beginning, the central server uniformly assign the uses into $w$ disjoint groups, each with roughly $N/w$ users~\footnote{Precisely, if $N \mod w \neq 0$, it may be $\lfloor N/w \rfloor$ for some groups or $\lfloor N/w \rfloor+1$ for the rests. For simplicity, we assume $N/w$ for each.}. At each timestamp, it requests a group of users that have never been requested before to report their value. In a window with $w$ timestamps, each group of users will only report once with the entire budget $\epsilon$. And after $w-1$ timestamps, each group users will be requested and report again for the new sliding window. In this case, any user does not report in each sliding window more than once, thus spending no more than $\epsilon$-LDP budget. Hence, $w$-event LDP is guaranteed for each user.
Fig.~\ref{fig: PDframework} illustrates the population division methodology. 

\begin{theorem}\label{lemma: mse comparison}
	Given the same FO protocol GRR or OUE, the MSE of \textsf{LPU} is smaller than that of \textsf{LBU}, i.e., $\text{MSE}_\textsf{LPU}<\text{MSE}_\textsf{LBU}$.
\end{theorem}
\begin{proof}
	See Appendix~\ref{appendice: UU and UP}.
\end{proof}\vspace{-2mm}
Note that, since only a portion of users participate in reporting, the population division methodology can also greatly reduce the communication cost. In \textsf{LPU}, the number of users upload perturbed data at each timestamp is only $1/w$ of the whole population in average. Therefore, the communication cost is $1/w$ of that in \textsf{LBU}. 



\textsf{LSP} in Section~\ref{subsub: sample} can be also seen as a population division method. Particularly, all users are regarded as to be divided into $w$ groups, in which, one group has the whole population and the rests have zero users. Then, that one group of users are assigned to report at a single timestamp within the window using the entire budget $\epsilon$ while no users for the next $w-1$ timestamps.

However, both \textsf{LPU} and \textsf{LSP} cannot be adaptive to streams with unknown fluctuations, which still limits their utility.

\subsection{Adaptive Population Division Methods}
\textsf{LBD/LBA} provides a reference framework that improves the utility of baseline methods via adaptively assigning privacy budget according to the non-deterministic sparsity in data streams. In the following, we present two adaptive population division methods \textsf{LPU/LPA}, which migrates this idea to the population division framework to significantly enhance the utility.

\subsubsection{Overview}

For better analogy to \textsf{LBD/LBA}, we still introduce the population division based methods with two sub mechanisms $\mathcal{M}_1$ and $\mathcal{M}_2$. 
We first evenly partition the whole population $U$ of size $N$ into \textit{dissimilarity users} $U_1$ of size $|U_1|$ for $\mathcal{M}_1$ and \textit{publication users} $U_2$ of size $|U_2|$ for $\mathcal{M}_2$, each with $\lfloor N/2 \rfloor$ users. Similarly, $\mathcal{M}_1$ mainly achieves private dissimilarity calculation. Differently, under the population division framework, $\mathcal{M}_2$ accomplishes private strategies determination and participant users allocation.

\textbf{\textit{Private dissimilarity calculation}}: Section~\ref{subsub: pde} defines the dissimilarity measure $dis$ in the LDP setting. In $\mathcal{M}_1$, at each timestamp $t$, we still aim to estimate the dissimilarity $dis^*=\frac{1}{d}\sum\nolimits_{k=1}^d(\mathbf{c}_t[k]-\mathbf{r}_l[k])^2$ based on Eq.~(\ref{eq: estimation}). Similarly, we have to first obtain an unbiased estimation $\mathbf{\overline{c}}_t$ through an FO at timestamp $t$.
Differently, under the population division methodology, it can only be obtained from the LDP protected data (with privacy budget $\epsilon$) of dissimilarity users $U_{t,1}$ at timestamp $t$. We here partition the $\lfloor N/2 \rfloor$ dissimilarity users over the $w$ timestamps evenly. That is to say, at each timestamp $t$, $|U_{t,1}|=\lfloor N/(2w) \rfloor$ dissimilarity users report their value via an FO using the entire budget $\epsilon$.



\textbf{\textit{Private strategy determination}}: In $\mathcal{M}_2$, the estimated dissimilarity $dis$ output by $\mathcal{M}_1$ (i.e., approximation error) and the potential publication error $err$ are compared to empirically choose a better strategy (i.e., with smaller error) from approximation and publication. The publication error $err$ equals to the estimation variance $V(\epsilon,|U_{t,2}|)$, can be calculated based on the available privacy budget $\epsilon_{t,2}$ as well as the number of potential publication users $|U_{t,2}|$, e.g., according to Eq.~(\ref{var}) in GRR. Under the population division framework, the budget is fixed as a constant $\epsilon$ and the publication error is determined by $|U_{t,2}|$, which is dynamically assigned in a sliding window. The more publication users, the less the publication error $err$ is. However, since any user only participates once in a window, the availability of publication users $U_{t,2}$ in each timestamp $t$ in the sliding window is limited and should be carefully assigned.

\textbf{\textit{Participant users allocation}}: In $\mathcal{M}_2$, with the above transition from budget division to population division, the adaptive budget allocation schemes, i.e., budget distribution (in \textsf{LBD}) and budget absorption (\textsf{LBA}) can be transferred for assigning the number of publication users $|U_{t,2}|$ under the population division framework. This also leads to two adaptive population division methods: \textit{population distribution \textsf{LPD}} and \textit{population absorption \textsf{LPA}}. 

In the following, we present the details of \textsf{LPD} and \textsf{LPA}. 

\begin{algorithm}[tbp]\footnotesize
	\SetAlgoLined
	\caption{LDP Population Distribution (\textsf{LPD})}\label{Alg UD}
	\KwIn{Total population $U$ of size $N=|U|$, privacy budget $\epsilon$, window size $w$}
	\KwOut{Released statistics $R_t=(\mathbf{r}_1,\mathbf{r}_2,\ldots,\mathbf{r}_t, \ldots)$}
	\BlankLine
	Initialize available userset $U_A=U$, and $\mathbf{r}_0=\langle 0, \ldots, 0 \rangle$; \\
	\For{each timestamp $t$}
	{
		\tcp{Sub Mechanism $\mathcal{M}_{t,1}$:}
		Sample users $U_{t,1}$ from $U_A$ with the size of $\lfloor N/(2w)\rfloor$, remove $U_{t,1}$ from $U_A$, i.e., $U_A=U_A \setminus U_{t,1}$; \\
		$\overline{D}_{t,1}\leftarrow$ Users in $U_{t,1}$ report via an FO with privacy budget $\epsilon$;\\
		Estimate $\overline{\mathbf{c}}_{t,1}\leftarrow \text{FO}(\overline{D}_{t,1},\epsilon)$ ; \\
		Calculate $dis=\frac{1}{d}\sum\nolimits_{k=1}^d(\overline{\mathbf{c}}_{t,1}[k]-\mathbf{r}_{t-1}[k])^2-\frac{1}{d}\sum\nolimits_{k=1}^d\text{Var}(\mathbf{\overline{c}}_{t,1}[k])$;\\[5pt]
		
		\tcp{Sub Mechanism $\mathcal{M}_{t,2}$:}
		Calculate remaining population size $N_{rm}=N/2- \sum\nolimits_{i=t-w+1}^{t-1}|U_{i,2}|$;  \\
		Set number of potential publication users $N_{pp}=N_{rm}/2$; \\
		Calculate potential publication error $err$ by Eq.~(\ref{var});  \\
		\eIf {$dis > err$ and $N_{pp}\geq u_\text{min}$}
		{\tcp{Publication Strategy}
			Sample a userset $U_{t,2}$ from $U_A$ with the size of $|U_{t,2}|=N_{pp}$, $U_A=U_A  \setminus U_{t,2}$;  \\
			$\overline{D}_{t,2}\leftarrow$ Users in $U_{t,2}$ report via an FO with budget $\epsilon$;\\
			Estimate $\overline{\mathbf{c}}_{t,2}\leftarrow \text{FO}(\overline{D}_{t,2},\epsilon)$; \\
			\Return $\mathbf{r}_t=\overline{\mathbf{c}}_{t,2}$; } 
		{\tcp{Approximation Strategy} \Return $\mathbf{r}_t=\mathbf{r}_{t-1}$.}
		\If{$t\geq w$}
		{\tcp{Recycling Users}	
			$U_{A}=U_A \cup U_{t-w+1,1} \cup U_{t-w+1,2}$.}
	}
\end{algorithm}
\subsubsection{\textbf{LDP Population Distribution Method (\textsf{LPD})}}\label{subsub: LPD}


Algorithm~\ref{Alg UD} presents the details of \textsf{LPD}. 
Firstly, for calculation of dissimilarity $dis$ in $\mathcal{M}_1$ (Lines 3-6), the dissimilarity population $|U_{1}|$ is uniformly divided into disjoint groups of dissimilarity users $U_{t,1}$ at each timestamp, i.e., $|U_{t,1}|=\lfloor N/(2w)\rfloor$.
Next in $\mathcal{M}_2$, the remaining number of publication users $N_\text{rm}$ is calculated by removing the already used publication users in the last $w-1$ timestamps from the total number of publication users $N/2$ (Line 7). 
Then, the number of potential publication users is set as $N_{pp}=N_\text{rm}/2$ to calculate a potential publication error $err$ (Lines 8-9).
By comparing $err$ with $dis$, the publication or approximation strategy is decided then (Lines 10-17). In case of too many publications and $N_{pp}$ decays too quickly to have no available user, a threshold $u_\text{min}$ (e.g., $u_\text{min}=1$) is set (Line 10). 
Once publication is chosen, $N_{pp}$ new users will be sampled as actual publication users $U_{t,2}$ from $U_A$ to accomplish publication (Lines 11-14). 
Otherwise, $\mathbf{r}_t$ is approximated by $\mathbf{r}_{t-1}$, without using $N_{pp}$ users (Line 16). Finally, both the used dissimilarity users and publication users (may be \emph{null}) at timestamp $t-w+1$, which is falling outside of the next active window, are recycled as available users $U_A$ (Line 19). The recycling process ensures each user can contribute again after $w$ timestamps while guaranteeing no users participate more than once. Detail description of \textsf{LPD} can be referred to Appendix~\ref{apped: UD example}.

\subsubsection{\textbf{LDP Population Absorption Method (\textsf{LPA})}}\label{subsub: LPA}

Algorithm~\ref{Alg UA} presents the details of \textsf{LPA}. The private dissimilarity calculation process of $\mathcal{M}_{t,1}$ in \textsf{LPA} is the same as that in \textsf{LPD}. In $\mathcal{M}_{t,2}$, the basic idea is to uniformly allocated users at all timestamps then the unused publication users is absorbed at the timestamps where publication is chosen. Once a publication occurs at time $l$, the same number of users must be skipped from the succeeding timestamps to ensure available users within the active sliding window. So, the number of timestamps to be nullified $t_N$ is first calculated based on the number of publication users at timestamp $l$, and thus skipped with approximation (Lines 4-6). After that, based on the timestamps can be absorbed, the number of potential publication users $N_{pp}$ is calculated at each time $t$, which can further derive the potential publication error $err$ (Lines 8-10). By comparing $err$ with $dis$, $\mathcal{M}_{t,2}$ decides whether to freshly publish with the potential publication users (Lines 11-15) or continues to approximate with the last release (Lines 16-18). Similarly, both the used dissimilarity users and publication users at timestamp $t-w+1$ are finally recycled as available users $U_A$ (Lines 20-22). Detail description of \textsf{LPA} can be referred to Appendix~\ref{apped: UA example}.

\begin{algorithm}[htbp]\footnotesize
	\SetAlgoLined
	\caption{LDP Population Absorption (\textsf{LPA})}\label{Alg UA}
	\KwIn{Total population $U$ of size $N=|U|$, privacy budget $\epsilon$, window size $w$}
	\KwOut{Released statistics $R_t=(\mathbf{r}_1,\mathbf{r}_2,\ldots,\mathbf{r}_t, \ldots)$}
	\BlankLine
	Initialize available userset $U_A=U$, and $\mathbf{r}_0=\langle 0,\ldots, 0 \rangle$, last publication timestamp $l=0$, and $\epsilon_{l,2}=0$;\\
	\For{each time $t$}{
		\tcp{Sub Mechanism $\mathcal{M}_{t,1}$:}
		Same as Lines 3-6 in \textbf{Algorithm~\ref{Alg UD}} \\[5pt]
		
		\tcp{Sub Mechanism $\mathcal{M}_{t,2}$:}
		Calculate timestamps to be nullified $t_N= \frac{|U_{l,2}|}{\lfloor N/(2w) \rfloor} -1$;\\
		\eIf{$t-l \leq t_N$}
		{
			\Return $\mathbf{r}_t=\mathbf{r}_{t-1}$;}
		{
			Calculate timestamps can be absorbed $t_A=t-(l+t_N)$;\\
			Set number of potential publication users $N_{pp}=\lfloor N/(2w) \rfloor \cdot $min$(t_A,w)$;\\
			Calculate potential publication error $err$ by Eq.~(\ref{var}); \\
			\eIf {$dis > err$}
			{\tcp{Publication Strategy}
				Sample a userset $U_{t,2}$ from $U_A$ with the size of $|U_{t,2}|=N_{pp}$, $U_A=U_A \setminus U_{t,2}$;\\
				$\overline{D}_{t,1}\leftarrow$ Users in $U_{t,2}$ report via an FO with budget $\epsilon$;\\
				Estimate $\overline{\mathbf{c}}_{t,2}\leftarrow \text{FO}(\overline{D}_{t,2},\epsilon)$; \\
				\Return $\mathbf{r}_t=\overline{\mathbf{c}}_{t,2}$, set $l=t$;}
			{\tcp{Approximation Strategy} \Return $\mathbf{r}_t=\mathbf{r}_{t-1}$.}
		}
		\If{$t\geq w$}
		{	\tcp{Recycling Users}	
			$U_{A}=U_A \cup U_{t-w+1,1} \cup U_{t-w+1,2}$;}
	}
\end{algorithm}

\subsection{Analysis}
\subsubsection{Privacy Analysis}\label{sec: budget-division privacy analysis}
\textsf{LPD} and \textsf{LPA} satisfy $w$-event LDP because each user reports to the server at most once in a time window of size $w$ and each report goes through an FO with $\epsilon$-LDP. 
\begin{theorem}\label{theorem: privacy analysis}
\textsf{LPD} and \textsf{LPA} satisfies $w$-event LDP for each user.
\end{theorem}\vspace{-2mm}
\begin{proof}\vspace{-2mm}
See Appendix~\ref{appendice: privacy proof UD/UA}.
\end{proof}

\subsubsection{Utility Analysis}
With the same assumptions, similar MSE expression can be obtained as Eq.~(\ref{eq: BD/BA-MSE}) in Section~\ref{sec: budget-division analysis}.
Then, under the population division framework, in \textsf{LPD}, since the population is distributed to the $m$ publications in an exponentially decreasing way, the population alloction sequence of $N_{t,2}$ is then $N/4, N/8, \ldots, N/2^{m+1}$. There is
\begin{align}
\sum_{i = 1}^m\text{Var}_\textsf{LPD}[\mathbf{r}_{p_i}]=\sum_{i = 1}^m V(\epsilon,N/2^{i+1}) 
\end{align}
Therefore, the error of \textsf{LPD} would still increase with $m$. However, according to Lemma~\ref{lemma: mse comparison}, we can conclude that $V(\epsilon,N/2^{m+1})<V(\epsilon/2^{m+1},N)$. That is to say, \textsf{LPD} can achieve less error than \textsf{LBD}.
Similarly, in \textsf{LPA},  
\begin{align}
\sum\limits_{i = 1}^m\text{Var}_\textsf{UA}[\mathbf{r}_{p_i}]=&m\cdot V(\epsilon,\frac{w+m}{4\cdot w \cdot m}\cdot N)
\end{align}
which is smaller than $m\cdot V(\frac{w+m}{4\cdot w \cdot m}\cdot\epsilon, N)$ in Eq.~(\ref{eq: BAerror}) of the budget absorption method \textsf{LBA}, given the same assumptions.
\subsubsection{Communication Analysis}\label{subsub: communication population}
In \textsf{LPD}, all $m$ publications in a window need $\sum\limits_{i = 1}^m (N/2^{i+1}+N/(2w))=(\frac{1-(1/2)^m}{2}+\frac{m}{2w})\cdot N$ users to communicate and the rest $w-m$ approximations need $\frac{w-m}{2w}\cdot N$ users. Therefore, the average CFPU is $\frac{1}{w\cdot N}[(\frac{1-(1/2)^m}{2}+\frac{m}{2w})\cdot N+\frac{w-m}{2w}\cdot N]=\frac{1}{w}-\frac{1}{w\cdot 2^{m+1}}$.

In \textsf{LPA}, all $m$ publications in a window need $m \cdot (\frac{w+m}{4\cdot w \cdot m}\cdot N+N/(2w))$ users to communicate and the rest timestamps need $(w-m)\cdot N/(2w)$ users. Therefore, the average CFPU is $\frac{1}{w\cdot N}[m \cdot (\frac{w+m}{4\cdot w \cdot m}\cdot N+N/(2w))+(w-m)\cdot N/(2w)]=\frac{1}{2w}+\frac{w+m}{4w^2}$.

\subsection{Discussion}
In this subsection, we briefly discuss the differences of our methods from existing ones, open problems, and future directions.

Remark 1: LDP methods \textsf{LBD/LBA} and \textsf{LPD/LPA} proposed above are inspired from \textsf{BD/BA} in the centralized setting of DP. But they are different in many aspects, including but not limited to, the information that the server can access to and operations it can perform, the perturbation mechanisms, the measurement of error in perturbation and approximation, and the population division framework instead of budget division. 

Remark 2: The \textsf{LPD} and \textsf{LPA} methods proposed above can be applied to a large spectrum of IoTs scenarios that massive reliable devices persistently monitor the environment or events, such as smart metering systems, security and video cameras. 
Note that, in mobile scenarios, the number of joining devices may be time-varying, e.g., new users may join in and churn randomly, which may make this framework complicated. 

Remark 3: Besides \textsf{LPD} and \textsf{LPA}, the population division-based LDP framework can be easily applied and extended to other state-of-the-art DP methods for streams (including user-level DP), such as FAST~\cite{Fan2014AnAA}, PeGaSus~\cite{Chen2017PeGaSusDD} and RescueDP~\cite{Wang2016RescueDPRS}, which may achieve better utility but the techniques in these methods often require complicated parameters tuning with extra effort. 

\section{Performance Evaluation}\label{sec: experiment}
In this section, we conducted extensive experiments to evaluate the performance of our proposed algorithms.


\subsection{Experimental Setup}

\subsubsection{\textbf{Synthetic Datasets}}

We synthesized binary streaming datasets with different sequence models. Given a probability process model $p_t=f(t)$, the length of time $T$, and user population $N$,  we first generated a probability sequence $(p_1, p_2, \ldots, p_T)$ with $T$ timestamps. Then, at each timestamp $t$, we randomly chose a portion of $p_t$ users from the total $N$ users to set their true report value $v^j_t$ as $1$, and set the rest as $0$. The following typical sequence patterns were used.

\begin{itemize}
	\item \textsf{LNS} is a linear process $p_t=p_{t-1}+\mathcal{N}(0,Q)$, where $p_0=0.05$ and $\mathcal{N}(0,Q)$ is Gaussian noise with the standard variance $\sqrt{Q}=0.0025$.
	\item \textsf{Sin} is a sequence composed by a sine curve $p_t=A sin(bt)+h$ with $A=0.05$, $b=0.01$ and $h=0.075$.
	\item \textsf{Log} is a series with the logistic model $p_t=A/(1+e^{-bt})$ where $A=0.25$ and $b=0.01$.
\end{itemize}

Without specifying, we used above models and default parameters to generate synthetic binary streams with $800$ timestamps of $200,000$ users. 
To demonstrate the varying fluctuations, we set $N$ fixed but changed the parameters $Q$ in \textsf{LNS} and $b$ in \textsf{Sin} respectively to obtain different datasets. 
To demonstrate the varying populations, we used the probability sequences generated with the default parameters above, but performed different number of sampling processes to obtain datasets with different population $N$.

%
%
\subsubsection{\textbf{Real-world Datasets}}
To evaluate the practical performance of algorithms, the following three real-world datasets with non-binary values were also used. 

\begin{itemize}
	\item\textsf{Taxi}\footnote{\url{https://www.microsoft.com/en-us/research/publication/t-drive-trajectory-data-sample/}} contains the real-time trajectories of $10,357$ taxis during the period of Feb. 2 to Feb. 8, 2008 within Beijing. We obtained $N=10,357$ data streams for each taxi by extracting $T=886$ timestamps (each at $10$-minute level) and partitioning area into $5$ grids, i.e., $d=5$.
	\item\textsf{Foursquare}\footnote{\url{https://sites.google.com/site/yangdingqi/home/foursquare-dataset}} includes $33,278,683$ check-ins of Foursquare users from Apr. 2012 to Sep. 2013, where each record includes time, place and user ID. We transformed it into $N=265,149$ data streams with the length of $T=447$ timestamps, each records a user's check-in sequence over $d=77$ countries.
	\item \textsf{Taobao}\footnote{\url{https://tianchi.aliyun.com/dataset/dataDetail?dataId=56}} contains the AD click logs of 1.14 million customers at Taobao.com. For simplicity, we first grouped the AD commodities into $d=117$ categorizes. Then, we extracted all the $N=1,023,154$ customers' click data streams, where each item corresponds to the categorize of the user's last click during each ten minutes in three consecutive days, i.e., $T=432$ timestamps.
\end{itemize}

\subsubsection{\textbf{Compared Algorithms}}
We compared the following algorithms. All are implemented using Matlab 2020.

\vspace{3pt}\begin{small} \centering
\begin{tabular}{|l|c|c|c|c|} 
	\hline
	\multicolumn{1}{|c|}{\multirow{2}{*}{\textbf{\begin{tabular}[c]{@{}c@{}}Scheme \\ Dimensions\end{tabular}}}} & \multicolumn{2}{c|}{\textbf{\begin{tabular}[c]{@{}c@{}}Non-adaptive\\ allocation\end{tabular}}} & \multicolumn{2}{c|}{\textbf{\begin{tabular}[c]{@{}c@{}}Adaptive\\ allocation\end{tabular}}}                           \\ \cline{2-5} 
	\multicolumn{1}{|c|}{}                                                                                       & \multicolumn{1}{l|}{Uniform}                   & \multicolumn{1}{l|}{Sampling}                  & \multicolumn{1}{l|}{Distribution}                         & \multicolumn{1}{l|}{Absorption}                           \\ \hline
	\textbf{\begin{tabular}[c]{@{}l@{}}Budget\\ division\end{tabular}}                                           & \textsf{LBU}                                            & \multirow{2}{*}{\textsf{LSP}}                           & \begin{tabular}[c]{@{}c@{}}\textsf{LBD}\\ (Algorithm~\ref{Alg BD})\end{tabular} & \begin{tabular}[c]{@{}c@{}}\textsf{LBA}\\ (Algorithm~\ref{Alg BA})\end{tabular} \\ \cline{1-2} \cline{4-5} 
	\textbf{\begin{tabular}[c]{@{}l@{}}Population\\ division\end{tabular}}                                       & \textsf{LPU}                                            &                                                & \begin{tabular}[c]{@{}c@{}}\textsf{LPD}\\ (Algorithm~\ref{Alg UD})\end{tabular} & \begin{tabular}[c]{@{}c@{}}\textsf{LPA}\\ (Algorithm~\ref{Alg UA})\end{tabular} \\ \hline
\end{tabular}
\end{small}

\textsf{LBU} (Sec. \ref{subsub: LBU}) and \textsf{LPU} (Sec. \ref{subsub: LPU}) are baseline methods that uniformly divide budget and population, respectively. \textsf{LSP} (Sec. \ref{subsub: sample} and \ref{subsub: LPU}) invests the entire budget and users at sampling timestamps with fixed interval. \textsf{LBD}/ \textsf{LBA} (Algs.~\ref{Alg BD}, \ref{Alg BA} in Sec. \ref{subsub: BDBA}) and \textsf{LPD}/\textsf{LPA} (Algs.~\ref{Alg UD}, \ref{Alg UA} in Sec.~\ref{subsub: LPD} and~\ref{subsub: LPA}) adaptively allocate the budget and population via two different schemes, respectively.

All experiments were conducted on a PC with an Intel Core i5-6300HQ 3.20GHz and 16GB memory.

\subsubsection{\textbf{Performance Metrics}}
We evaluated the performance of different algorithms in terms of data utility, event monitoring efficiency, and communication efficiency. The utility was measured as the \textit{mean relative error} (MRE) between the released and true statistics. 
The event monitoring efficiency was measured as the ratio that, from perturbed reports, the server successfully detects extreme events, i.e., the statistics of which are greater than a given threshold. The communication efficiency is mainly compared by counting the \textit{communication frequency per user} (CFPU).

\subsection{Overall Utility}

Fig.~\ref{fig:ErrorEpsilon} shows the release accuracy of all compared $w$-event LDP methods on all synthetic and real-world datasets, with different privacy budget $\epsilon$. These methods are categorized into budget division-based (\textsf{LBU, LBD, LBA}) and population division-based methods (\textsf{LSP, LPU, LPD, LPA}). Overall, the error of all methods decreases with $\epsilon$, which shows the tradeoff between data utility and privacy.  Besides, the population division-based methods significantly outperform budget-division ones with much smaller MRE. This is because LDP is more sensitive to the budget division than population division.
\textsf{LBD/LBA} generally shows smaller error than the straightforward method of \textsf{LBU}. This is because \textsf{LBD/LBA} can utilize temporal correlations in data streams to reduce the privacy budget consumption rate.  
The advantage of \textsf{LPD/LPA} is clearer as noise does not increase dramatically when the population is divided. Although \textsf{LSP} achieves even smaller error than \textsf{LPD/LPA}, its performance varies dramatically across datasets. 
\begin{figure}[tbp]
	\centering	
	\subfigure[\textsf{LNS}, $w=20$]{
		\label{LMAEw-L} 
		\includegraphics[width=115pt,height=80pt]{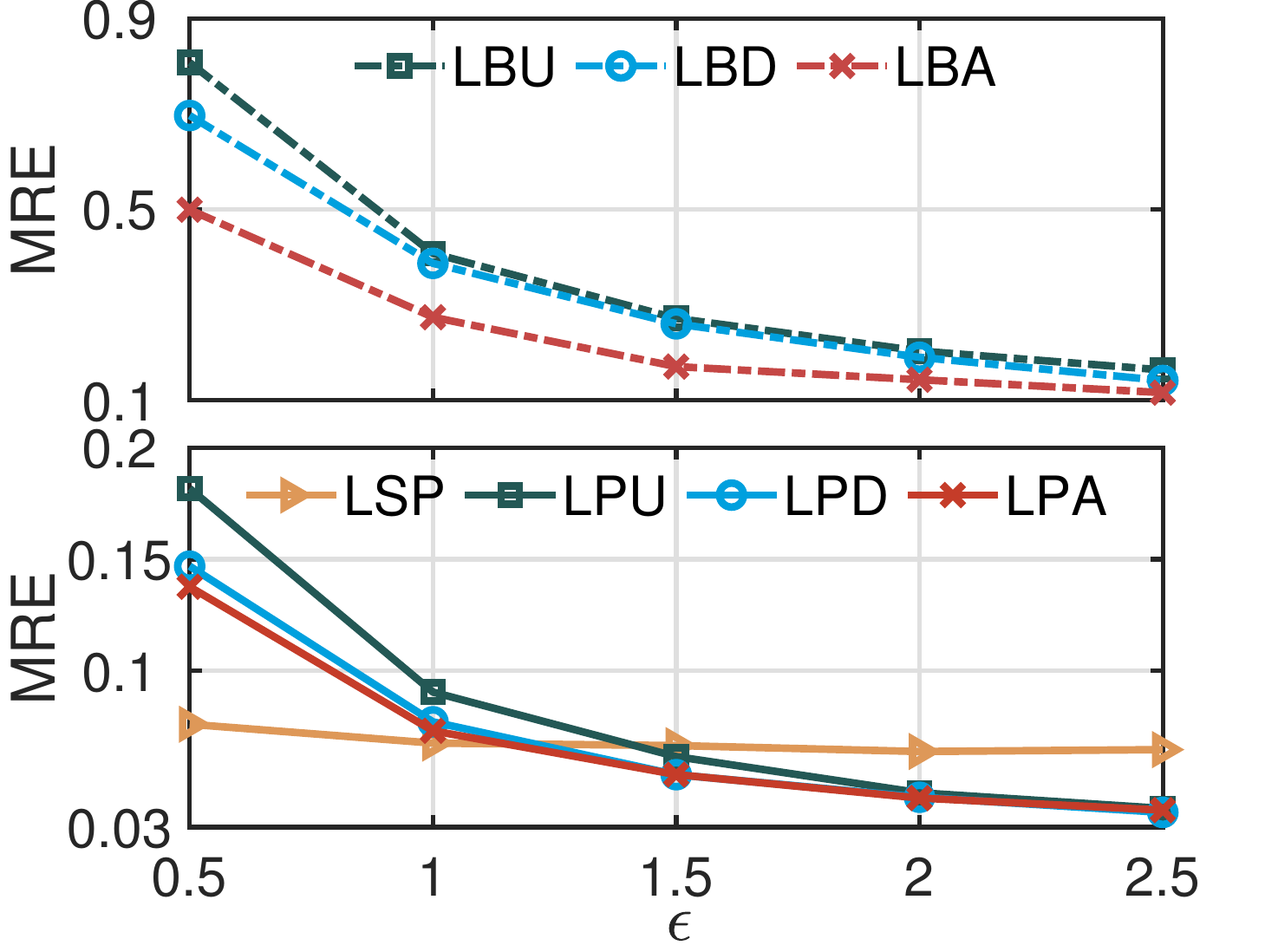}}\vspace{-1.5mm}
	\subfigure[\textsf{Sin}, $w=20$]{
		\label{LMAEw-L} 
		\includegraphics[width=115pt,height=80pt]{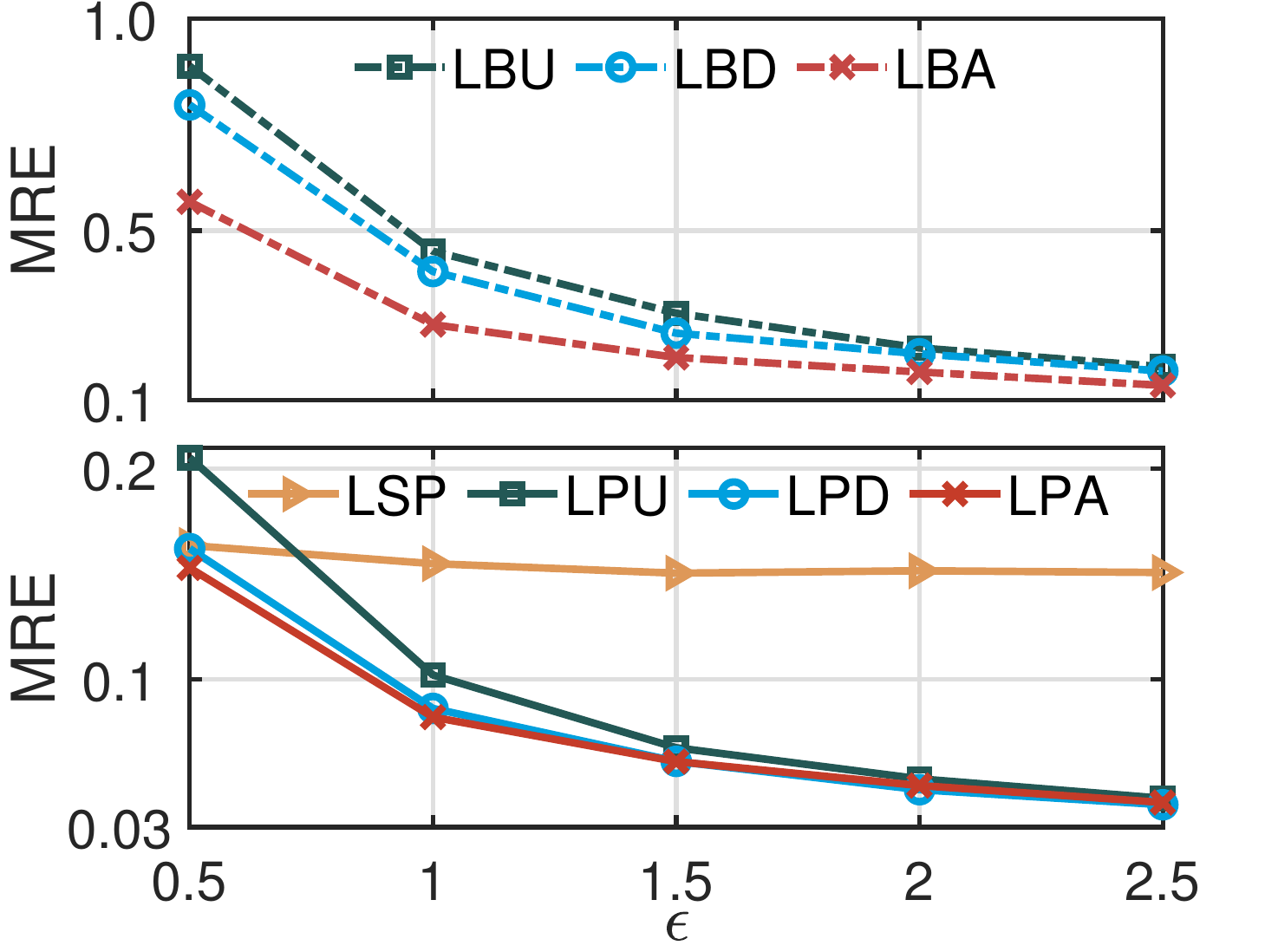}}\vspace{-1.5mm}
	\subfigure[\textsf{Log}, $w=20$]{
		\label{LMAEw-L} 
		\includegraphics[width=115pt,height=80pt]{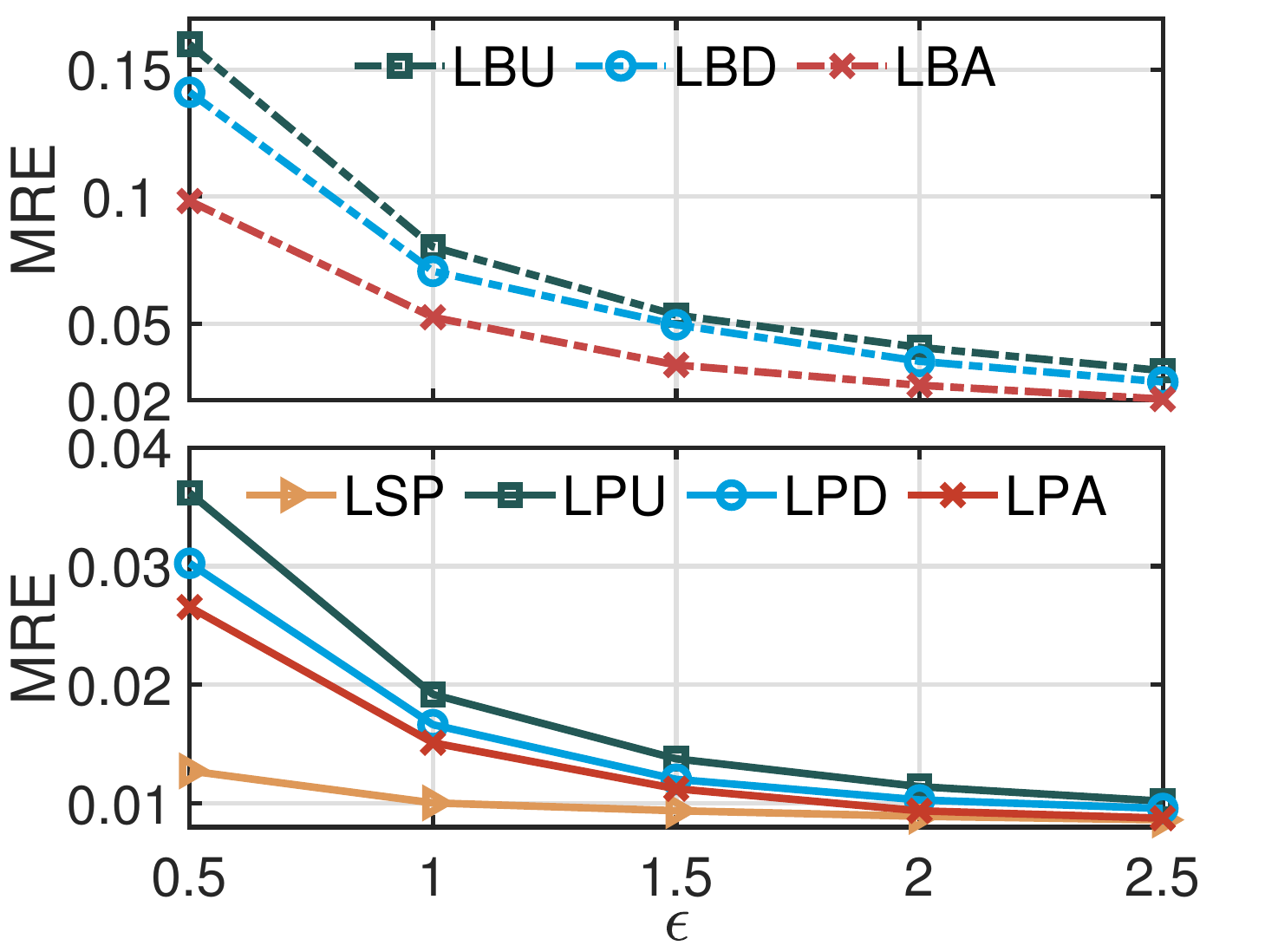}}\vspace{-1.5mm}
	\subfigure[\textsf{Taxi}, $w=20$]{
		\label{LMAEw-L} 
		\includegraphics[width=115pt,height=80pt]{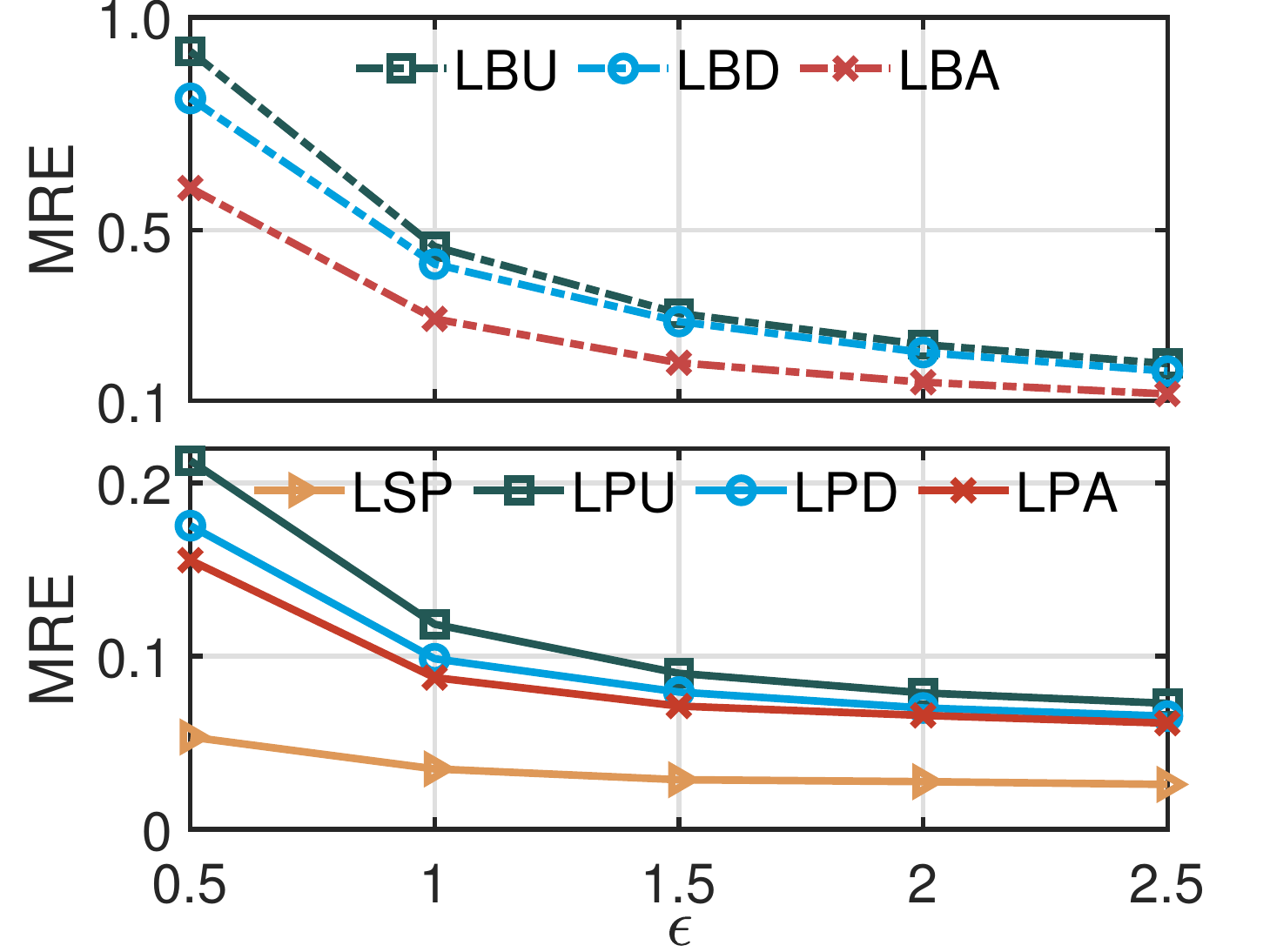}}\vspace{-1.5mm}
	\subfigure[\textsf{Foursquare}, $w=20$]{
		\label{LMAEw-L} 
		\includegraphics[width=115pt,height=80pt]{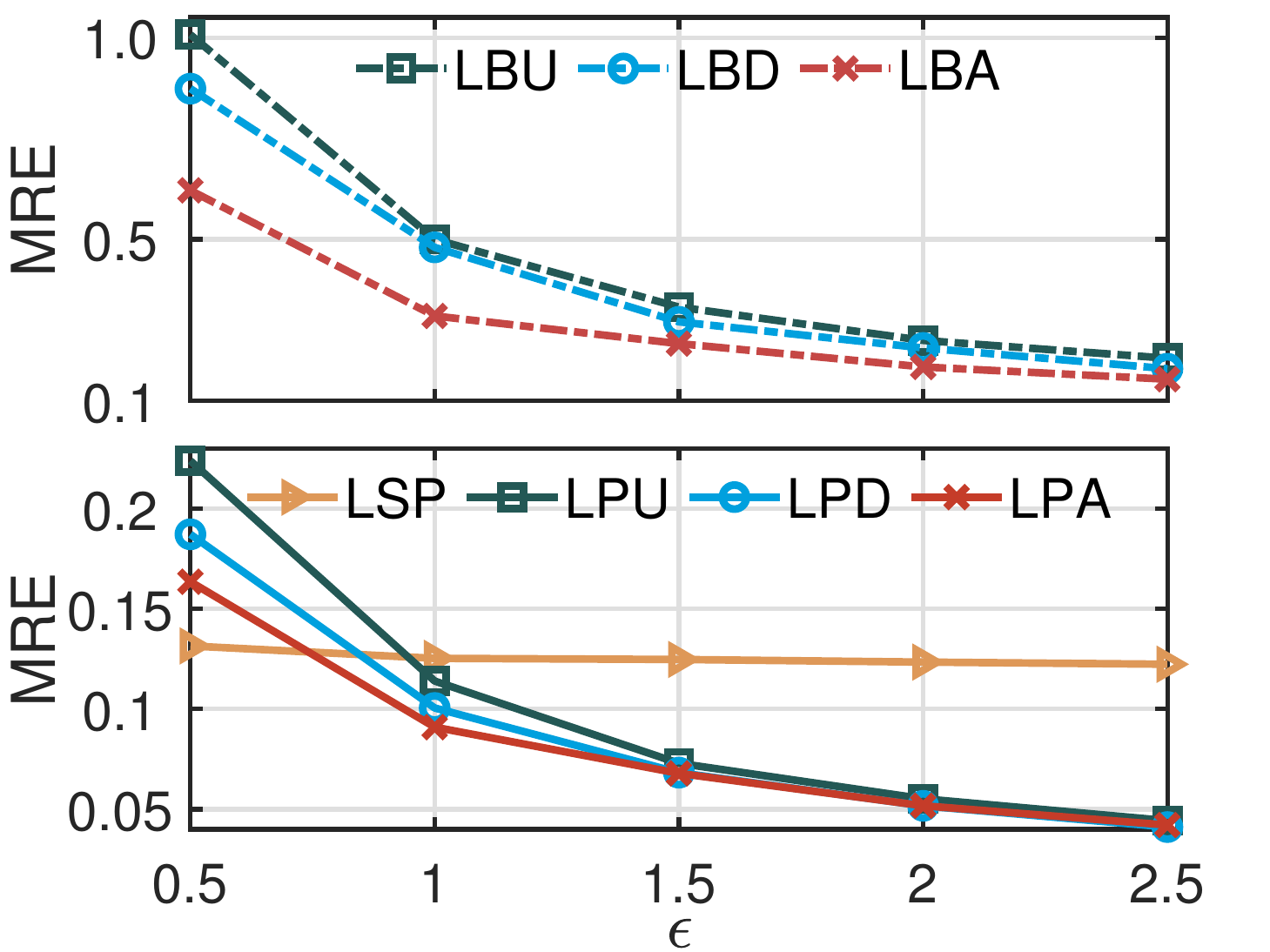}}\vspace{-1.5mm}	
	\subfigure[\textsf{Taobao}, $w=20$]{
		\label{LMAEw-L} 
		\includegraphics[width=115pt,height=80pt]{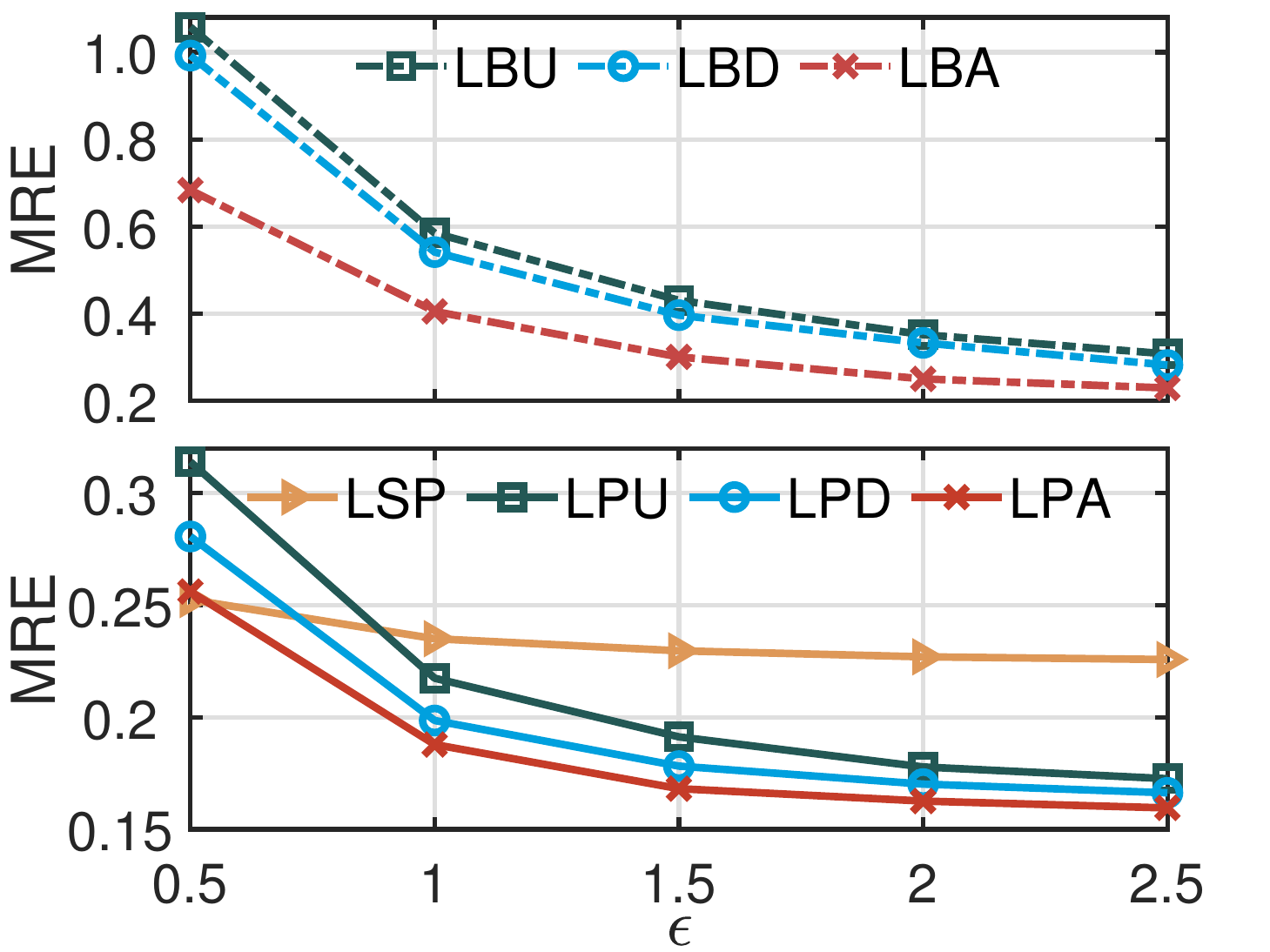}}\vspace{-1.5mm}
	\caption{Data utility with different $\epsilon$}\centering
	\label{fig:ErrorEpsilon} 
\end{figure}
\begin{figure}[htbp]
	\centering
	\subfigure[\textsf{LNS}, $\epsilon=1$]{
		\label{LMAEw-F} 
		\includegraphics[width=115pt,height=80pt]{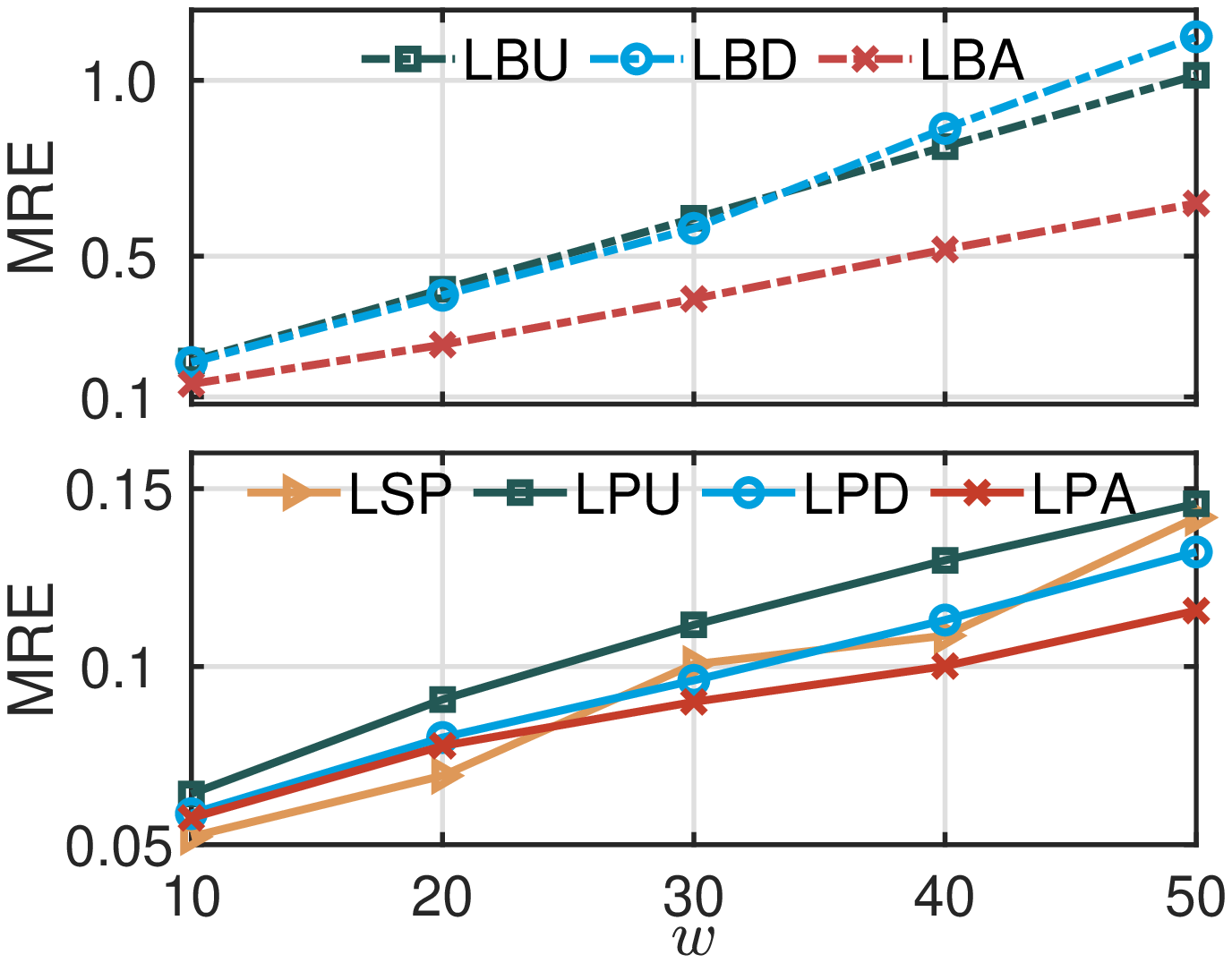}}\vspace{-1.5mm}
	\subfigure[\textsf{Sin}, $\epsilon=1$]{
		\label{LMAEw-F} 
		\includegraphics[width=115pt,height=80pt]{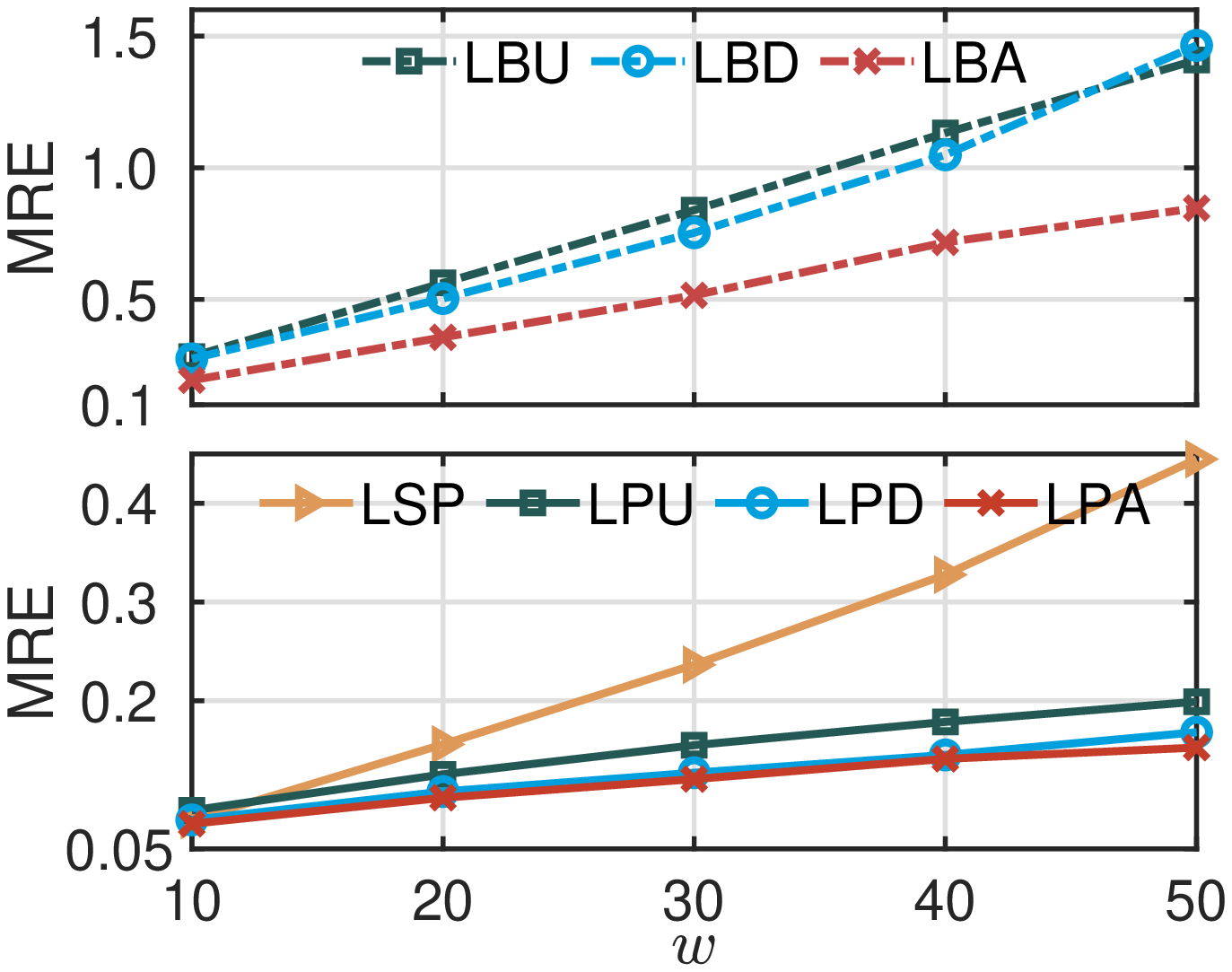}}\vspace{-1.5mm}
	\subfigure[\textsf{Log}, $\epsilon=1$]{
		\label{LMAEw-F} 
		\includegraphics[width=115pt,height=80pt]{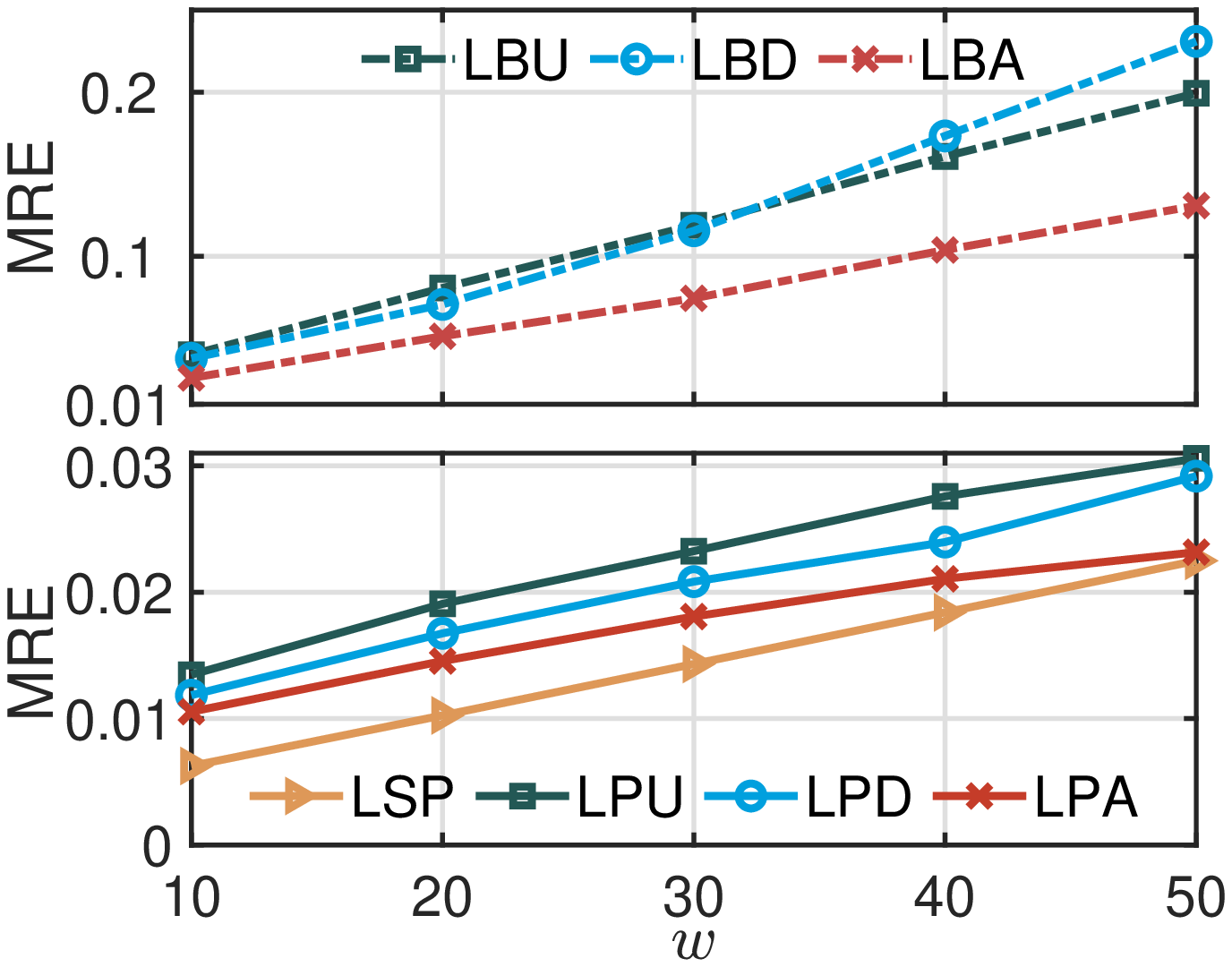}}\vspace{-1.5mm}
	\subfigure[\textsf{Taxi}, $\epsilon=1$]{
		\label{LMAEw-F} 
		\includegraphics[width=115pt,height=80pt]{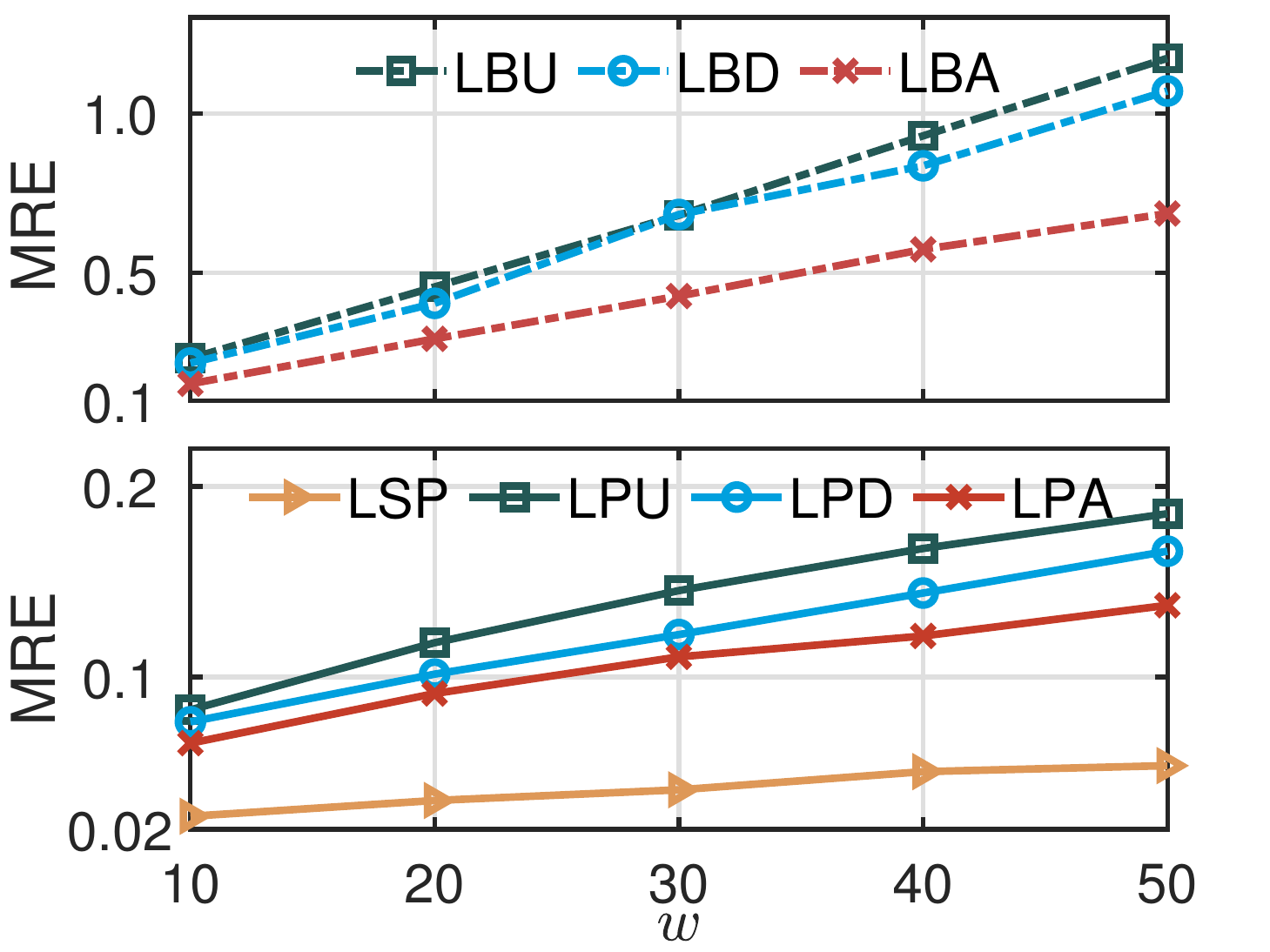}}\vspace{-1.5mm}
	\subfigure[\textsf{Foursquare}, $\epsilon=1$]{
		\label{LMAEw-F} 
		\includegraphics[width=115pt,height=80pt]{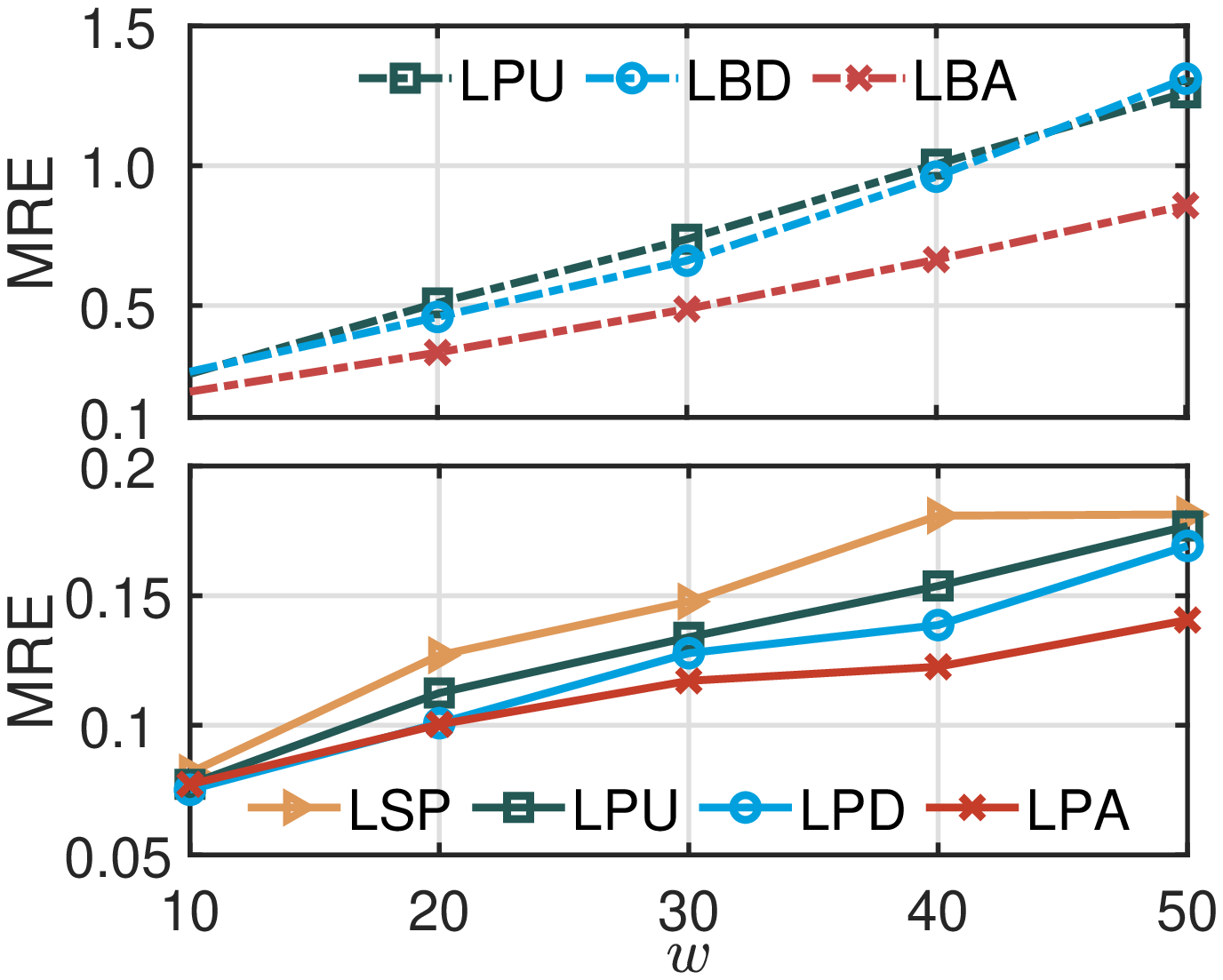}}	\vspace{-1.5mm}
	\subfigure[\textsf{Taobao}, $\epsilon=1$]{
		\label{LMAEw-F} 
		\includegraphics[width=115pt,height=75pt]{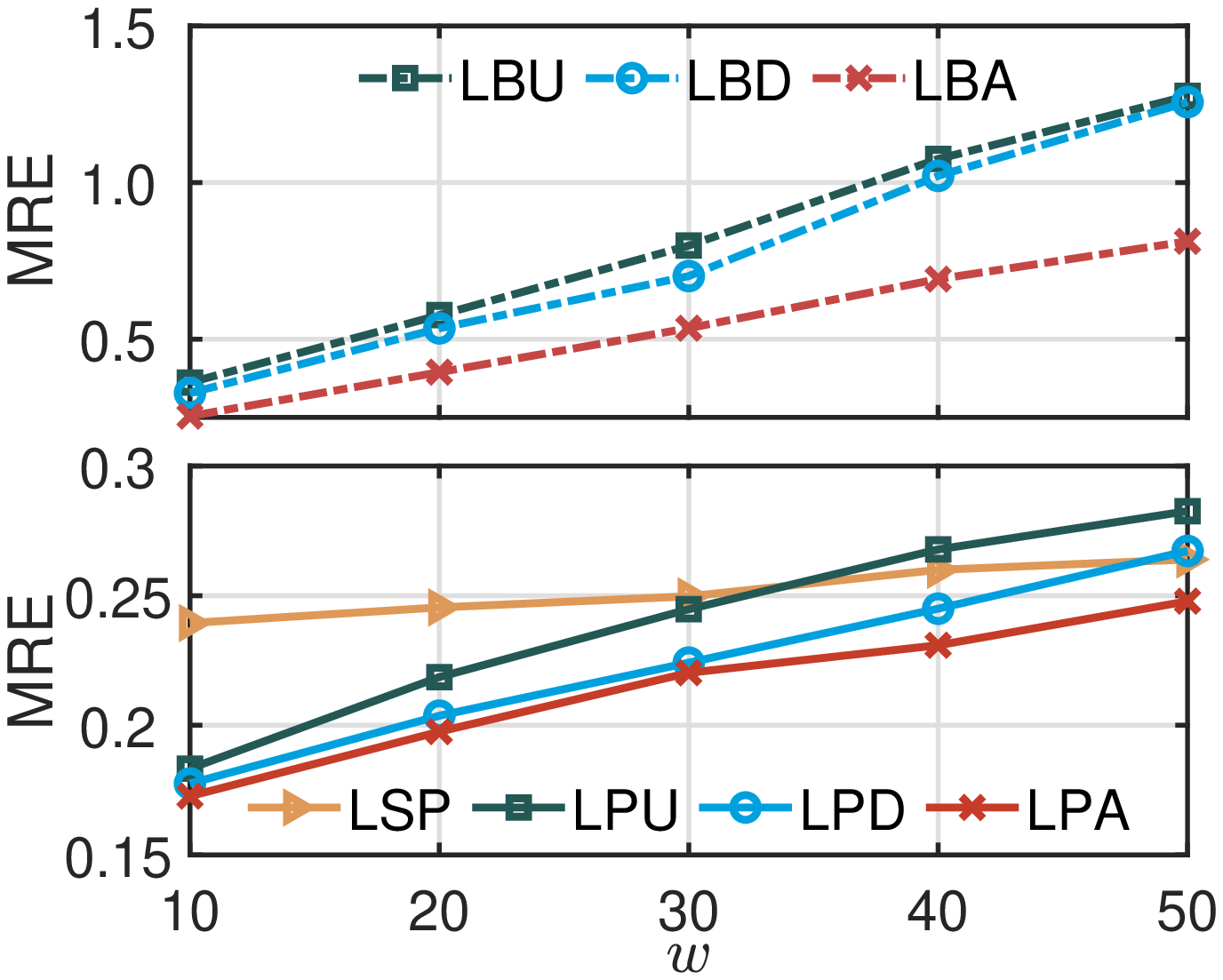}}	\vspace{-1.5mm}
	\caption{Data utility with different $w$}\centering
	\label{fig:ErrorWindow} 
\end{figure}	

Fig.~\ref{fig:ErrorWindow} shows the release accuracy of all compared $w$-event LDP methods, with different window size $w$, on all datasets. In general, the MRE of all methods increases with $w$, since less privacy budget or users will be allocated to each timestamp. For budget division methods, with the increasing $w$, \textsf{LBD} distributes budget in an exponentially decaying way and allocates very small budget for the newest timestamp in the window, thus causing quite large estimation error. For example, when $w$ is large, \textsf{LBD} may even have larger MRE than \textsf{LBU}. \textsf{LBA} to avoid this issue and well adapt to data fluctuations. For population division methods, despite the similar trends, \textsf{LPD} manages to achieve smaller MRE than baseline \textsf{LPU} while \textsf{LPA} performs even much better. Besides, as shown, with large $w$,  \textsf{LPD/LPA} gains more prominent advantages.



\subsection{Impact of Dataset Parameters}
To demonstrate the impact of dataset (population size and data fluctuations) on utility performance, we changed population $N$, and noise variance $Q$ and period parameter $b$ in the synthetic datasets \textsf{LNS} and \textsf{Sin} respectively. Note that, while varying the population size $N$, we kept the frequency fixed. 

Figs.~\ref{LMREn-L} and \ref{LMREn-S} show the MRE of compared methods on \textsf{LNS} and \textsf{Sin} with respect to different population size $N$. As shown, the MRE of all methods decreases with $N$. This is because that, enlarging $N$ while fixing the frequency leads to better estimation accuracy. 
Figs.~\ref{LMREq-L} and \ref{LMREq-S} show the MRE on \textsf{LNS} with different variance $Q$ and \textsf{Sin} with different period parameter $b$, respectively. On \textsf{LNS}, the MRE of all methods increases with $Q$, which measures the fluctuation in streams. This result verifies that these methods are data-dependent and perform better on streams with few changes. We can also see that budget division methods including \textsf{LBU}, \textsf{LBD} and \textsf{LBA} have much larger error than the population division methods. Although \textsf{LSP} manages to have the smallest error when the variance is small (i.e., $\sqrt{Q}=0.001, 0.002$), it grows fast and surpasses \textsf{LPD} and \textsf{LPA} with the increase of $Q$. \textsf{LPD} and \textsf{LPA} induce much smaller error, and perform slightly worse than \textsf{LPU} when the variance is large. Note that, period parameter $b$ also represents data fluctuations and larger $b$ means larger fluctuations. Similar conclusions can be obtained from \textsf{Sin} that population division based LDP methods manage to achieve much higher utility, and \textsf{LPD} and \textsf{LPA} manage to further improve the utility on steady streams.
\begin{figure}[htbp]
	\centering
	\subfigure[Varying population $N$ (\textsf{LNS})]{
		\label{LMREn-L} 
		\includegraphics[width=115pt,height=80pt]{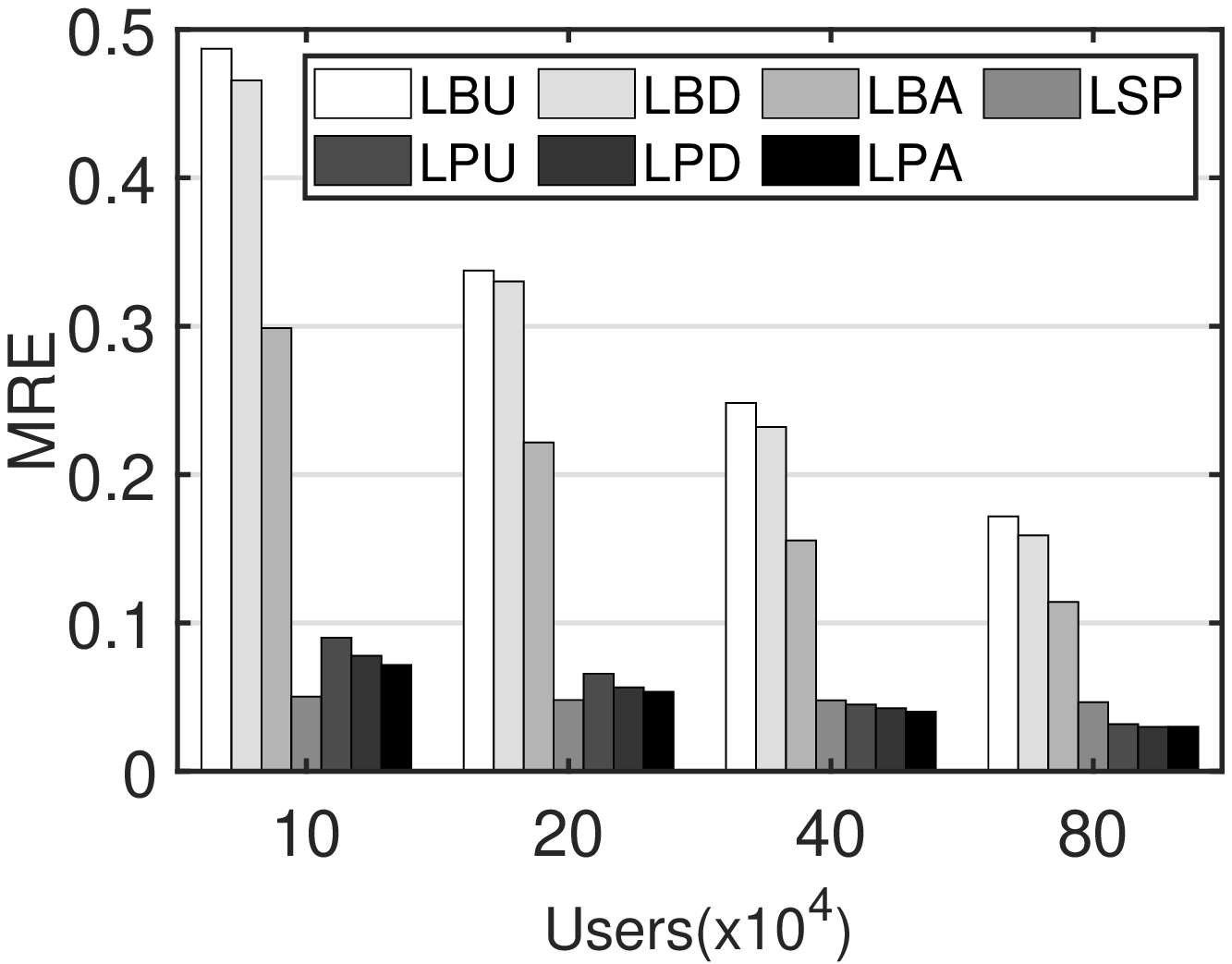}}
	\subfigure[Varying population $N$ (\textsf{Sin})]{
		\label{LMREn-S} 
		\includegraphics[width=115pt,height=80pt]{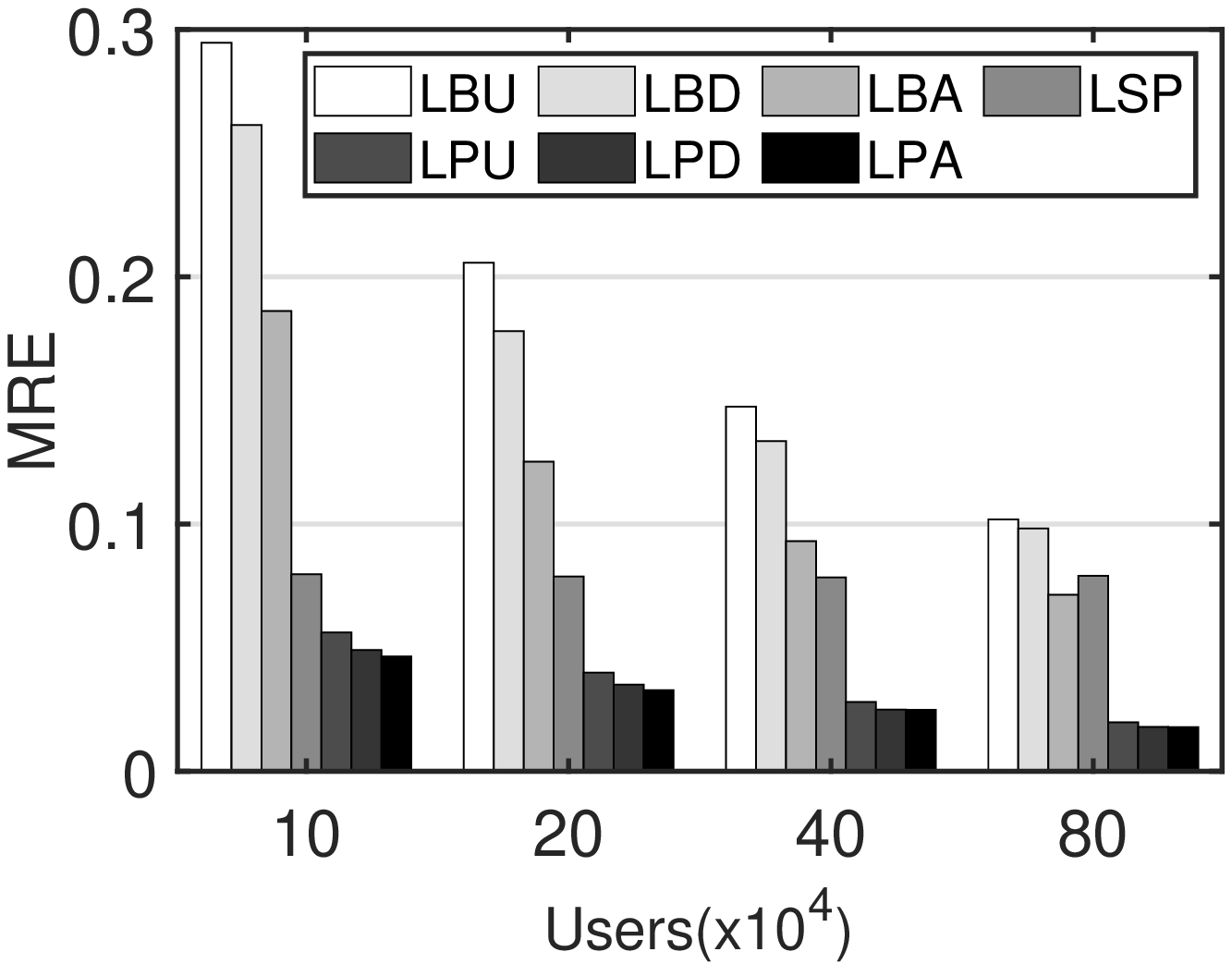}} \\[-10pt]
	\subfigure[Varying fluctuation $Q$ (\textsf{LNS})]{
		\label{LMREq-L} 
		\includegraphics[width=115pt,height=80pt]{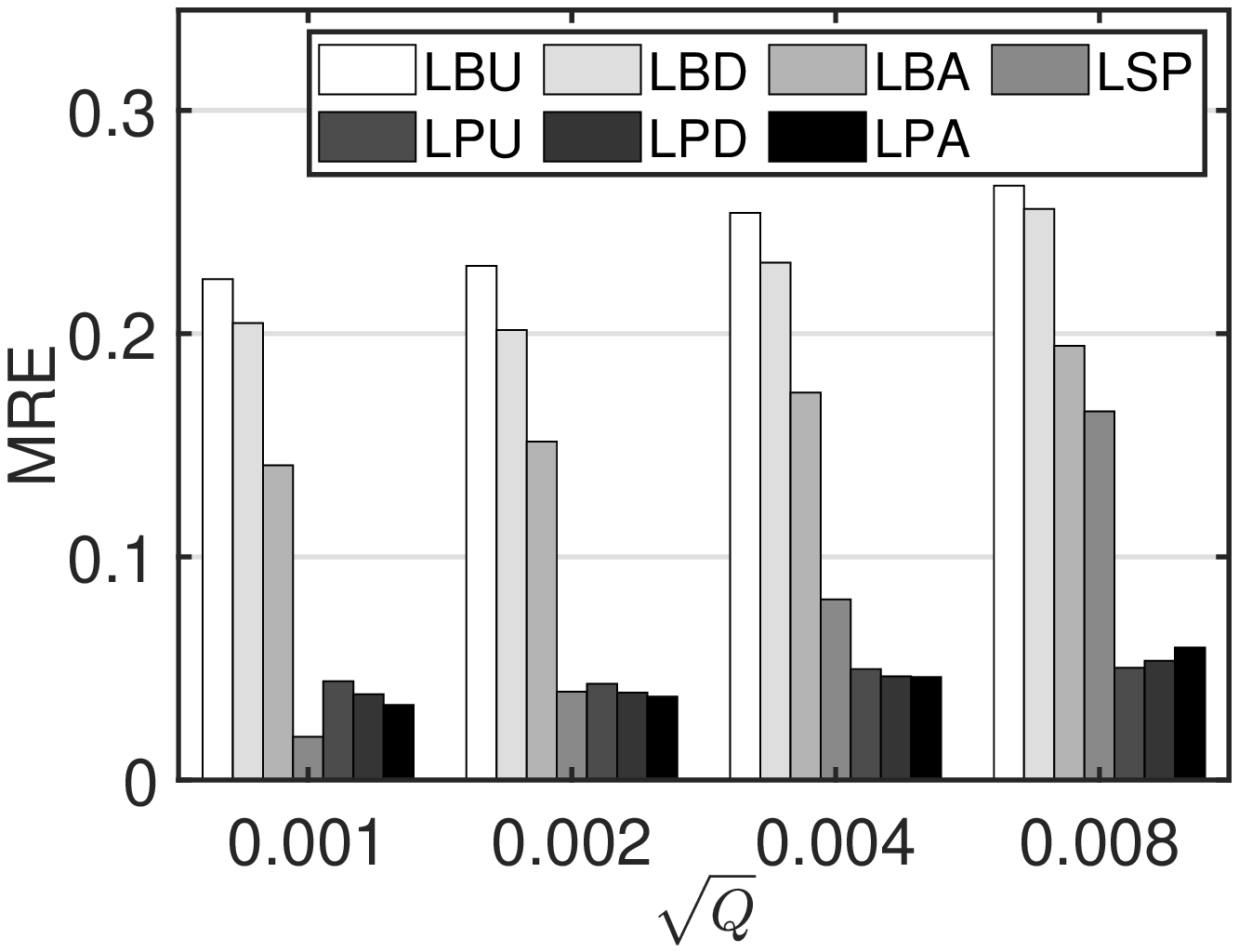}}
	\subfigure[Varying fluctuation $b$ (\textsf{Sin})]{
		\label{LMREq-S} 
		\includegraphics[width=115pt,height=80pt]{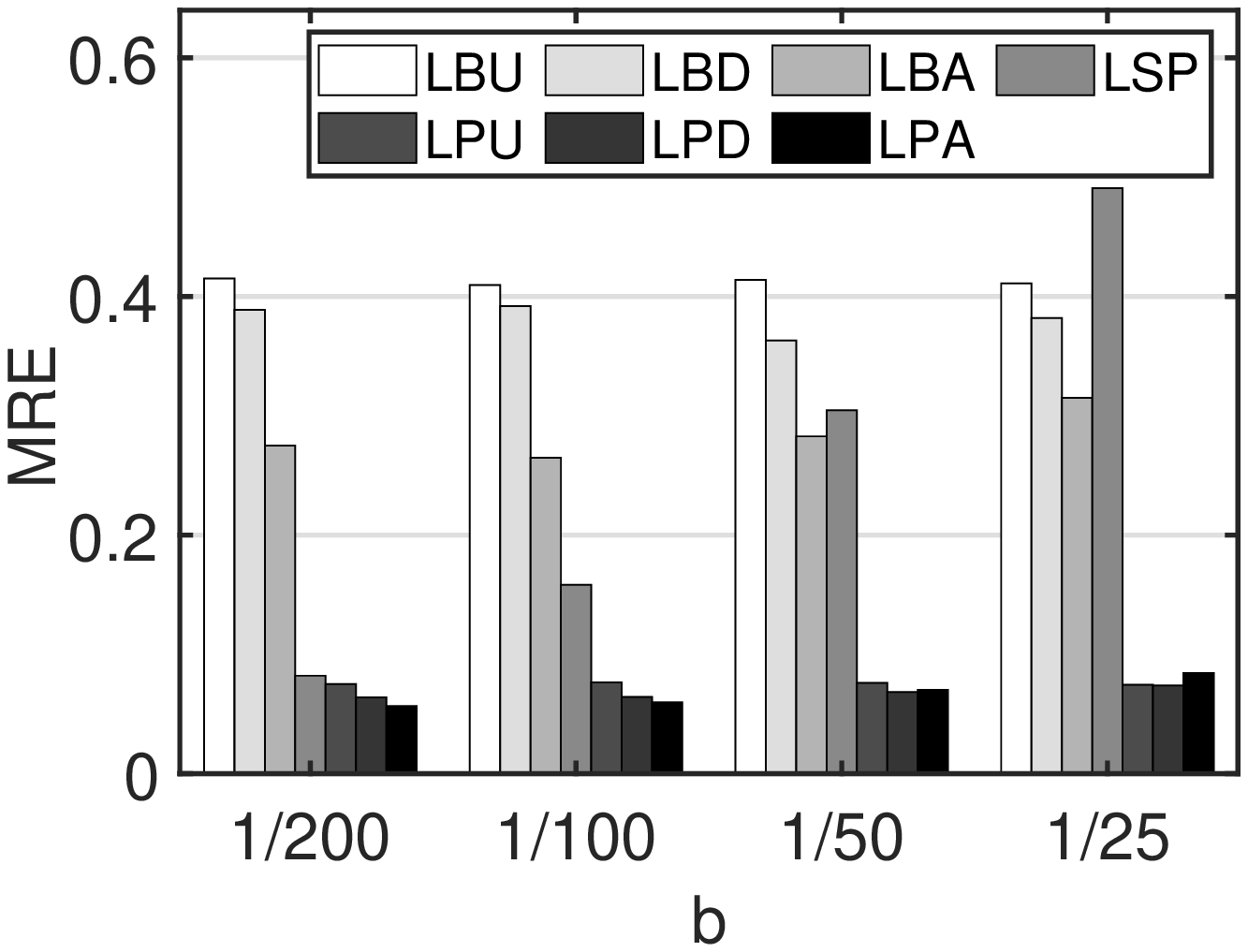}}
	\vspace{-4mm} 
	\caption{Impact of dataset parameters ($\epsilon=1$, $w=30$)} 
	\label{fig:ErrorwLDP} 
\end{figure}

\subsection{Event Monitoring}

The overall distance metric of MRE on the whole stream cannot reflect the estimate accuracy at individual timestamps. Instead, event monitoring is commonly used in data streams to detect whether the estimate at each timestamp is larger than a given threshold $\delta$. 
Fig. \ref{fig:ROC} displays the ROC curves for detecting the above-threshold points on all six dataset. On synthetic binary datasets \textsf{LNS}, \textsf{Sin}, \textsf{Log}, $\delta$ was directly set as $0.75 \times (\max(\bf{c})-\min(\bf{c}))+\min(\bf{c})$. On other three non-binary real-world datasets, we monitored the mean-value $\bm{c}_\text{mean}$ of the histogram $\bm{c}_t$ and $\delta$ was set as $0.75 \times (\max(\bf{c}_\text{mean})-\min(\bf{c}_\text{mean}))+\min(\bf{c}_\text{mean})$.
Overall, the population division methods generally perform better than the budget division method \textsf{LBA}, as they can achieve higher accuracy in above-threshold value detection. Despite varying on different datasets, \textsf{LPD} and \textsf{LPA} in general outperform the other three methods. 
Although \textsf{LSP} manages to have much smaller MRE in Figs.~\ref{fig:ErrorEpsilon} and \ref{fig:ErrorWindow}, it generally performs the worst on most datasets. This is because too many approximations are adopted in \textsf{LSP}, which hinders its efficiency in detecting real-time changes. 

\begin{figure}[hbtp]
	\centering
	\subfigure[\textsf{LNS}]{
		\label{ROC-L} 
		\includegraphics[width=110pt,height=75pt]{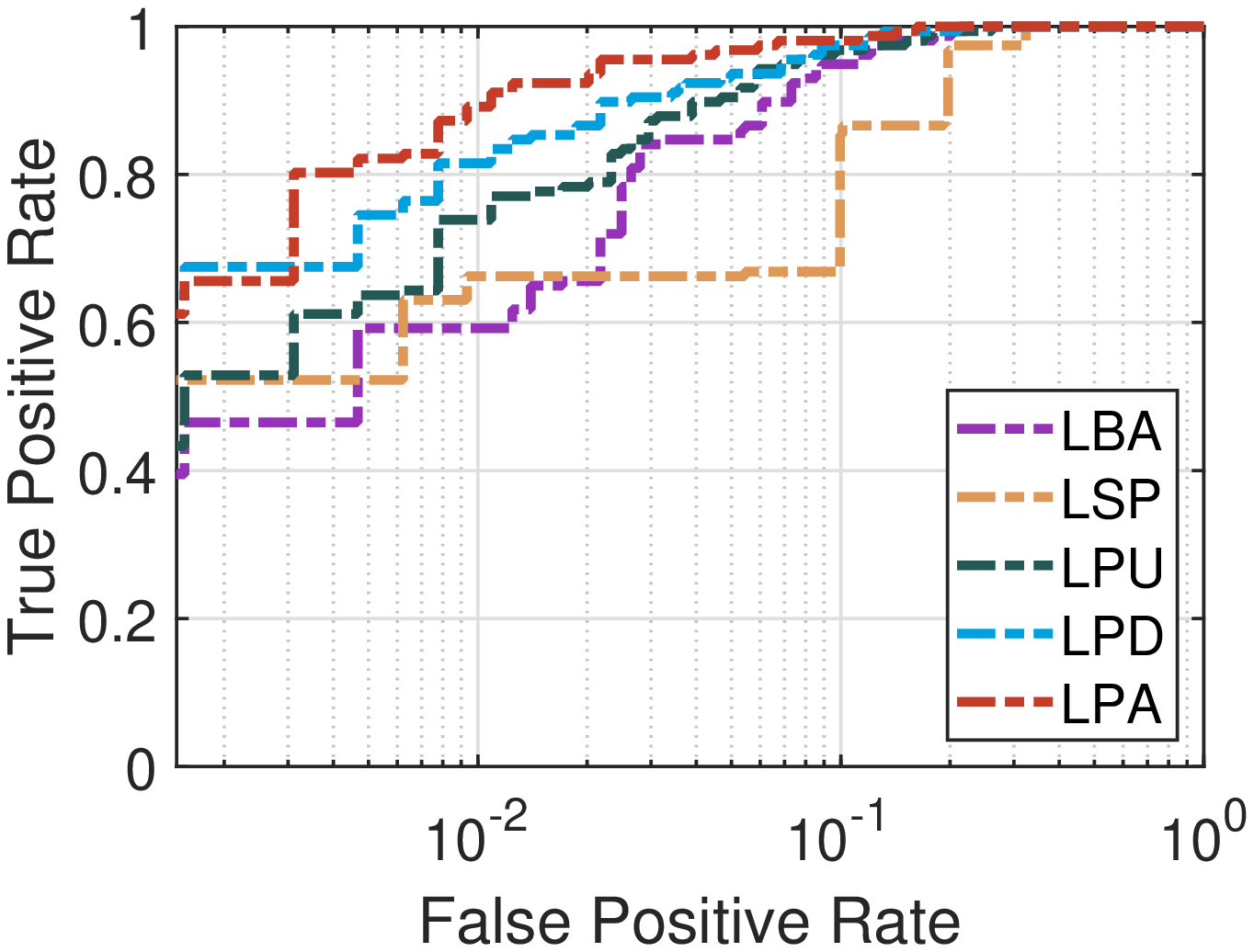}}
	\subfigure[\textsf{Sin}]{
		\label{ROC-F} 
		\includegraphics[width=110pt,height=75pt]{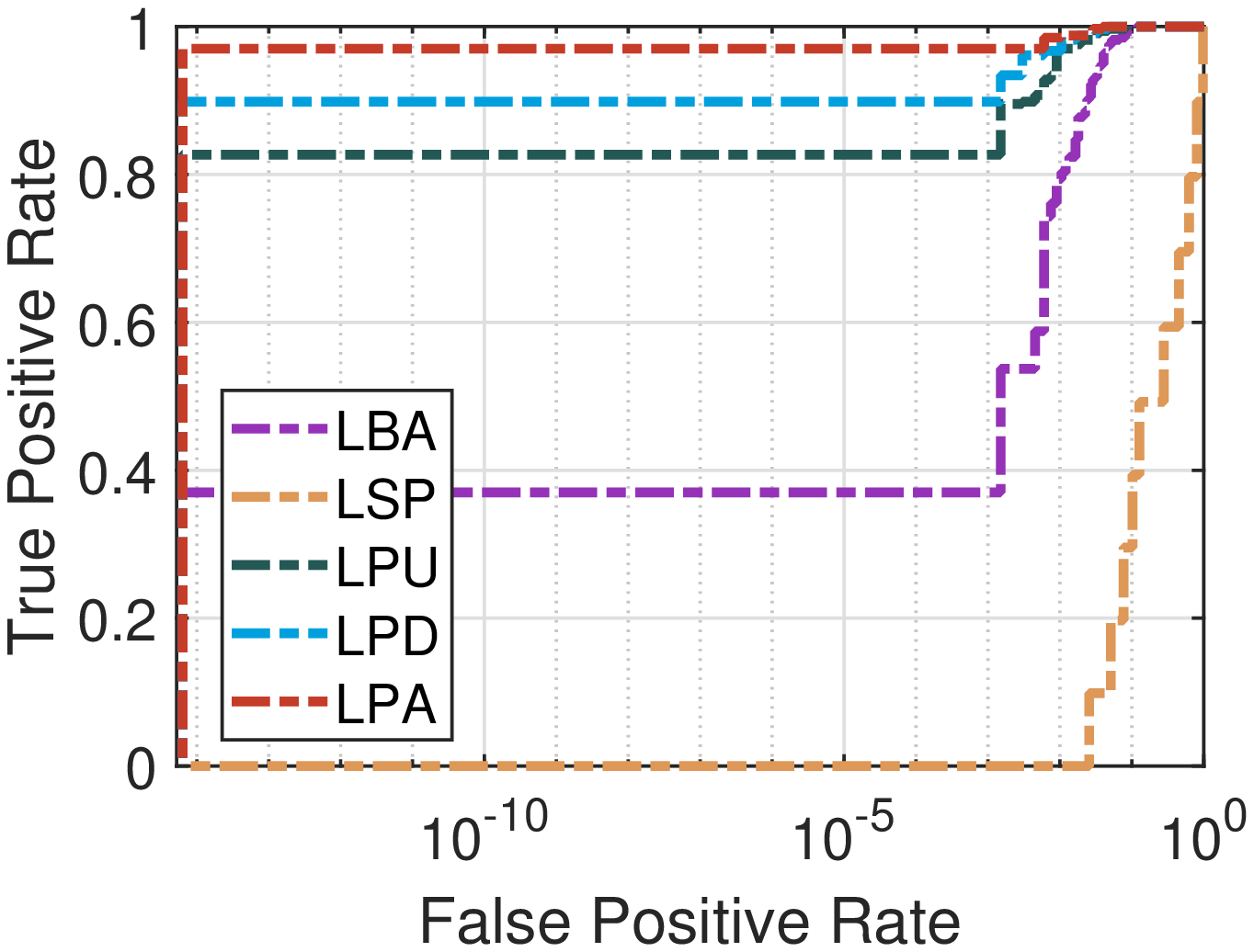}} \\[-10pt] 
	\subfigure[\textsf{Log}]{
		\label{ROC-T} 
		\includegraphics[width=110pt,height=75pt]{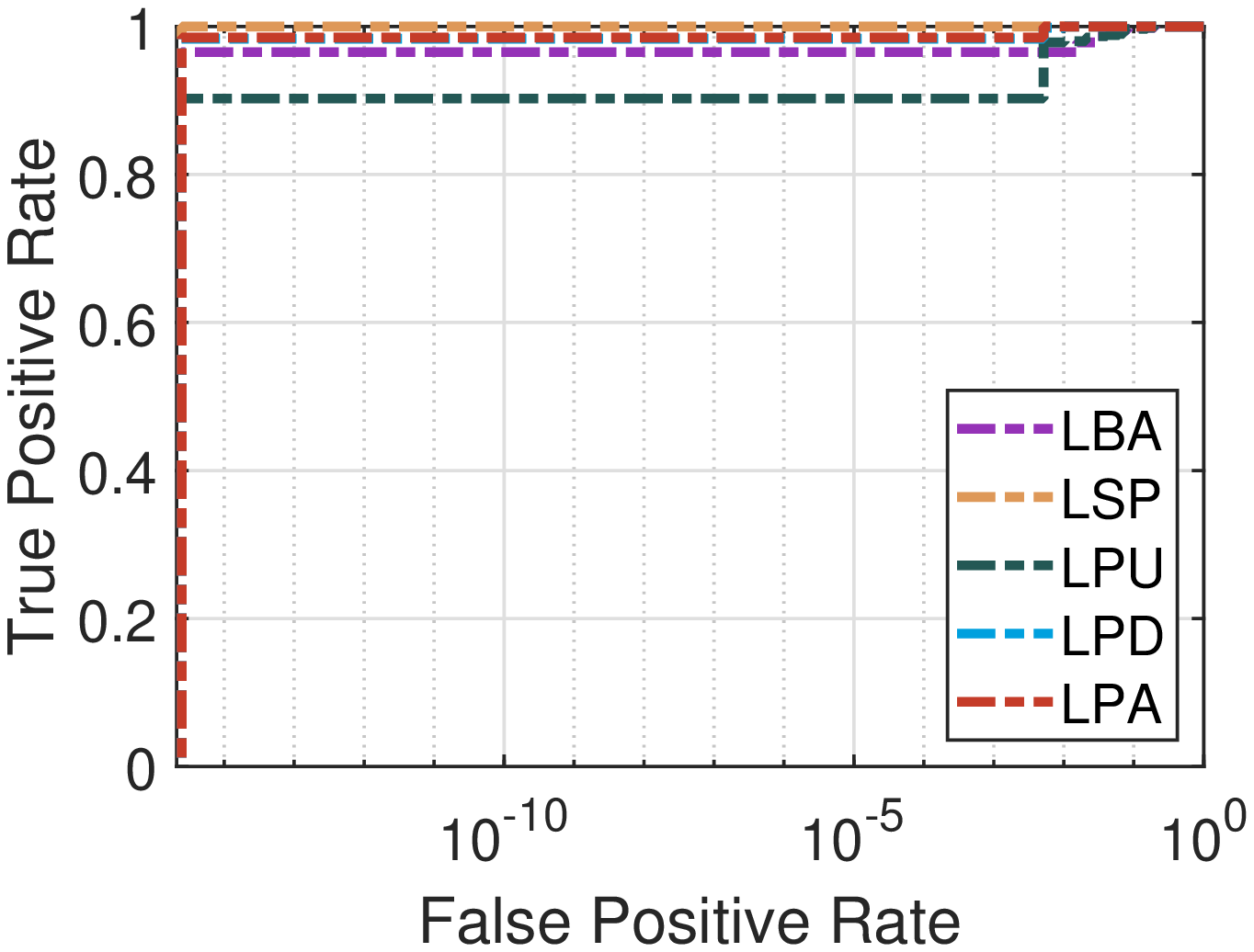}} 
	\subfigure[\textsf{Taxi}]{
		\label{ROC-T} 
		\includegraphics[width=110pt,height=75pt]{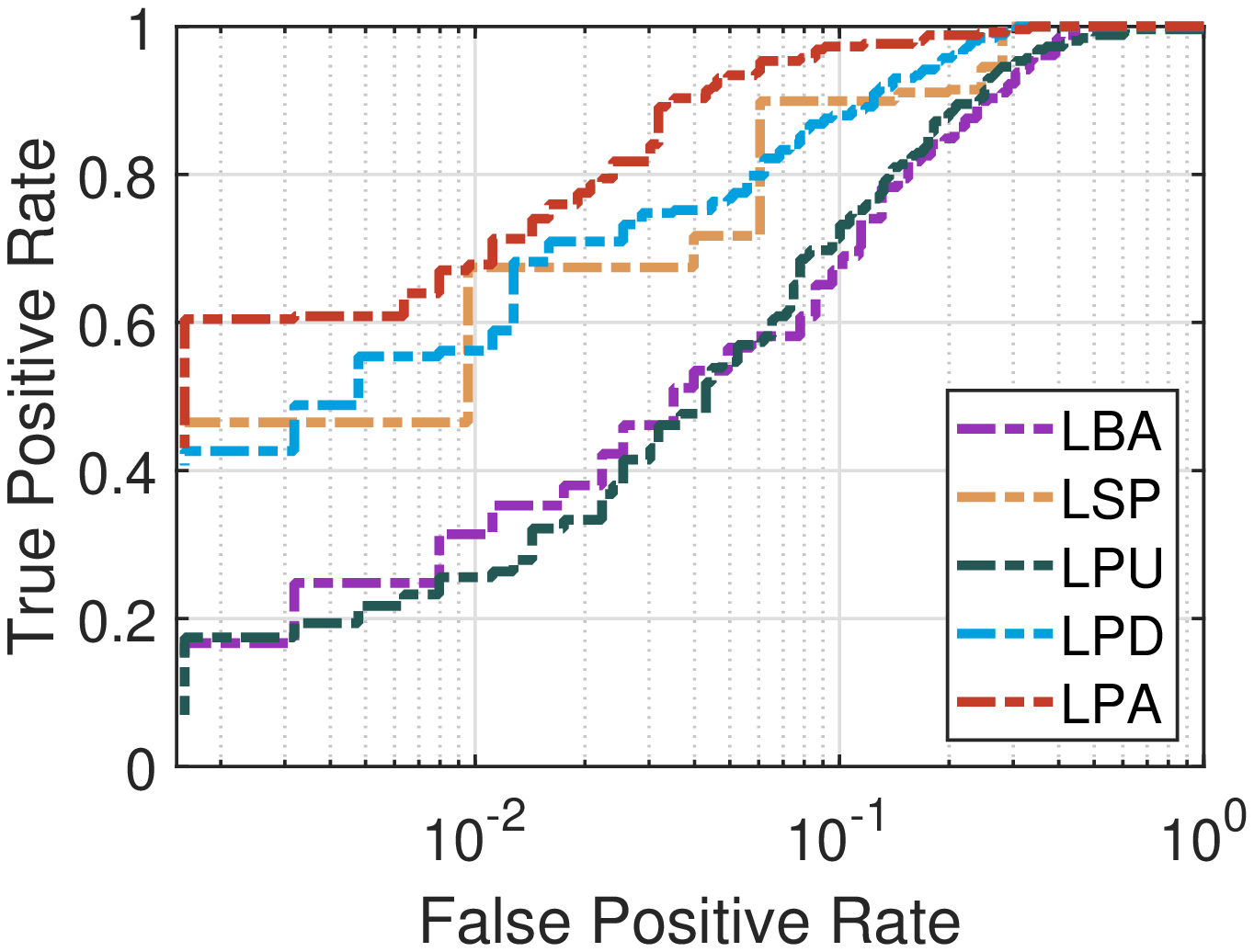}} \\[-10pt] 
	\subfigure[\textsf{Foursquare}]{
		\label{ROC-F} 
		\includegraphics[width=110pt,height=75pt]{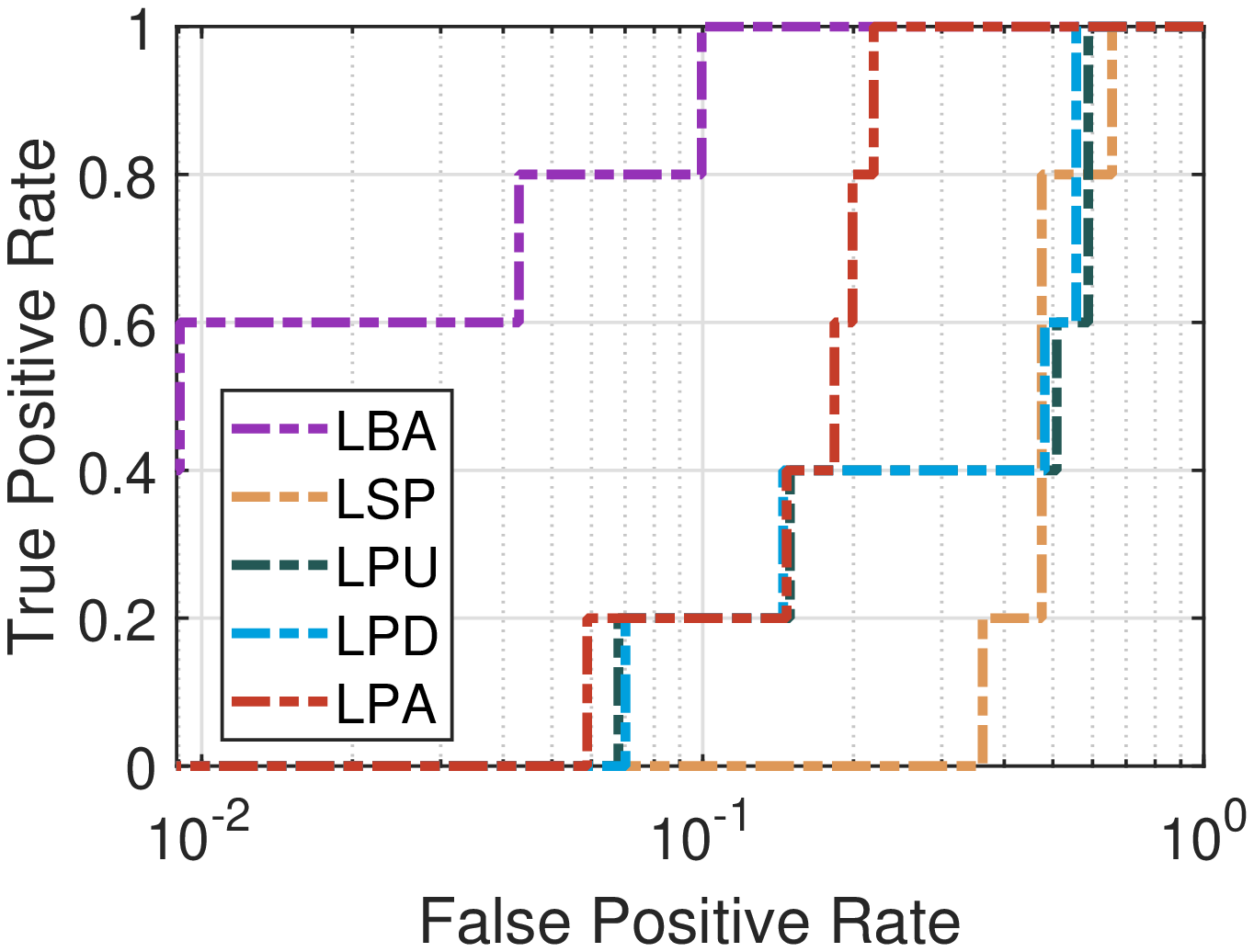}} 
	\subfigure[\textsf{Taobao}]{
		\label{ROC-T} 
		\includegraphics[width=110pt,height=75pt]{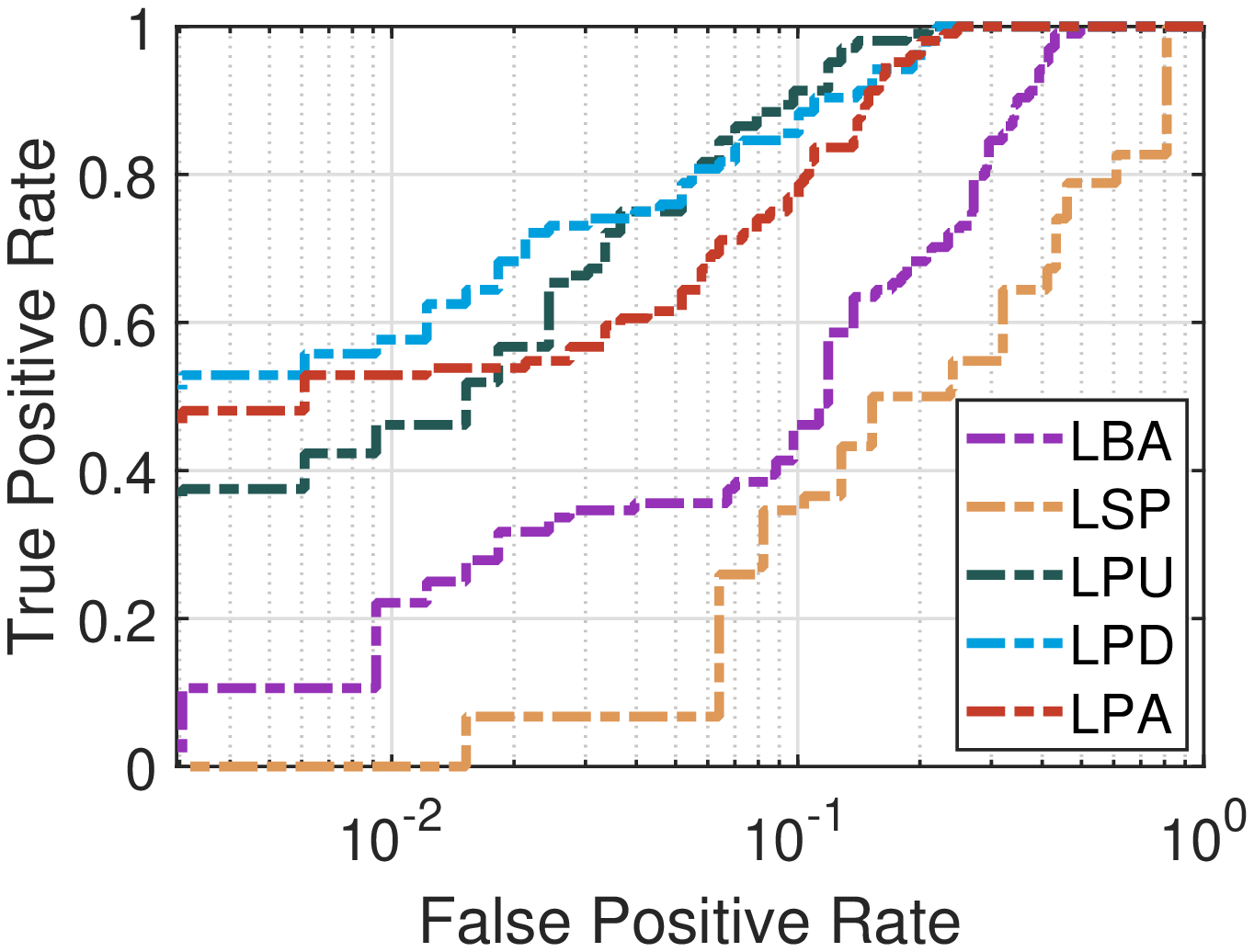}}
	\vspace{-4mm}
	\caption{ROC Curve for Event Monitoring ($\epsilon=1$, $w=50$)} 
	\label{fig:ROC} 
\end{figure}

\subsection{Communication Efficiency}

Fig. \ref{fig:communication} compares the average communication frequency per user (CFPU) of different methods, with different parameters, on \textsf{LNS}. 

\begin{figure}[htbp]
	\centering
	\subfigure[CFPU wrt. $N$]{
		\label{LNS Com N} 
		\includegraphics[width=115pt,height=80pt]{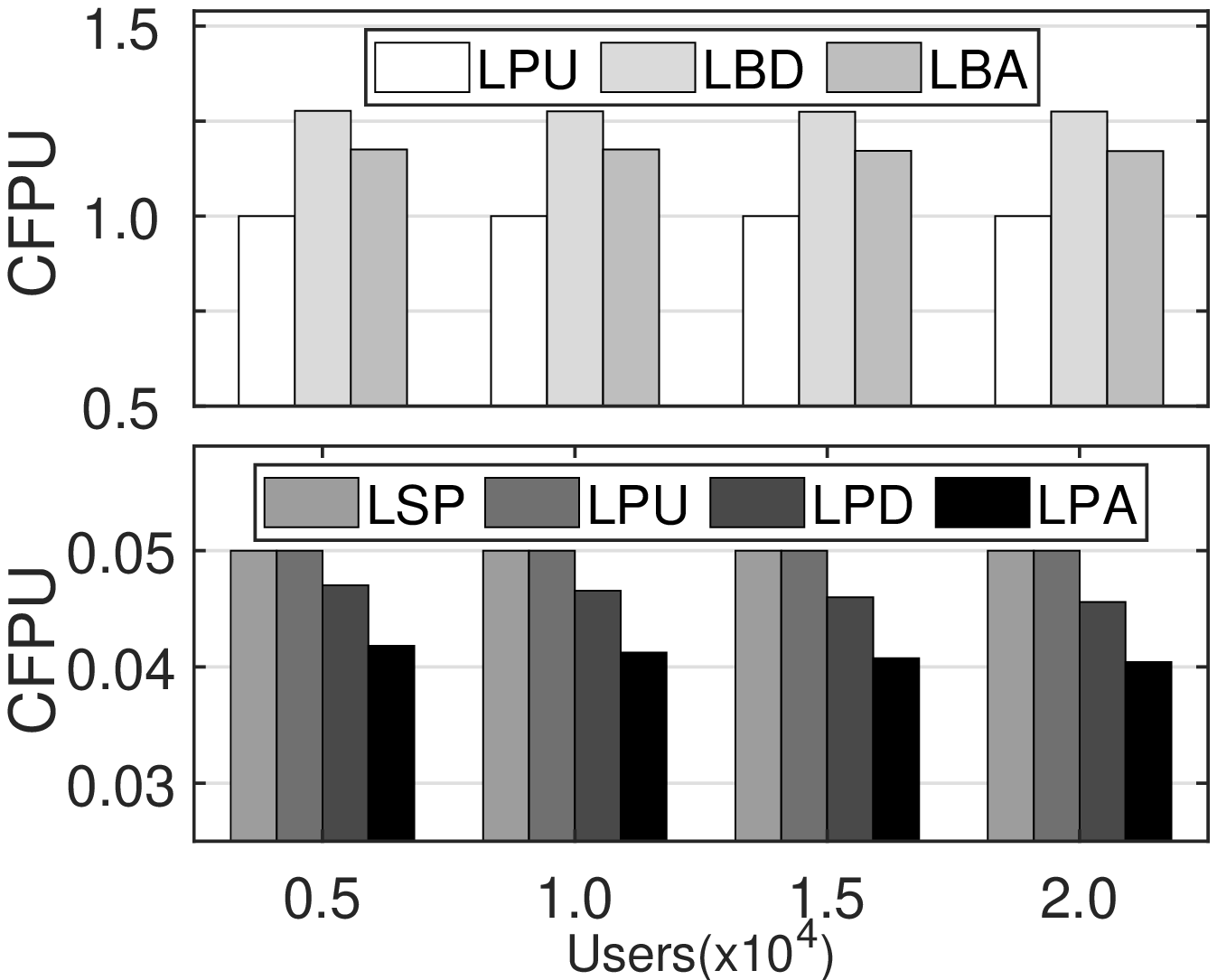}}
	\subfigure[CFPU wrt. $Q$]{
		\label{LNS Com Q} 
		\includegraphics[width=115pt,height=80pt]{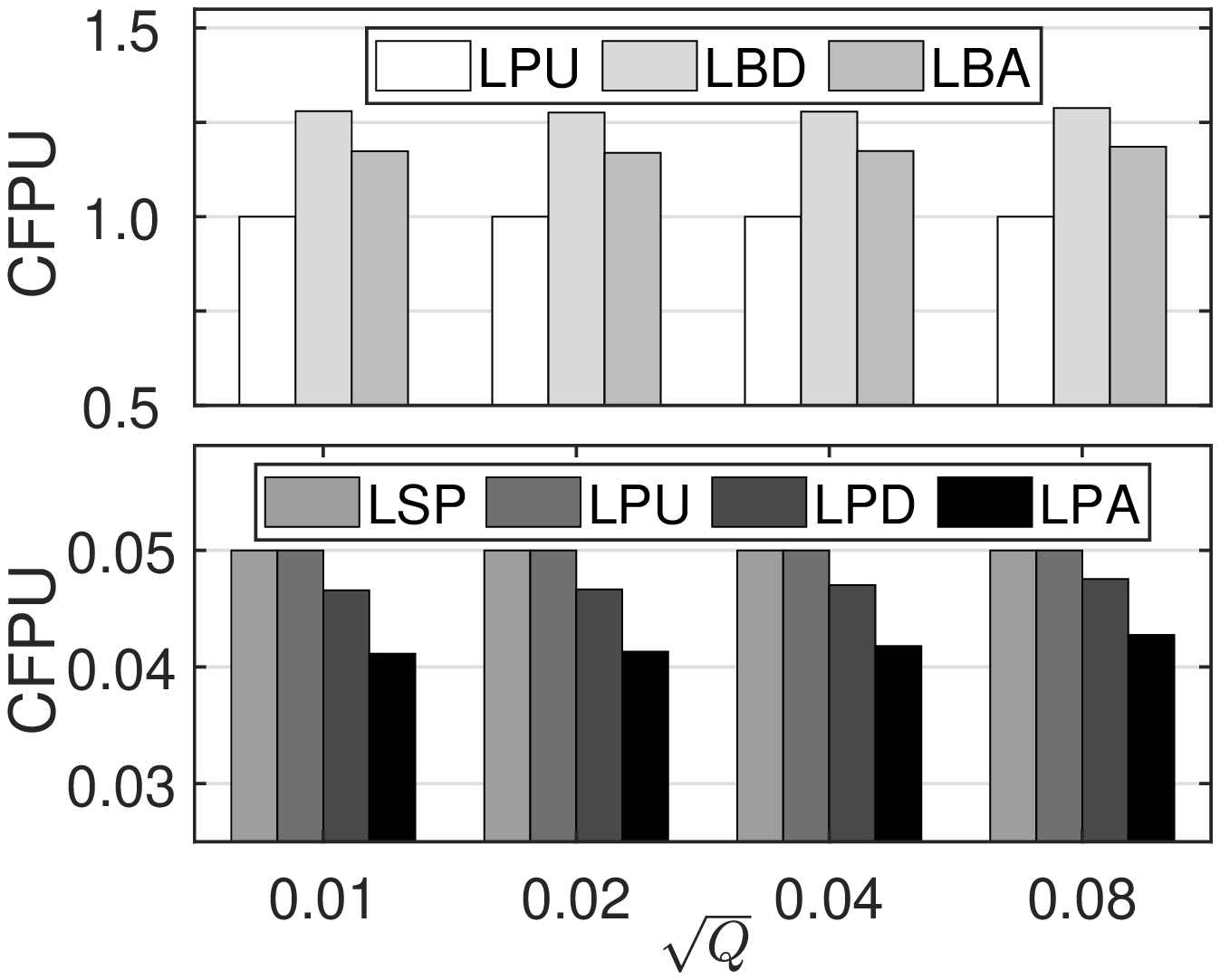}}\\[-10pt] 
	\subfigure[CFPU wrt. $\epsilon$]{
		\label{LNS Com e} 
		\includegraphics[width=115pt,height=80pt]{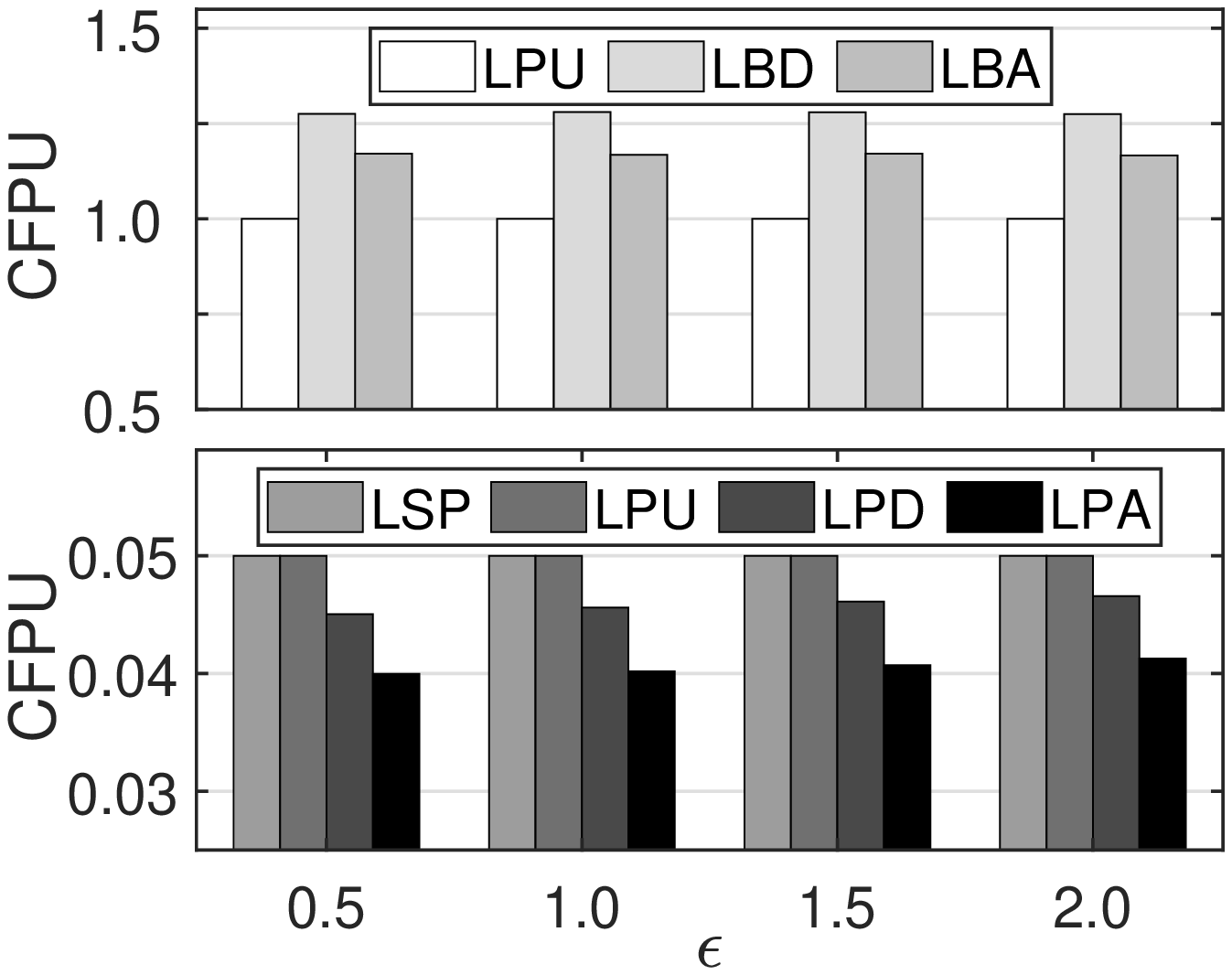}} 
	\subfigure[CFPU wrt. $w$]{
		\label{LNS Com w} 
		\includegraphics[width=115pt,height=80pt]{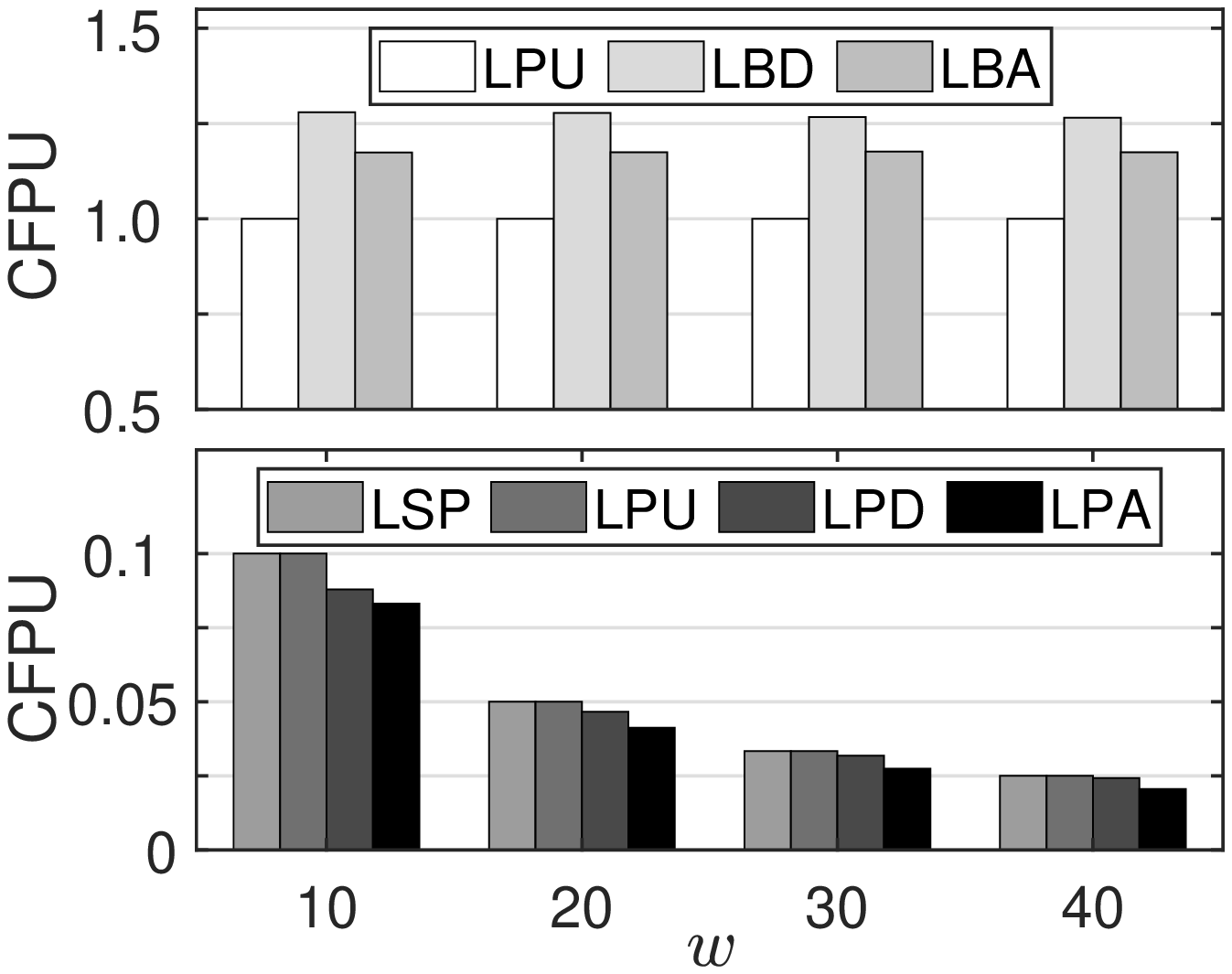}} 
	\vspace{-4mm} \caption{CFPU with respect to different parameters (\textsf{LNS})}\centering
	\label{fig:communication} 
\end{figure}

Fig.~\ref{LNS Com N} depicts the impact of population $N$ on CFPU. For budget division methods, the average CFPU is above $1.0$ since each user has to report at least once at each timestamp. For \textsf{LBU}, each user just reports once at each timestamp. For \textsf{LBD} and \textsf{LBA}, some users may report twice for both $\mathcal{M}_1$ and $\mathcal{M}_2$ at each timestamp. Differently, for population division methods, the CFPU is significantly less since only a small portion of users contribute data at each timestamp. In both \textsf{LSP} and \textsf{LPU}, each user only reports once in a window of $w=20$ timestamps, with the average chance of $1/w=0.05$ per timestamp. In \textsf{LPD} and \textsf{LPA}, users are adaptively divided and chosen to report according to data density of streams, therefore, the CFPU can be effectively reduced. 

Fig.~\ref{LNS Com Q} shows CFPU with respect to data variance $Q$. Similarly, $Q$ has no impact on data-independent methods \textsf{LBU}, \textsf{LSP} and \textsf{LBU}. However, with the increase of $Q$, data-dependent methods \textsf{LBD/LBA} and \textsf{LPD/LPA} have a larger CFPU, although it is not evident for \textsf{LBD/LBA}. The former is because that, with more fluctuations, data-dependent methods would increase publication frequency to adapt to stream changes. The latter is because that the larger noise in \textsf{LBD/LBA} overwhelms the impact of data fluctuations. 

Fig.~\ref{LNS Com e} presents the CFPU with respect to $\epsilon$. For budget division methods, it generally keeps the same since a slight increase of $\epsilon$ has rather limited impact. However, the CFPU increases with $\epsilon$ in \textsf{LPD} and \textsf{LPA}. The reason is that, with more budget $\epsilon$, the publication error would become smaller and more publications would be chosen in the adaptive methods.

Fig.~\ref{LNS Com w} shows the CFPU, with respect to $w$. Among the budget division methods, \textsf{LBU} is data-independent and keeps unchanged. For \textsf{LBD} and \textsf{LBA}, the CFPU shows a slight decrease with $w$. This is because the publication noise increases with $w$ and the approximation strategy is more favorable. Among the population division methods, besides data-dependent methods \textsf{LPD} and \textsf{LPA}, the CFPU in \textsf{LSP} and \textsf{LPU} also decreases with $w$ because each user reports at each timestamp with a probability of $1/w$. 

Table~\ref{table: communication} further summarizes the CFPU of different methods over other datasets. As shown, 
the population division methods manage to have much less communication overhead. Data-adaptive methods \textsf{LPD} and \textsf{LPA} further reduce the communication cost via exploiting the sparsity in data streams. All above results are consistent with the analysis in Sections~\ref{subsub: communication budget} and \ref{subsub: communication population}.

\begin{table}[htbp]\scriptsize
	\centering
	\caption{CFPU Comparison on All Datasets}\label{table: communication} 
		\vspace{-2mm} \begin{tabular}{|c|c|c|c|c|c|c|}
		\hline
		\multicolumn{2}{|c|}{\multirow{2}{*}{$\epsilon=1, w=20$}} & \multicolumn{2}{c|}{Synthetic Datasets} & \multicolumn{3}{c|}{Real-world Datasets} \\ \cline{3-7} 
		\multicolumn{2}{|c|}{}                          & \textbf{Sin}      & \textbf{Log }    & \textbf{Taxi}    & \textbf{Foursquare } & \textbf{Taobao}  \\ \hline
		\multirow{3}{*}{B}      & \textsf{LBU}      & 1.0000          & 1.0000       & 1.0000  & 1.0000      & 1.0000  \\ \cline{2-7} 
		& \textsf{LBD}      & 1.2719          & 1.2671       & 1.2734  & 1.2733      & 1.2962  \\ \cline{2-7} 
		& \textsf{LBA}      & 1.1709          & 1.1687       & 1.1685  & 1.1775      & 1.1996  \\ \hline
		\multirow{4}{*}{P}  & \textsf{LSP}  & 0.0500          & 0.0500       & 0.0500  & 0.0500      & 0.0500  \\ \cline{2-7} 
		& \textsf{LPU}      & 0.0500          & 0.0500       & 0.0500  & 0.0500      & 0.0500  \\ \cline{2-7} 
		& \textsf{LPD}      & 0.0457          & 0.0457       & 0.0461  & 0.0458      & 0.0467  \\ \cline{2-7} 
		& \textsf{LPA}      & 0.0404          & 0.0403       & 0.0406  & 0.0403      & 0.0418  \\ \hline
		\multicolumn{2}{|c|}{\multirow{2}{*}{$\epsilon=2, w=20$}} & \multicolumn{2}{c|}{Synthetic Datasets} & \multicolumn{3}{c|}{Real-world Datasets} \\ \cline{3-7} 
		\multicolumn{2}{|c|}{}                          & \textbf{Sin}     & \textbf{Log }   & \textbf{Taxi}    & \textbf{Foursquare}  & \textbf{Taobao}  \\ \hline
		\multirow{3}{*}{B}      & \textsf{LBU}      & 1.0000          & 1.0000       & 1.0000  & 1.0000      & 1.0000  \\ \cline{2-7} 
		& \textsf{LBD}      & 1.2800          & 1.2823       & 1.2762  & 1.2692      & 1.3243  \\ \cline{2-7} 
		& \textsf{LBA}      & 1.1731          & 1.1737       & 1.1682  & 1.1704      & 1.2350  \\ \hline
		\multirow{4}{*}{P}  & \textsf{LSP}  & 0.0500          & 0.0500       & 0.0500  & 0.0500      & 0.0500  \\ \cline{2-7} 
		& \textsf{LPU}      & 0.0500          & 0.0500       & 0.0500  & 0.0500      & 0.0500  \\ \cline{2-7} 
		& \textsf{LPD}      & 0.0466          & 0.0468       & 0.0475  & 0.0468      & 0.0475  \\ \cline{2-7} 
		& \textsf{LPA}      & 0.0414          & 0.0413       & 0.0425  & 0.0412      & 0.0434  \\ \hline
		\multicolumn{2}{|c|}{\multirow{2}{*}{$\epsilon=2, w=40$}} & \multicolumn{2}{c|}{Synthetic Datasets} & \multicolumn{3}{c|}{Real-world Datasets} \\ \cline{3-7} 
		\multicolumn{2}{|c|}{}                          & \textbf{Sin}      & \textbf{Log}    & \textbf{Taxi}    & \textbf{Foursquare}  & \textbf{Taobao}  \\ \hline
		\multirow{3}{*}{B}      & \textsf{LBU}      & 1.0000          & 1.0000       & 1.0000  & 1.0000      & 1.0000  \\ \cline{2-7} 
		& \textsf{LBD}      & 1.2643          & 1.2575       & 1.2641  & 1.2487      & 1.2771  \\ \cline{2-7} 
		& \textsf{LBA}      & 1.1729          & 1.1676       & 1.1755  & 1.1670      & 1.2046  \\ \hline
		\multirow{4}{*}{P}  & \textsf{LSP}  & 0.0250          & 0.0250       & 0.0250  & 0.0250      & 0.0250  \\ \cline{2-7} 
		& \textsf{LPU}      & 0.0250          & 0.0250       & 0.0250  & 0.0250      & 0.0250  \\ \cline{2-7} 
		& \textsf{LPD}      & 0.0242          & 0.0245       & 0.0244  & 0.0245      & 0.0245  \\ \cline{2-7} 
		& \textsf{LPA}      & 0.0206          & 0.0207       & 0.0210  & 0.0204      & 0.0214  \\ \hline
	\end{tabular}
\\{*B and P refer to budget division and population division, respectively.}
\end{table}

\section{Conclusion}\label{sec: conclusion}
We propose LDP-IDS, a decentralized privacy-preserving scheme for infinite streaming data collection and analysis. 
We first formalize the definition of $w$-event LDP for infinite data streams.
Then, based on the budget division methodology, we present several baseline methods that can satisfy $w$-event LDP for streaming data collection and analysis. 
Furthermore, we propose a novel framework of population division, which can achieve significant utility improvement and communication reduction for streaming data collection and analysis with LDP. Specifically, considering the non-deterministic sparsity in data streams, two data-adaptive methods are also presented to achieve further utility improvement. 
Through theoretical analysis and experiments with real-world datasets, we demonstrate the superiority of our LDP solutions aganist the budget division-based benchmark methods in terms of estimation accuracy, practical event monitoring efficiency and communication cost. 

\bibliographystyle{ACM-Reference-Format}
\bibliography{Reference}
\normalsize
\appendix

\section{Proofs to Theorems}

\subsection{Proof to Theorem \ref{thm-w-event LDP}}\label{appendice: w-event LDP}
   
%
%

\begin{proof}\label{proof:average w privacy}
	Considering the independent randomness of each mechanism $\mathcal{M}_i$, for stream prefix $V_t$ and a given output transcript $\mathbf{O^*}=(\mathbf{o}_1, \mathbf{o}_2,...,\mathbf{o}_t) \in \mathcal{O}$, there is
	\begin{align}
		\text{Pr}[\mathcal{M} (V_t)\in \mathbf{O^*}]=\prod_{i=1}^t \text{Pr}[\mathcal{M}_i (v_i)=\mathbf{o}_i].
	\end{align} 
	Similarly, for any $w$-neighboring stream prefix $V'_t$ and the same $\mathbf{O^*}=(\mathbf{o}_1, \mathbf{o}_2,...,\mathbf{o}_t)$, there is
	\begin{align}
		\text{Pr}[\mathcal{M} (V'_t)\in \mathbf{O^*}]=\prod_{i=1}^t \text{Pr}[\mathcal{M}_i (v'_i)=\mathbf{o}_i].
	\end{align} 
	According to the definition of $w$-neighboring, there exists $\tau \in [t]$, such that $v_i=v'_i$ for $1 \leq k \leq \tau-w$ and $\tau+1 \leq i \leq t $. Then, there is
	\begin{align}
		\frac{\text{Pr}[\mathcal{M} (V_t) \in \mathbf{O^*}]}{\text{Pr}[\mathcal{M} (V'_t) \in \mathbf{O^*}]}=\prod_{i=\tau-w+1}^\tau \frac{\text{Pr}[\mathcal{M}_i (v_i)=\mathbf{o}_i]}{\text{Pr}[\mathcal{M}_i (v'_i)=\mathbf{o}_i]}
	\end{align}
	Note that database pairs $v_i$ and $v'_i$ are neighboring for $\tau-w+1 \leq i \leq i$, and $\mathcal{M}_i$ satisfies $\varepsilon_{i}$-DP. So, there is $\frac{\text{Pr}[\mathcal{M}_i (v_i)=\mathbf{o}_i]}{\text{Pr}[\mathcal{M}_i (v'_i)=\mathbf{o}_i]} \leq e^{\varepsilon_i}$. Then, we can have
	\begin{align}
		\log\left( \frac{\text{Pr}[\mathcal{M} (V_t)\in \mathbf{O^*}]}{\text{Pr}[\mathcal{M} (V'_t)\in \mathbf{O^*}]} \right) 
		&\leq \log \left(\prod_{i=\tau-w+1}^\tau e^{\epsilon_i} \right) \\ \nonumber
		&=\sum_{i=\tau-w+1}^{\tau} \epsilon_i \\ \nonumber
	\end{align}
	Therefore, for any $\mathbf{O}\in \mathcal{O}$, we have 
	\begin{align}
	\log\left( \frac{\text{Pr}[\mathcal{M} (V_t)\in \mathbf{O}]}{\text{Pr}[\mathcal{M} (V'_t)\in \mathbf{O}]} \right) \leq \sum\limits_{k=\tau-w+1}^\tau \epsilon_k.
	\end{align}
	And, if formula $\sum\limits_{k=\tau-w+1}^\tau \epsilon_k=\epsilon$ holds, then we have $\log\left( \frac{\text{Pr}[\mathcal{M} (V_t)\in \mathbf{O}]}{\text{Pr}[\mathcal{M} (V'_t)\in \mathbf{O}]} \right) \leq \epsilon$, which concludes the proof.
\end{proof}

\subsection{Proof to Theorem \ref{theorem: dissimilarity}}\label{appendice: dis calculation}
\begin{proof}
The LDP guarantee can be directly obtained according to the post-processing theorem~\cite{Dwork2014TheAF}. In the following, we mainly prove $dis$ is an unbiased estimation of $dis^*$.

For any $k$ that $1 \leq k \leq d$, since $\overline{\mathbf{c}}_{t,1}[k]$ is an unbiased estimate of $\mathbf{c}_t[k]$, we can have
\begin{align}\label{eq: unbiase}
&\mathbb{E}[\overline{\mathbf{c}}_{t,1}[k]]=\mathbf{c}_t[k]\\ &\text{Var}(\overline{\mathbf{c}}_{t,1}[k])=\mathbb{E}(\overline{\mathbf{c}}_{t,1}[k]-\mathbf{c}_t[k])^2
\end{align}
We rewrite the variance formula as
\begin{align}	
\nonumber
&\text{Var}(\overline{\mathbf{c}}_{t,1}[k])=\mathbb{E}(\overline{\mathbf{c}}_{t,1}[k]-\mathbf{c}_t[k])^2 \\
\nonumber
=&\mathbb{E}((\overline{\mathbf{c}}_{t,1}[k]-\mathbf{r}_l[k])-(\mathbf{c}_t-\mathbf{r}_l[k]))^2\\ \nonumber	
=&\mathbb{E}[(\overline{\mathbf{c}}_{t,1}[k]-\mathbf{r}_l[k])^2+(\mathbf{c}_t[k]-\mathbf{r}_l[k])^2] \\&-2\mathbb{E}[(\overline{\mathbf{c}}_{t,1}[k]-\mathbf{r}_l[k])\cdot (\mathbf{c}_t[k]-\mathbf{r}_l[k])]
\end{align}
Both $\mathbf{c}_t[k]$ and $\mathbf{r}_l[k]$ are constant value, then the above equation can be further written as
\begin{align}\label{eq: derivation}
\nonumber
&\text{Var}(\overline{\mathbf{c}}_{t,1}[k])=\mathbb{E}(\overline{\mathbf{c}}_{t,1}[k]-\mathbf{c}_t[k])^2 \\
\nonumber
=&\mathbb{E}(\overline{\mathbf{c}}_{t,1}[k]-\mathbf{r}_l[k])^2+(\mathbf{c}_t[k]-\mathbf{r}_l[k])^2-2(\mathbf{c}_t[k]-\mathbf{r}_l[k])^2\\
=&\mathbb{E}(\overline{\mathbf{c}}_{t,1}[k]-\mathbf{r}_l[k])^2-(\mathbf{c}_t[k]-\mathbf{r}_l[k])^2.
\end{align}
Hence, there is
\begin{align}\label{eq: derivation}
\mathbb{E}(\overline{\mathbf{c}}_{t,1}[k]-\mathbf{r}_l[k])^2
=(\mathbf{c}_t[k]-\mathbf{r}_l[k])^2+\text{Var}(\overline{\mathbf{c}}_{t,1}[k])
\end{align}
Therefore, the expectation of $dis$ in Eq~(\ref{eq: estimation}) satisfies
	\begin{align}
	\nonumber
	&\mathbb{E}{(dis)}\\ \nonumber
	=&\mathbb{E}{\left(\frac{1}{d} \sum_{k=1}^d  (\overline{\mathbf{c}}_{t,1}[k]-\mathbf{r}_l[k])^2-\frac{1}{d} \sum_{k=1}^d \text{Var}(\overline{\mathbf{c}}_{t,1}[k])\right)} \\ \nonumber
	=&\frac{1}{d} \sum_{k=1}^d  \mathbb{E}(\overline{\mathbf{c}}_{t,1}[k]-\mathbf{r}_l[k])^2-\frac{1}{d} \sum_{k=1}^d \text{Var}(\overline{\mathbf{c}}_{t,1}[k]) \\ \nonumber	
	=&\frac{1}{d} \sum_{k=1}^d \left( (\mathbf{c}_t[k]-\mathbf{r}_l[k])^2+\text{Var}(\overline{\mathbf{c}}_{t,1}[k]) \right) -\frac{1}{d} \sum_{k=1}^d \text{Var}(\overline{\mathbf{c}}_{t,1}[k]) \\ \nonumber
	=&\frac{1}{d} \sum_{k=1}^d (\mathbf{c}_t[k]-\mathbf{r}_l[k])^2=dis^*
	\end{align}
\end{proof}

\subsection{Proof to Theorem \ref{theorem: LBD privacy}}\label{appendice: privacy proof BD/BA}

\begin{proof}
	We prove the privacy guarantee of \textsf{LBD} and \textsf{LBA} as followings. 
	
	\textbf{(1) \textsf{LBD} satisfies $w$-event LDP:}
	
	In sub mechanism $\mathcal{M}_{1}$, the dissimilarity budget $\epsilon_{t,1}$ at each timestamp $t$ is $\epsilon/(2w)$. Then, for every $t$, there is 
	\begin{align}
		\sum_{k=t-w+1}^t \epsilon_{k,1}=\epsilon/2.
	\end{align}
	
	In sub mechanism $\mathcal{M}_{2}$, 
	at each timestamp $t$, at most half of the remaining publication budget is allocated if publication occurs. That is to say, $\epsilon_{t,2}=(\epsilon/2-\sum_{k=t-w+1}^{t-1} \epsilon_{k,2})/2$.

	Firstly, for any $1 \leq t \leq w$, \textsf{LBD} distributes the budget in a sequence of $\epsilon/4, \epsilon/8, \ldots$,
	there will be at most $w$ publications as a time window consisting of $w$ timestamps. So, 
	\begin{align}
		\sum_{i = 1}^t \epsilon_{k,2} \leq (\epsilon/2) \cdot (1-\frac{1}{2^w})\leq \epsilon/2.
	\end{align}
	
	Suppose that $\sum_{k=t-w+1}^t \epsilon_{k,2} \leq \epsilon/2$ holds for $t=w+m$, i.e., $\sum_{k=m+1}^{w+m} \epsilon_{k,2} \leq \epsilon/2$. Then, at timestamp $t=w+m+1$, there is
	\begin{align}\label{eq: m+2}
		\sum_{k=m+2}^{w+m+1} \epsilon_{k,2}=\sum_{k=m+2}^{w+m} \epsilon_{k,2}+\epsilon_{w+m+1, 2}
	\end{align}  
Since $\epsilon_{w+m+1, 2}$ is half of the remaining publication budget at time $w+m+1$ , there is 
\begin{align}\label{eq: w+m+1}
	\epsilon_{w+m+1, 2} \leq (\epsilon/2-\sum_{m+2}^{w+m} \epsilon_{k,2})/2.
\end{align}
By substituting Eq.~\ref{eq: w+m+1} in to Eq.~\ref{eq: m+2}, there is
	\begin{align}\label{eq: m+2}
	\sum_{k=m+2}^{w+m+1} \epsilon_{k,2}=&\sum_{k=m+2}^{w+m} \epsilon_{k,2}+(\epsilon/2-\sum_{m+2}^{w+m} \epsilon_{k,2})/2 \\ \nonumber
	=&\epsilon/4+(\sum_{k=m+2}^{w+m} \epsilon_{k,2})/2 \leq \epsilon/4+\epsilon/4 =\epsilon/2.
\end{align}  
This implies that, if $\sum_{k=t-w+1}^t \epsilon_{k,2} \leq \epsilon/2$ holds for $t=w+m$, then it also holds for $t=w+m+1$. Besides, $\sum_{k=t-w+1}^t \epsilon_{k,2} \leq \epsilon/2$ always holds for $1\leq t \leq w$. Therefore, we can prove that, for every timestamp $t \geq 1$, there is 
\begin{align}
	\sum_{k=t-w+1}^t \epsilon_{k,2} \leq \epsilon/2.
\end{align}

Because \textsf{LBD} executes $\mathcal{M}_{t,1}$ and $\mathcal{M}_{t,2}$ sequentially at each timestamp $t$, the total privacy budget in a window of size $w$ should be
\begin{align}
	\sum_{k=t-w+1}^t \epsilon_{k}=\sum_{k=t-w+1}^t \epsilon_{k,1}+\sum_{k=t-w+1}^t \epsilon_{k,2} \leq \epsilon,
\end{align}
which proves that \textsf{LBD} satisfies $w$-event $\epsilon$-LDP.

\textbf{(2) \textsf{LBA} satisfies $w$-event LDP:}

The sub mechanism $\mathcal{M}_{t,1}$ in \textsf{LBA} is identical to that in \textsf{LBD}. That is, for every $t$, there is 
\begin{align}
	\sum_{k=t-w+1}^t \epsilon_{k,1}=\epsilon/2.
\end{align}

The sub mechanism $\mathcal{M}_{t,2}$ in \textsf{LBA} is $\epsilon_{t,2}$-LDP where $\epsilon_{t,2}$ may be nullified, or be absorbed, or absorb unused budgets from previous timestamps. 

Without loss of generality, we suppose $t$ be a publication timestamp that absorbed the unused budget from the last $\alpha$ timestamps. Then, according to Algorithm~\ref{Alg BA}, the publication budget $\epsilon_{t,2}$ at current timestamp equals to $(1+\alpha)\cdot \frac{\epsilon}{2\cdot w}$, but the publication budget at both the preceding $\alpha$ timestamps ( $i \in [t-\alpha, i-1]$ which are absorbed) and the succeeding $\alpha$ timestamps ( $i \in [t+1, t+\alpha]$ which are nullified), will be $0$, i.e., $\epsilon_{i,2}=0$.

Then, any window of size $w$ sliding over timestamp $i$, must cover at least $\alpha$ timestamps with $\epsilon_{i,2}=0$. Suppose that, there are $n$ timestamps having $\epsilon_{i,2}=0$ due to the absorption or nullfication by timestamp $t$. Then, the sum of budget of publication timestamp $t$ and the $n$ zero-budget timestamps is at most $(1+\alpha)\cdot \frac{\epsilon}{2\cdot w}$. This is equivalent to the case where each of these $n+1$ $timestamps$ is assigned with uniform budget of $(1+\alpha)\cdot \frac{\epsilon}{2\cdot w \cdot (n+1)} \leq \frac{\epsilon}{2\cdot w}$. This also holds for other publication timestamp $t'$ that absorbed its previous unused budget in the same window as $t$. Then, the total publication budget in a time window of size $w$, summing up the non-zero budget at publication timestamps and zero-budget at nullified and absorbed timestamps, would be $\sum_{k=t-w+1}^t \epsilon_{k,2} \leq \frac{\epsilon}{2\cdot w} \cdot w= \epsilon/2$.

Similarly, as  \textsf{LBA} also executes $\mathcal{M}_{t,1}$ and $\mathcal{M}_{t,2}$ sequentially at each timestamp $t$, the total privacy budget in a window of size $w$ should be
\begin{align}
	\sum_{k=t-w+1}^t \epsilon_{k}=\sum_{k=t-w+1}^t \epsilon_{k,1}+\sum_{k=t-w+1}^t \epsilon_{k,2} \leq \epsilon,
\end{align}
which proves that \textsf{LBA} satisfies $w$-event $\epsilon$-LDP.

\end{proof}

\subsection{Proof to Lemma~\ref{lemma: mse comparison}}\label{appendice: UU and UP}

\begin{proof}
	We compare the MSE of \textsf{LBU} (budget division framework) and that of \textsf{LPU} (population division framework) with the same FO, e.g., GRR, which are denoted as $ \text{MSE}_\text{LBU+GRR} $ and $\text{MSE}_\text{LPU+GRR}$ respectively.
	
	With GRR protocol, there is 
	\begin{align}
		\text{MSE}_\text{LBU+GRR} =V_\text{GRR} (\epsilon/w,N)=\dfrac{d-2+e^{\epsilon/w}}{N\cdot (e^{\epsilon/w}-1)^2}
	\end{align}
	and there is
	\begin{align}
		\text{MSE}_\text{LPU+GRR} =V_\text{GRR} (\epsilon,N/w)=w\cdot \dfrac{d-2+e^{\epsilon}}{ N \cdot(e^{\epsilon}-1)^2}
	\end{align}
	Then, there is 
	\begin{align*}
		&\text{MSE}_\text{LBU+GRR}-\text{MSE}_\text{LPU+GRR} \\
		= &\frac{1}{N} \left[ \dfrac{d-2+e^{\epsilon/w}}{(e^{\epsilon/w}-1)^2} - w \dfrac{d-2+e^{\epsilon}}{(e^{\epsilon}-1)^2} \right] \\
		=& \frac{d-2}{N} \left[ \dfrac{1}{(e^{\epsilon/w}-1)^2} - \dfrac{w}{(e^{\epsilon}-1)^2}  \right]+ \\
		&   \frac{e^{\epsilon/w}}{N}\left[ \dfrac{1}{(e^{\epsilon/w}-1)^2} - \dfrac{we^{\epsilon-\epsilon/w}}{(e^{\epsilon}-1)^2} \right] \\
		=& \dfrac{(d-2)}{N(e^{\epsilon/w}-1)^2(e^{\epsilon}-1)^2} \left[ (e^{\epsilon}-1)^2-w(e^{\epsilon/w}-1)^2 \right]+ \\
		&  \dfrac{e^{\epsilon/w}}{N(e^{\epsilon/w}-1)^2(e^{\epsilon}-1)^2} \left[ (e^{\epsilon}-1)^2-we^{\epsilon-\epsilon/w}(e^{\epsilon/w}-1)^2 \right]
	\end{align*}
	
	Simply, we denote $e^{\epsilon/w}$ as $z$. Since $\epsilon>0$, $z>1$. Then, we have
	\begin{align*}
		&(e^{\epsilon}-1)^2-w((e^{\epsilon/w}-1)^2)\\
		=&(z^w-1)^2-w(z-1)^2\\
		=&(z-1)^2[(1+z^2+...+z^{w-1})^2-w]>0
	\end{align*}
	and there is	
	\begin{align*}
		&(e^{\epsilon}-1)^2-we^{\epsilon-\epsilon/w}(e^{\epsilon/w}-1)^2\\
		=&(z^w-1)^2-wz^{w-1}(z-1)^2\\
		=&(z-1)^2[(1+z^2+...+z^{w-1})^2-wz^{w-1}]>0
	\end{align*}
	
	Besides, $d \ge 2$, therefore, $\text{MSE}_\text{LBU+GRR}-\text{MSE}_\text{LPU+GRR}>0$. That is to say, $\text{MSE}_\text{LBU+GRR}>0$ is always smaller than $\text{MSE}_\text{LPU+GRR}$. 
	
	The proof details under the OUE protocol are similar, and therefore omitted here. So, given the same FO protocol GRR or OUE, the MSE of \textsf{LPU} is smaller than that of \textsf{LBU}, i.e., $\text{MSE}_\textsf{LPU}<\text{MSE}_\textsf{LBU}$.
	
	In conclusion, it can achieve much better utility to divide the population, instead of dividing privacy budget, in a time window to achieve $w$-event LDP.
\end{proof}

\subsection{Proof to Theorem~\ref{theorem: privacy analysis}}\label{appendice: privacy proof UD/UA}
\begin{proof}
	We prove this theorem by proving the following claims.
	\begin{enumerate}
		\item In any time window consists of $w$ consecutive timestamps, each user reports to the server at most once. 
		\item Each user's reported data satisfies $\epsilon$-LDP.
	\end{enumerate}
	
	The second claim holds for both \textsf{LPD} and \textsf{LPA} apparently. This is because, each time \textsf{LPD} and \textsf{LPA} request users to report value to the server, the selected users will report via an FO with $\epsilon$ as privacy budget (e.g., Lines 4 and 5 in Algorithm~\ref{Alg UD}, and Line 13 in Algorithm~\ref{Alg UA}).
	
	We prove the first claim by proving that, at each timestamp $t$, the total number of users allocated in a window of size $w$ is no larger than $N$, i.e., $\sum_{k=t-w+1}^t|U_k|<N$. Then, as long as we sample a fresh set of users $U_k$ at each time $k$ and ensure that $U_{k,1}\cap U_{k,2}=\emptyset$, we can guarantee that each user participates only once in a window of size $w$.
	
	For \textsf{LPD}, in $\mathcal{M}_{t,1}$, $\lfloor N/(2w)\rfloor$ users are allocated at each timestamp $t$. So, for every $t$ and $i \in [t]$, there is $\sum_{k=i-w+1}^i|U_{k,1}| \leq N/2$. $\mathcal{M}_{t,2}$ each time either publishes with additional users $U_{t,2}$ or approximates the last release without assigning any user. In the latter case, $|U_{t,2}|$ is simply zero. In the former case, $|U_{t,2}|=(N/2-\sum_{k=i-w+1}^{i-1} |U_{k,2}|)/2$. Particularly, since $\mathcal{M}_{t,2}$ always uses up to half of the available users, there is always. $0 \leq \sum_{k=t-w+1}^t|U_{k,2}| \leq N/2$. Therefore,  for every $t$ and $i \in [t]$, the total number of publication users in a time window of size $w$, should be $\sum_{k=i-w+1}^i|U_{k}|=\sum_{k=i-w+1}^i|U_{k,1}|+\sum_{k=i-w+1}^i|U_{k,2}| \leq N$.
	
	For \textsf{LPA}, similarly, in $\mathcal{M}_{t,1}$, there is $\sum_{k=i-w+1}^i|U_{k,1}| \leq N/2$  for every $t$ and $i \in [t]$.  In $\mathcal{M}_{t,2}$, suppose $i$ is a timestamp which absorbed budget from $m$ preceding timestamps, where $m$ must be smaller than the window size $w$, i.e., $0 \leq m \leq w-1$. Then, at the current timestamp $i$, the population of publication users should be $|U_{t,2}|=\frac{(m+1)N}{2w}$; at the timestamps $i-m \leq k \leq i-1$ and $i+1\leq k \leq i+m$, the publication users are either skipped or nullified, i.e., $|U_{k,2}|=0$. Since, according to \textsf{UA}, any $w$-timestamp-long window that contains timestamp $i$ would also have $l \geq m$ timestamps that were either absorbed or nullified, i.e., $|U_{k,2}|=0$. Therefore, the sum of population of $i$ along with the these $l$ timestamps is at most $(m+1) N/(2w)$, which is at most equal to the case where each of these $l+1$ timestamps receives uniform population $|U_{k,2}|=\frac{(m+1)N/(2w)}{l+1}\leq N/(2w)$. This holds for any timestamp $i'$ that absorbed users from $m'$ previous timestamps and lies in the same window as $i$. Therefore, $\sum_{k=i-w+1}^i |U_{k,2}| \leq \sum_{k=i-w+1}^i N/(2w)=N/2$. As $m \geq 0$, $|U_{k,2}| \geq 0$ holds for every $k$, so, $\sum_{k=i-w+1}^i |U_{k,2}| \geq 0$.
	
	So far, we have proved that at each timestamp $t$, the total number of users allocated in a window of size $w$ is no larger than $N$, i.e., $\sum_{k=t-w+1}^t|U_k|<N$.
	In both \textsf{LPD} and \textsf{LPA}, the sampled users $U_{t,2}$ at each time always come from the remaining available users $U_{A}$ excludes $U_{t,1}$. So, there must be $U_{k,1}\cap U_{k,2}=\emptyset$. Therefore, we can always guarantee that each user participate only once in a window of size $w$. 
\end{proof}

\section{Supplementary Method Descriptions}

\subsection{Detailed Description of \textsf{LPD}}\label{apped: UD example}

\subsubsection{Algorithm Description}
Algorithm \ref{Alg UD} elaborates the implementation of \textsf{LPD}. At each timestamp $t$, \textsf{LPD} consists of two sub mechanisms, $\mathcal{M}_{t,1}$ for private dissimilarity calculation and $\mathcal{M}_{t,2}$ for private strategies determination and publication users allocation. Initially, the whole population $U$ is taken as available users $U_A$, and the released frequency histogram is set as $\mathbf{r}_0=\langle 0, \ldots, 0 \rangle^d$ (Line 1). Then, at each timestamp $t$, \textsf{LPD} performs $\mathcal{M}_{t,1}$ and $\mathcal{M}_{t,2}$ in sequence.




\textbf{Sub Mechanism $\mathcal{M}_{t,1}$} (Lines 3-6). $\mathcal{M}_{t,1}$ aims to privately calculate the dissimilarity measure $dis$ with the dissimilarity users $U_{t,1}$. Since half the population $\lfloor N/2 \rfloor$ is divided for $\mathcal{M}_{1}$ in a window of size $w$, $\mathcal{M}_{t,1}$ randomly samples a subset $U_{t,1}$ with $\lfloor N/(2w) \rfloor$ users from $U_A$ at timestamp $t$. These sampled users are removed temporarily from $U_A$, i.e., $U_A\leftarrow U_A \setminus U_{t,1}$ (Line 3) to ensure they only participate once within a window. Then the server requests all users in $U_{t,1}$ to upload locally perturbed data with the entire privacy budget $\epsilon$, which is stored in a database $\overline{D}_{t,1}$ (Line 4). Using the LDP frequency oracle on $\overline{D}_{t,1}$, an unbiased estimation $\mathbf{\overline{c}}_{t,1}\leftarrow$ FO($\overline{D}_{t,1}$, $\epsilon$) of the true frequency count $\mathbf{{c}}_{t,1}$ can be obtained from $\overline{D}_{t,1}$ (Line 5). Based on Theorem \ref{theorem: dissimilarity}, the server computes an LDP but unbiased estimation of dissimilarity as $dis=\frac{1}{d}\sum\nolimits_{k=1}^d(\overline{\mathbf{c}}_{t,1}[k]-\mathbf{r}_{t-1}[k])^2-\frac{1}{d}\sum\nolimits_{k=1}^d\text{Var}(\mathbf{\overline{c}}_{t,1}[k])$, where $\frac{1}{d}\sum\nolimits_{k=1}^d\text{Var}(\mathbf{\overline{c}}_{t,1}[k])$ can be calculated based on $\epsilon$ and $|U_{t,1}|$. We replace $\mathbf{r}_{l}$ with $\mathbf{r}_{t-1}$ since that the approximation $\mathbf{r}_{t}=\mathbf{r}_{t-1}$ is always adopted until a publication occurs. The calculated dissimilarity $dis$ is then passed to $\mathcal{M}_{t,2}$.

\textbf{Sub Mechanism $\mathcal{M}_{t,2}$} (Lines 7-19). The sever compares the dissimilarity $dis$ returned by $\mathcal{M}_{t,1}$ with the possible publication error $err$ to choose publication or approximation. The possible publication error $err$ is determined by the privacy budget $\epsilon$ and number of potential publication users, i.e., $|U_{t,2}|$.
So, to estimate the potential publication error, the server first calculates the number of remaining publication users (indicated as $N_{rm}$) by subtracting the number of used publication users in the current window $\sum\nolimits_{k=t-w+1}^{t-1}|U_{k,2}|$ ($U_{k,2}=\varnothing$ for $k \leq 0$) from the number of total publication users $\lfloor N/2 \rfloor$ (Line 7). After that, the number of potential publication users is set as $|U_{t,2}|=N_{rm}/2$ to make it exponentially decreasing within a window (Line 8). Then, the potential publication error $err$ at the current timestamp equals to $V(\epsilon, N_{rm}/2)$, can be calculated by Eq.~(\ref{var}). Next, the server simply compares $dis$ with $err$ and chooses a strategy with less error: (1) if $dis > err$, it implies that the approximation error is larger than potential publication error, the server chooses publication to reduce the error. In particular, it samples a fresh group of publication users $U_{t,2}$ and then requests them to upload perturbed data, from which, the server can obtain an unbiased estimation $\mathbf{\overline{c}}_{t,2}$ as the output. Note that, with exponential decaying, $|U_{t,2}|$ may drop quickly and lead to large sampling error due to insufficient users. Therefore, we set a minimal threshold $u_\text{min}$ and $\mathcal{M}_{t,2}$ directly chooses the approximation strategy if $|U_{t,2}|$ is too small. If the publication error $err$ is larger than $dis$, the server then decides to approximate with the last release $\mathbf{{r}}_{t-1}$ with no actual publication, i.e., setting $U_{t,2}= \varnothing$ (Line~16). Finally, at the end of the timestamp $t$, the server recycles both the used dissimilarity and publication users at $t-w+1$ as available for the next timestamp, i.e., $U_{A}=U_A \cup U_{t-w+1,1} \cup U_{t-w+1,2}$. The recycling process ensures each user can contributes again after $w$ timestamps but only contribute once in a window of size $w$ to satisfy $w$-event LDP. 

\subsection{Detailed Description of \textsf{LPA}}\label{apped: UA example}

\subsubsection{Algorithm Description}
The details of \textsf{LPA} are shown in Algorithm \ref{Alg UA}. $\mathcal{M}_{t,1}$ for private dissimilarity calculation is identical to that in \textsf{LPD}. We mainly explain the details of $\mathcal{M}_{t,2}$ for publication strategy decision and participant users allocation.
At each timestamp $t$, $\mathcal{M}_{t,2}$ first allocates a fixed number of publication users $|U_{t,2}|=\lfloor N/(2 w) \rfloor$.
According to the basic idea, there will be two cases at $t$: (i) if the approximation strategy was selected (i.e., the publications were skipped) in the previous timestamps, the corresponding publication users will be added to $U_{t,2}$; (2) if the publication strategy is selected in the previous timestamp, $U_{t,2}$ may be nullified and approximation strategy may be selected at current time $t$.
Therefore, $\mathcal{M}_{t,2}$ has to identify which case the current timestamp belongs to. Denote $U_{l,2}$ as the population of publication users at the timestamp $l$ where last publication $\mathbf{r}_l$ occurred.
Since $\lfloor N/(2w) \rfloor$ users are uniformly allocated at each timestamp in $\mathcal{M}_{t,2}$, there will be $t_N=|U_{l,2}|/(\lfloor N/(2w) \rfloor)-1$ timestamps after $l$ where the user allocation must be nullified (Line 4). Therefore, if $t-l \leq t_N$, the user allocation at current time $t$ is nullified and $\mathcal{M}_{t,2}$ outputs with the last release $\mathbf{r}_t=\mathbf{r}_{t-1}$.
Otherwise, Case (i) holds, $\mathcal{M}_{t,2}$ absorbs the users of previously skipped publications and decides whether to freshly publish with these absorbed users (Lines 11-15) or continues to approximate with the last release (Lines 16-17).
In particular, it first computes the number of skipped publications between the last publication timestamp $l$ and $t$, which corresponds to $t_A=t-(l+t_N)$. To satisfy $w$-event LDP, at most $w-1$ timestamps can be absorbed. Therefore, the remaining population is $N_{rm}=\lfloor N/(2w) \rfloor \cdot $min$(t_A,w)$. Using these remaining users as the publication users $U_{t,2}=N_{rm}$, the potential publication error $err$ can be estimated.
If the dissimilarity $dis>err$, the server requests a random sampled userset $U_{t,2}$ with the population $|U_{t,2}|$ to respond with LDP and outputs an unbiased $\overline{\mathbf{c}}_{t,2}$; otherwise, $\mathcal{M}_{t,2}$ chooses the approximation strategy and simply outputs $\mathbf{r}_t= \mathbf{r}_{t-1}$ without publication user $|U_{t,2}|$ allocation. Finally, it also recycles both the dissimilarity and publication users used at timestamp $t-w+1$.

\end{document}